\documentclass[options]{IEEEtran}
\usepackage{xcolor}
\usepackage{etoolbox}
\makeatletter
\@ifundefined{color@begingroup}%
  {\let\color@begingroup\relax
   \let\color@endgroup\relax}{}%
\def\fix@ieeecolor@hbox#1{%
  \hbox{\color@begingroup#1\color@endgroup}}
\patchcmd\@makecaption{\hbox}{\fix@ieeecolor@hbox}{}{\FAILED}
\patchcmd\@makecaption{\hbox}{\fix@ieeecolor@hbox}{}{\FAILED}

\usepackage{generic}
\usepackage{cite}

\usepackage{amsmath,amssymb,amsfonts}
\usepackage{algorithm}
\usepackage{algpseudocode}
\usepackage{graphicx}
\usepackage[hidelinks]{hyperref}
\usepackage{textcomp}
\usepackage{wrapfig,graphicx}
\usepackage{hyperref}
\usepackage{subcaption}
\graphicspath{{.}}

\usepackage{amsthm}
\newtheorem{thm}{Theorem}
\newtheorem{definition}{Definition}
\newtheorem{coro}{Corollary}

\newtheorem{lemma}{Lemma}
\newtheorem{problem}{\textbf{Problem}}
\newtheorem{remark}{Remark}
\newtheorem{assumption}{Assumption}
\newtheorem{mechanism}{Mechanism}
\usepackage{etoolbox}
\usepackage{float}
\usepackage{tabularx}
\usepackage{mathtools}
\usepackage{bm}
\mathtoolsset{showonlyrefs}
\usepackage{xcolor}

\newcommand{\baike}[1]{\textcolor{black}{#1}}

\title{Scalable Distributed Reproduction Numbers of Network Epidemics with Differential Privacy}
\author{Bo Chen$^*$\thanks{$^*$ indicates equal contribution.}, Baike She$^*$, Calvin Hawkins, Philip E. Par\'e, Matthew T. Hale
\thanks{B. Chen, B. She, and M.T. Hale are with the School of Electrical and Computer Engineering,
Georgia Institute of Technology, Atlanta, GA, USA. Emails: \texttt{\{bchen351,bshe6,matthale\}@gatech.edu}.}
\thanks{Work by C. Hawkins was performed while he was a PhD student in the 
Department of Mechanical and Aerospace Engineering
at the University of Florida,
Gainesville, FL, USA. Email: \texttt{cal.hawk.ch@gmail.com}.} 
\thanks{B. Chen, B. She, C. Hawkins, and M. Hale were supported by 
NSF under CAREER grant 2422260,
AFOSR under grant FA9550-19-1-0169,
ONR under grant N00014-21-1-2502,
and 
DARPA under grant HR00112220038. P. E. Par\'e was partially supported by NSF-ECCS-\#2238388.}
}
\begin{document}
\maketitle
\newcommand{\realspace}{\mathbb{R}}
\newcommand{\nonnegativeRealSpace}{\realspace_{\geq 0}}
\newcommand{\bigO}[1]{O\left(#1\right)}
\newcommand{\bigOmega}[1]{\Omega\left(#1\right)}
\newcommand{\privateReproNum}{\Tilde{R}}
\newcommand{\reproNum}{R}
\newcommand{\expectation}[1]{\mathbb{E}\left[#1\right]}
\newcommand{\variance}[1]{\text{Var}\left[#1\right]}
\newcommand{\prob}[1]{\mathbb{P}\left[#1\right]}

\newcommand{\adjacent}{\sim}
\newcommand{\localRandomizer}{\mathcal{R}}
\newcommand{\database}{X}

\newcommand{\centralDP}{\mathcal{M}}

\newcommand{\vectElement}{\Bar{R}_{i,\chi_r}^t}
\newcommand{\vect}{\bm{\zeta}}

\newcommand{\domain}{\mathbb{D}}
\newcommand{\privateVect}{\Tilde{\vect}}
\newcommand{\privateVectElement}{\Tilde{R}_{i,\chi_r}^t}
\newcommand{\vectLowerBound}{\underline{\vect}}
\newcommand{\vectUpperBound}{\bar{\vect}}

\newcommand{\vectElementLowerBound}{l}
\newcommand{\vectElementUpperBound}{u}

\newcommand{\permutation}{\pi}
\newcommand{\diag}{\textnormal{diag}}
\newcommand{\localRandPrivacyLevel}{\epsilon_0}
\newcommand{\adjacency}{k}

\newcommand{\offsetVectElement}{c}
\newcommand{\offsetVect}{\bm{\offsetVectElement}}

\newcommand{\bo}[1]{\textcolor{red}{#1}}

\newcommand{\vectNumInShuffler}{n}
\newcommand{\setVectShuffler}{\underline{\vectNumInShuffler}}

\newcommand{\locAuthNum}{n}
\newcommand{\regionNum}{m}
\newcommand{\setLocAuthNum}{\underline{\locAuthNum}}
\newcommand{\setRegionNum}{\underline{\regionNum}}

\begin{abstract}
Reproduction numbers are widely used for the estimation and prediction of epidemic spreading processes over networks.
However, \textcolor{black}{conventional reproduction numbers of an overall} network do not 
indicate \emph{where} an epidemic is spreading. 
Therefore, we propose a novel notion of \textcolor{black}{local} distributed reproduction numbers to capture the spreading behaviors of each node in a network.
We first show how to compute them and then use them
to derive new conditions under which an outbreak can occur.
These conditions are then used to derive new conditions for the existence, uniqueness, and stability of equilibrium states of
the underlying epidemic model.~\textcolor{black}{Building upon these local distributed reproduction numbers, we define  \emph{cluster} distributed reproduction numbers to model the spread between clusters composed of nodes. Furthermore, we demonstrate that the local distributed reproduction numbers can be aggregated into cluster distributed reproduction numbers at different scales.}
However, both local and cluster distributed reproduction numbers can reveal the frequency of interactions between nodes in a network, which raises 
privacy concerns. \textcolor{black}{Thus, we next develop a privacy framework that implements a differential privacy mechanism to provably protect  \baike{the frequency of interactions between nodes}  when computing distributed reproduction numbers.}
Numerical experiments show that, even under differential privacy,
\baike{the} distributed reproduction numbers provide accurate
estimates of \baike{the} epidemic spread while also providing more 
insights than conventional reproduction numbers. 
\end{abstract}

\section{Introduction}
\label{sec:introduction}
Reproduction numbers are critical metrics in infectious disease epidemiology \cite{van2017reproduction}, as they 
are easily understood by policymakers and the public. These numbers also help design control interventions \baike{during} pandemics \cite{soltesz2020effect,she2024framework}. There are two common types of reproduction number: the \textit{basic reproduction number}, which represents the number of secondary infections caused by one infected case in a fully susceptible population, and the \textit{effective reproduction number}, which reflects the number of secondary infections caused by one infected case in a mixed susceptible and infected population \cite{van2017reproduction}. The critical threshold for the reproduction number is~one, as epidemic behavior changes significantly when the reproduction number is above or below this value.

Reproduction numbers have been used to model and design  epidemic mitigation strategies~\cite{she2021network,pascal2022nonsmooth, smith2021convex}. Recent studies establish threshold conditions to analyze both transient and steady-state behaviors in disease spreading models, based on whether the reproduction number is above or below~one~\cite{mei2017epidemics_review,pare2020modeling_review,zino2021analysis,nowzari2016epidemics}. This concept has been extended from classic $SIS$ models \cite{van2011n} to more complex \textcolor{black}{network} models, such as network bi-virus systems~\cite{bivirus} and coupled network models~\cite{she2021network}. 

Nonetheless, viral spread often exhibits high heterogeneity across different sub-populations, making it difficult to use the reproduction number of an entire spreading network 
to quantify the behavior of individual entities. For example, the spread of COVID-19 varied significantly across regions in the United States, with differences in infection growth and peak dates~\cite{ihme2021modeling}.
As a result, it is challenging to 
use a single network-level reproduction number to make inferences about specific communities, counties, states, and/or countries in highly heterogeneous spreading networks~\cite{arino2003multi_city_SIS}. 

To address this issue, we 
propose \textcolor{black}{the} notion of \emph{distributed reproduction numbers} 
to capture epidemic spreading at various scales within a network. Then, we develop threshold conditions based on these distributed reproduction numbers
under which outbreaks can occur
in the classic network $SIS$ and $SIR$ models.
Typically, infectious disease spreading networks are modeled with edges that represent transmission rates between nodes. However, absolute transmission rates alone do not directly capture the dynamics of spreading behavior between nodes, such as whether infection cases are increasing or decreasing. 
\textcolor{black}{Furthermore, it is unclear how to aggregate the transmission rates from the node level to the cluster level directly in a network spreading model.}
Alternatively, we can represent the spreading network with edges corresponding to distributed reproduction numbers across the network, using the threshold value to indicate disease spread within and between nodes. 
\textcolor{black}{Additionally, we can aggregate the distributed reproduction numbers at the node level into distributed reproduction numbers at different cluster levels. These clusters are composed of multiple nodes, and their cluster-level distributed reproduction numbers allow us to model the spread at different scales, providing a more comprehensive 
\baike{set of tools for}
understanding the spreading dynamics.}
Thus, similar to the \baike{network-level} reproduction number, distributed reproduction numbers \textcolor{black}{at different scales} provide simple but informative threshold \baike{metrics} that effectively convey the severity of the spread across the network.

Constructing spreading network models
involves 
privacy-sensitive spatio-temporal data related to human activities, such as contact tracing~\cite{eames2010assessing}, traffic flow~\cite{le2022high}, and mobile data~\cite{balcan2009multiscale}.
It is well-known that revealing even aggregate statistics of such information can compromise \baike{the privacy of individuals~\cite{imola2021locally,Karwa2014Private,Day2016Publishing,ZHANG2021Differentially,chen2021edge},} 
which makes it undesirable to share distributed reproduction numbers exactly.
%
Accordingly, we propose \textcolor{black}{a privacy framework} that uses \emph{differential privacy}\cite{dwork2014algorithmic} to provide formal privacy guarantees for the sensitive data
that is used to compute
distributed reproduction numbers. Differential privacy offers strong, formal protections for sensitive network data and allows for post-processing without harming its protections~\cite{dwork2014algorithmic}. 
It has been successfully used to privatize a range of dynamical and control 
systems~\cite{cortes2016differential,hale19,hawkins20,hawkins23}, 
and we therefore seek to bring these same benefits
to this \baike{domain}. To do so, 
rather than sharing the exact values of the distributed reproduction numbers for final analysis, we perturb them by adding properly calibrated noise before sharing them, 
thereby ensuring that this process provably provides differential privacy. 

To summarize, our contributions are: 
 \begin{itemize}
   \item We introduce a new group of \textcolor{black}{local} distributed reproduction numbers for spreading networks (Definition~\ref{Def:NRN})
    \item We use the \textcolor{black}{local} distributed reproduction number to analyze the transient and steady-state behaviors of network spreading processes 
    (Theorem~\ref{thm:connection} and Corollary~\ref{Coro_ditri_effe_repro_matrix})
    \item We 
    derive \textcolor{black}{cluster} distributed reproduction numbers at various scales in the network
    to analyze epidemic spread at different resolutions
    (Definition~\ref{Def:Generalized_Rt}, Theorem~\ref{Thm:Gen_Reprod}, and Corollary~\ref{coro:reproduction_subpop})
    \item We develop a privacy framework to compute privatized \textcolor{black}{local} and \textcolor{black}{cluster} distributed reproduction numbers 
    (\textcolor{black}{Algorithm~\ref{AL_1}} and Mechanism~\ref{mech:bounded_Gaussian_mechanism})
 \item We quantify the accuracy of privatized distributed reproduction numbers (Theorem~\ref{thm:differential_privacy_accuracy})
 \item We use real-world epidemic data to demonstrate all of these developments (Section~\ref{sec_Simulation})
 \end{itemize}
 
The rest of the paper is organized as follows. We introduce background and problem statements in Section~\ref{Sec:Background_and_Problem_Formulation}. 
In Section~\ref{Sec_The_Distributed_Reproduction_Number}, we introduce and study distributed reproduction numbers.
In Section~\ref{Privacy_Mechanism_For_The_Distributed_Reproduction_Numbers}, we design the differential privacy framework
that \baike{for the} distributed reproduction numbers.
Section~\ref{sec_Simulation} illustrates the results by analyzing real-world network spreading scenarios. 
Section~\ref{sec_Conclusion} concludes. 

In our previous work~\cite{she2023distributed}, we defined distributed reproduction numbers at the entity level only.
This paper differs by defining more general notions of distributed reproduction numbers at different resolutions.
Additionally, we implemented differential privacy when computing basic reproduction numbers in~\cite{chen2023differentially}. 
In this work, we differ by implementing privacy for distributed reproduction numbers and by validating our privacy
results on real epidemic data. 

\subsection*{Notation}
We use $\mathbb{R}$ to denote the real numbers, $\mathbb{R}_{\geq 0}$ to denote the non-negative reals, and $\mathbb{R}_{> 0}$ to denote the positive reals. We use $\mathbb{N}_{>0}$ to denote the positive integers, $\mathbb{Z}$ to denote the integers, and
$\mathbb{Z}_{\geq k}$ to denote all integers greater than or equal to~$k \in \mathbb{Z}$.
For a random variable $X$, $\mathbb{E}[X]$ denotes its expectation and $\text{Var}[X]$ denotes its variance. Let $\mathbf{1}_{T}(\cdot)$ denote the indicator function of set $T$.
We use $\underline{n}$ to denote the index set $\{1,2, \dots, n\}$ for $n\in\mathbb{N}_{>0}$.
For a real square matrix $M:=[m_{ij}]\in\mathbb{R}^{n\times n}$ with
$i,j\in\underline{n}$, 
we use $\rho(M)$ to denote
its spectral radius. For any two matrices $A:=[a_{ij}],C:=[c_{ij}]\in\mathbb{R}^{n\times n}$,
we write $A\geq C$ if $a_{ij}\geq c_{ij}$, 
$A> C$ if $a_{ij}\geq c_{ij}$ and $A\neq C$, 
and $A\gg C$ if $a_{ij}>c_{ij}$, for all $i,j\in \underline{n}$. These comparison notations between matrices apply to vectors as well.
For a vector~$v \in \mathbb{R}^n$, we write~$\textnormal{diag}(v)\in \mathbb{R}^{n\times n}$
to denote the diagonal matrix whose~$i^{th}$ diagonal
entry is~$v_i$ for each~$i \in \underline{n}$. We use $||\cdot||_F$ to denote the Frobenius norm of a matrix.

\baike{We use} $G=(V,E,W)$ to denote a directed, strongly connected, and weighted graph with node set $V$, edge set $E$, and weighted adjacency matrix $W:=[w_{ij}]\in\mathbb{R}_{\geq 0}^{n\times n}$, 
where $w_{ij} \geq 0$ denotes the $i^{th}j^{th}$ entry of the weighted adjacency matrix $W$. Let $|\cdot|$ denote the cardinality of a set. 
For a given matrix~$W$, 
we use $n_W=|\{w_{ij}>0 : i, j \in \underline{n}\}|$ 
to denote the number of positive entries in $W$. We use $\mathcal{G}_n$ to denote \baike{the} set of all possible directed, strongly connected, weighted graphs $G$ on $n$ nodes.

\section{Background and Problem Formulation}
\label{Sec:Background_and_Problem_Formulation}
Now we introduce background on network epidemic models and differential privacy, then give problem statements. 

\subsection{Network Epidemic 
Models}\label{sec:prelim_epidemic}
We consider network susceptible-infected-susceptible ($SIS$) and susceptible-infected-recovered ($SIR$) models to study disease spread over connected sub-populations at various levels, whether globally, nationally, regionally, or within a community. 
The entire population consists of a set of $n$ entities, where $n \in \mathbb{Z}_{\geq 2}$. Each entity can represent either an individual or a group of individuals, ranging from a small community, such as a neighborhood or social club, to a large population group, including a \textcolor{black}{county, state, or country.} 

Let $G=(V,E,B)\in\mathcal{G}_n$ denote an epidemic spreading network that models an epidemic spreading process over these $n$ connected entities. 
Let $V$ and $E$ denote the entities and the transmission channels between them, respectively.
We use $s(t), x(t), r(t)\in[0,1]^n$ to represent the susceptible, infected, and recovered state vectors, respectively. 
That is, for all~$i \in \underline{n}$, the value of~$s_i(t)\in [0, 1]$ is the susceptible portion of the population of  the~$i^{th}$ entity at time~$t$, 
the value of~$x_i(t) \in [0, 1]$ is the size of the infected proportion of the population of entity~$i$ at time~$t$, and the value of~$r_i(t) \in [0, 1]$ is the size of the recovered proportion of the population of entity~$i$ at time~$t$. 
We use $B:=[\beta_{ij}]\in \mathbb{R}^{n\times n}_{\geq 0}$, with $\beta_{ij}\in[0,1]$ for all~$i, j \in \underline{n}$, to denote the transmission matrix and $\Gamma=\text{diag}([\gamma_1,\gamma_2,\dots,\gamma_n])\in\mathbb{R}^{n\times n}$, with $\gamma_i\in (0,1]$ for all~$i \in \underline{n}$, to denote the recovery matrix.
Further, we use $B$ as the adjacency matrix of the spreading graph $G$.
Thus, $\beta_{ij}$ captures the transmission process from the $j^{th}$ entity to the $i^{th}$ entity, while $\gamma_i$ captures the recovery rate of entity $i$, for all $i,j\in\underline{n}$. 

The  network $SIS$ and  $SIR$ \textcolor{black}{dynamics are}
\begin{equation}
    \textnormal{SIS: }
    \begin{cases}\label{eq:SIS}
       \dot{s}(t) &= -\diag(s(t))Bx(t) + \Gamma x(t), \\
        \dot{x}(t)  &= \diag(s(t))Bx(t)-\Gamma x(t), 
    \end{cases}
\end{equation}
and 
\begin{equation}
\textnormal{SIR: }
    \begin{cases}\label{eq:SIR}
     \dot{s}(t)  &= -\diag(s(t))Bx(t), \\
        \dot{x}(t)  &= \diag(s(t))Bx(t) - \Gamma x(t),\\
        \dot{r}(t)  &= \Gamma x(t),
    \end{cases}
\end{equation}
respectively. For all $i\in \underline{n}$, we have that
$s_i(t), x_i(t),  r_i(t)\in[0,1]$, and $s_i(t)+x_i(t)+r_i(t)=1$~\cite{pare2020modeling}.

\begin{assumption}
\label{assum_graph_strongly_connected}
The graph~$G$ is strongly connected. 
\end{assumption}

Inspired by using the \textit{next generation matrix} to derive the basic reproduction number 
for network $SIS$ and $SIR$ spreading models~\cite{diekmann2010construction}, 
researchers have defined 
$W = \Gamma^{-1}B$ as  the next generation matrix 
to characterize the global behavior of network $SIS$ and $SIR$ models in~\eqref{eq:SIS} and~\eqref{eq:SIR}~\cite{mei2017epidemics_review,pare2020modeling_review,she2021peak}, respectively. 
These threshold conditions derived from
$W = \Gamma^{-1}B$
are defined in terms of \textit{the reproduction number of networks} (namely the network-level reproduction number).
\begin{definition}
(Network-Level Reproduction Number)
\label{Def:Repro}
Given an epidemic
spreading network $G=(V,E,B)\in\mathcal{G}_n$ and the next generation matrix $W = \Gamma^{-1}B$, for the network $SIS$ and $SIR$ models, the basic reproduction number is defined as $R^0= \rho (W)$ and the effective reproduction number is defined as $R^t = \rho (\diag(s(t)) W)$.\hfill $\lozenge$
\end{definition}

Definition~\ref{Def:Repro} implies that we can compute the network-level reproduction number when having access to the transmission matrix and the recovery matrix. 
However, the network-level reproduction number may fail to capture the spreading behavior of individual entities within a network. 
To illustrate this point, Figure~\ref{fig_inf} (Top) presents the infected proportion of each community in a network $SIR$ model over ten communities, with~$x_i$ showing the infected proportion
of community~$i$. 
The dashed line indicates the weighted sum of the infected cases
$w^{\top}_{t_p}x$, where $w_{t_p}\gg0$ is the normalized left eigenvector corresponding to the spectral abscissa of the matrix $\diag(s(t_p))B-\Gamma$ with $t_p$ being the peak infection time where $R^t=1$ at time $t_p$~\cite[Definition~2]{she2021peak}.

Figure~\ref{fig_inf} (Bottom) shows the corresponding network-level effective reproduction number, i.e., $R^t$. The effective reproduction number $R^t>1$ 
until roughly timestep~$25$. However, the infected proportions of most communities, including 
communities~$1$,~$2$, $3$, $4$, and $9$ have already significantly decreased by timestep~$25$. 
Therefore,
if we aim to analyze a single entity or a subnetwork of connected entities, 
then the network-level reproduction numbers $R^0$ and/or $R^t$ may fail to capture the spreading behavior
at that level of granularity,
since thresholds for each type of reproduction number 
only characterize overall network-level spreading, which can 
be quite different from local spreading. 

\begin{figure}
  \begin{center}
    \includegraphics[ trim = 0cm 0cm 0cm 0cm, clip, width=\columnwidth]{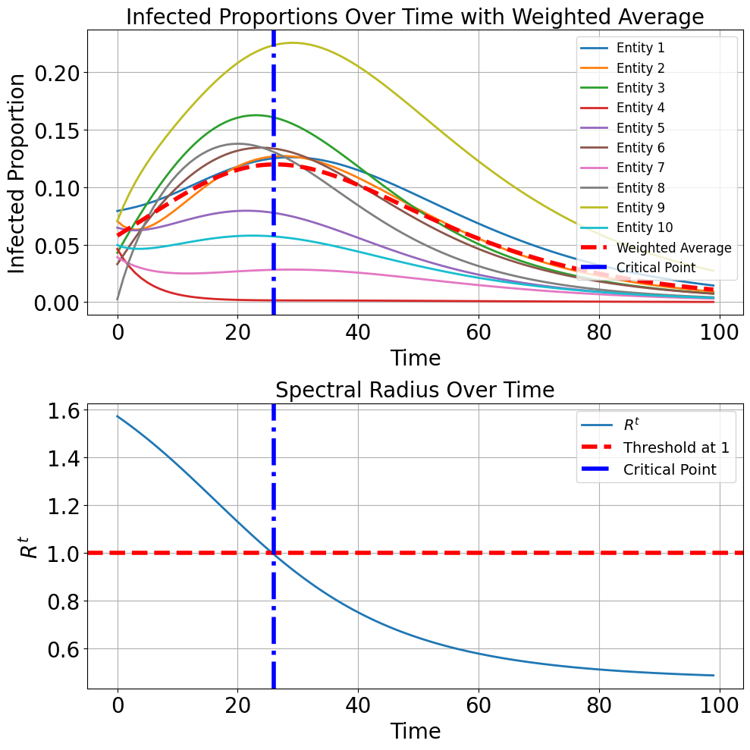}
  \end{center}
  \caption{\baike{(Top) Infected proportion of each node in a network $SIR$ model over ten entities. The red dashed line indicates the weighted sum of the infected cases
$w^{\top}_{t_p}x$, where $w_{t_p}\in\mathbb{R}^n_{>>0}$ is the normalized left eigenvector corresponding to the spectral abscissa of the matrix $\diag(s(t_p))B-\Gamma$ with $t_p$ being the peak infection time where $R^t=1$ at time $t_p$~\cite[Definition~2]{she2021peak}. (Bottom) The corresponding network-level effective reproduction number $R^t$ of the spreading dynamics. The effective reproduction number $R^t>1$ 
until roughly timestep~$25$. However, the infected proportions of most entities, including entities~$1$, $2$, $3$, $4$, and $9$, 
have already significantly decreased
by timestep~$25$ in Figure~\ref{fig_inf}.}
}
  \label{fig_inf}
\end{figure}


\subsection{Problem Statements Part 1: Reproduction Numbers}

\textcolor{black}{Motivated by this discussion of the reproduction numbers
of networks, we formulate the following problems.}

\begin{problem}
\label{prob:1}
Consider an epidemic
spreading network \baike{over} $G=(V,E,B)\in\mathcal{G}_n$ as defined in~\eqref{eq:SIS} and~\eqref{eq:SIR}. For each edge $e\in E$ (including self loops), develop a notion of \emph{\textcolor{black}{local} distributed reproduction number} that characterizes the spreading for the transmission interaction 
between the two entities that are connected by the edge~$e$. 
\hfill $\lozenge$
\end{problem}

The goal of \textcolor{black}{using local}
distributed reproduction numbers is to provide a method to study spreading behaviors in a decentralized manner. 
%
\textcolor{black}{Simultaneously, as indicated in Definition~\ref{Def:Repro}, the network-level reproduction number provides a global understanding of disease spread across the entire network, summarizing the spread of an epidemic in a single scalar value. Next, we seek to demonstrate that the local distributed reproduction numbers can also provide insights into the overall network spreading process.
}


\begin{problem} \label{prob:2}
Given an epidemic
spreading network \baike{over} $G=(V,E,B)\in\mathcal{G}_n$ as defined in~\eqref{eq:SIS} and~\eqref{eq:SIR} and the corresponding distributed reproduction numbers for all edges in $E$, \baike{study the 
spreading behavior of the entire network through the local distributed reproduction numbers.}\hfill $\lozenge$
\end{problem}



\textcolor{black}{The \textcolor{black}{local} distributed reproduction number captures the spread between entities at the highest resolution within a network, while the network-level reproduction number describes the overall spread across the entire network, i.e., at the lowest resolution. In real-world scenarios, it is often necessary to have intermediate-level information to model and characterize the spread between grouped populations, \textcolor{black}{i.e., clusters}. 
For example, when an entity represents a household in a spreading network of a large city, obtaining intermediate-level transmission \baike{knowledge}
{\textemdash} 
such as between different regions based on administrative divisions
{\textemdash}
can be crucial for analyzing and informing policy decisions during outbreaks. This intermediate information is essential, as it provides insights into the spread at a scale that lies between individual entities and the entire \baike{population}.}

\textcolor{black}{Consider partitioning the $n$ entities in the spreading network~$G$ into $m$ grouped entities, defined as \textit{clusters}. 
For~${q\in\underline{m}}$, 
we define the $q^{th}$ cluster as the collection of 
entities whose indices are contained in the set
$\chi_q $ ($\chi_q \neq \varnothing$), where 
$|\chi_q|$ is the number of the entities in cluster $q$. We use 
$X=\{\chi_1, \dots, \chi_m\}$ to represent all of the clusters that comprise the network,
with 
\begin{equation} \label{eq:chidef}
\bigcup_{q \in \underline{m}} \chi_q = \underline{n} \quad
\textnormal{ and } \quad
\chi_{q_1} \cap \chi_{q_2} = \varnothing
\end{equation}
for all distinct sets~$\chi_{q_1}, \chi_{q_2} \in X$.
We are next interested in modeling cluster-level spreading using the local distributed reproduction numbers.}

\begin{problem} \label{prob:3}
Let an epidemic spreading network \baike{over} $G=(V,E,B) \in \mathcal{G}_n$ as defined 
in~\eqref{eq:SIS} and~\eqref{eq:SIR} be given, along with 
the corresponding \textcolor{black}{local} distributed reproduction numbers for all edges in $E$.
Let $m$ denote the number of clusters from~\eqref{eq:chidef}, 
each one of which contains one or more entities in $G$. Then, find the \textcolor{black}{cluster}
distributed reproduction numbers that capture the spreading behavior between these clusters in terms
of the \textcolor{black}{local} distributed reproduction numbers. \hfill $\lozenge$
\end{problem}

Problem~\ref{prob:3} formulates the scalable properties of distributed reproduction numbers: 
it should be possible to use 
the \textcolor{black}{local} distributed reproduction numbers at a finer level 
to compute distributed reproduction numbers at a coarser cluster level. 
Additionally, this relationship also enables researchers/analysts 
to analyze the spread between clusters by post-processing these \textcolor{black}{local} distributed reproduction numbers without needing access to the sensitive raw data that was used to compute them. 
In the following subsection, we explore how this property relates to privacy-preserving analyses.

\subsection{Problem Statements Part 2: Privacy}
\label{Sec_Diff_Pri_Prob_Form}

Thanks to the decentralized property of distributed reproduction numbers, we can analyze the spreading network at different scales. 
\baike{However, sharing the \textcolor{black}{distributed reproduction numbers} may raise privacy concerns. 
While it may seem surprising that scalar-valued queries (such as cluster distributed reproduction numbers in our case) can leak information about local data, this principle has been firmly established in the graph privacy literature~\cite{imola2021locally,Karwa2014Private,Day2016Publishing,ZHANG2021Differentially,chen2021edge}. It has also led to the development of privacy-preserving methods for computing a wide array of graph properties, including spectra of graph adjacency \baike{matrices~\cite{Sealfon2016shortest}, properties of 
graph Laplacians~\cite{chen2021edge}, counts of subgraphs~\cite{imola2021locally}, degree sequences~\cite{Day2016Publishing}, and others~\cite{Hay2009accurate,blocki2013differentially}.} 
Each of these works has anticipated a type of graph analysis and applied privacy to it, and we do the same in this work.}

\textcolor{black}{Similar to the transmission rates between any pair of nodes in the network, the local distributed reproduction number also reveals the frequency of interactions between nodes in a network. For example, to construct a disease contact network at a French primary school~\cite{stehle2011high}, based on the accumulated contact time between any pair of students on campus during one day, all the students' guardians had to sign a privacy release statement. This example highlights that sharing the frequency of interactions (e.g., modeled by transmission rates and/or local distributed reproduction numbers) may violate the privacy of entities in the network. Therefore,
to mitigate the privacy risks associated with the distributed reproduction numbers, we apply differential privacy.} 

\baike{Consequently,} we introduce our privacy framework based on the communication network depicted in Figure~\ref{fig:network_structure}. This figure illustrates the reporting of local distributed reproduction numbers from entities to, for example, \textcolor{black}{public health officials, elected leaders, and other decision-makers}. 
\begin{definition}\label{Def_Authorities}
\textcolor{black}{ 
Consider a spreading network of $n\in\mathbb{N}_{>0}$ nodes that represents $n$ entities. For each entity $i\in\underline{n}$, we define its \textbf{local authority} as a local organization that has access to the \textcolor{black}{local} distributed reproduction numbers of entity $i$, as shown in Figure~\ref{fig:network_structure}. 
We then define the \textbf{central authority} as an organization \textcolor{black}{of the overall spreading network} to which policy-makers have access. 
\textcolor{black}{To retrieve the cluster distributed reproduction numbers between $m \leq n$ clusters, the local authority of each entity computes its local distributed reproduction numbers corresponding to other entities within the network. Then, the local authorities within the same cluster share their local distributed reproduction numbers with the unique \textbf{shuffler} of the cluster, as shown in Figure~\ref{fig:network_structure}.} After applying a shuffling mechanism to the local distributed reproduction numbers, the shuffler sends the output to the aggregator of the same cluster. That cluster aggregator then generates  the cluster distributed reproduction number. We define the combination of a shuffler and an aggregator of the same cluster as the \textbf{central aggregator} of the cluster.}
\textcolor{black}{At last, all $m$ central aggregators send their cluster distributed reproduction numbers to a data center, which subsequently shares these reproduction numbers with the central authority of the network for further analysis. }
\hfill $\lozenge$
\end{definition}

We note that local authorities and central aggregators are not controlled by the central authority. Therefore, the central authority has permission only to read the outputs from the central aggregators but does not have permission to inspect the privacy mechanisms implemented by them.
Further, to simplify the framework, we consider that the central aggregators are managed by their corresponding clusters in this work. However, the central aggregators can also be separated from their clusters, adding an additional layer of security.

\begin{figure}
    \centering \includegraphics[width=1\columnwidth]{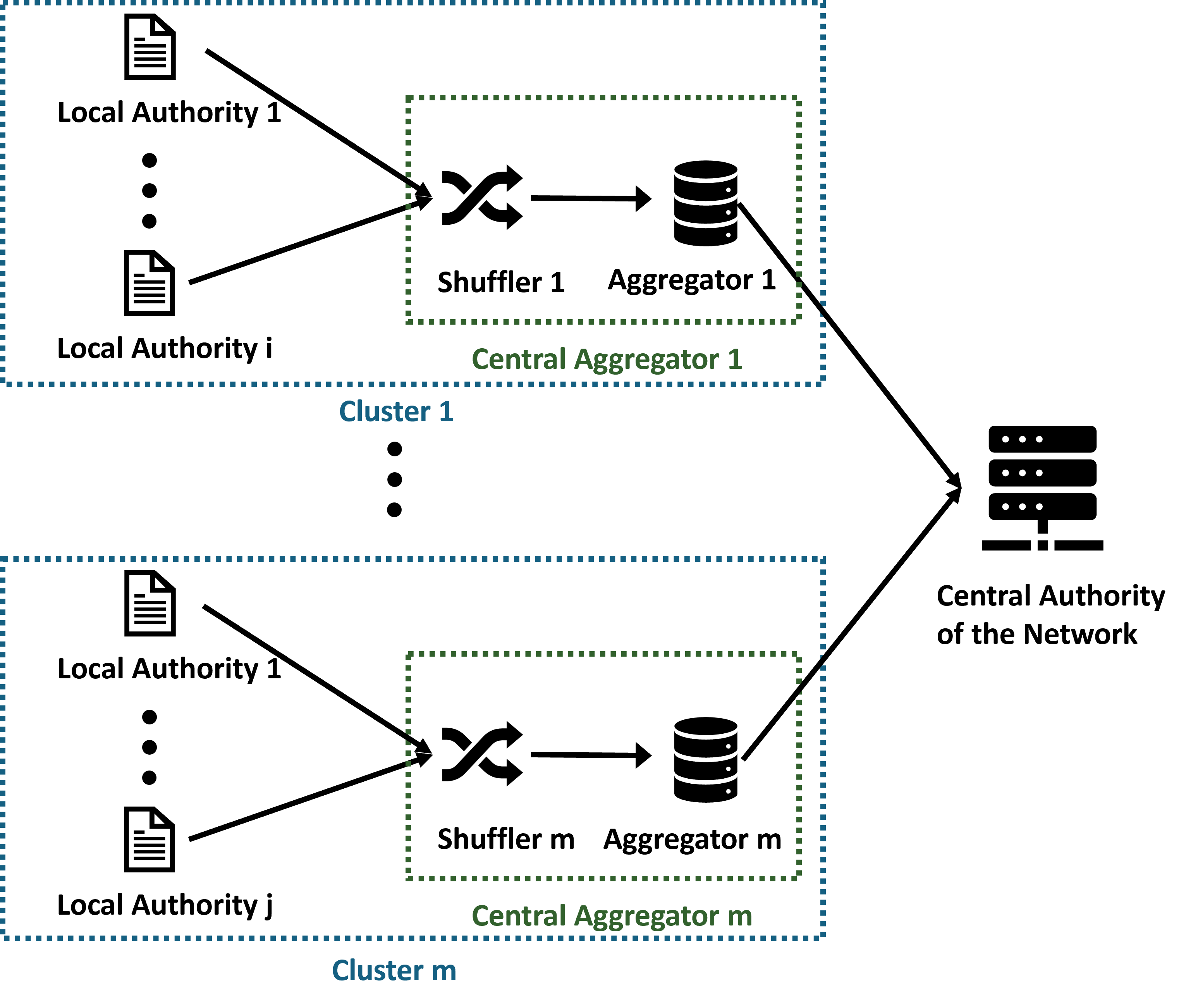}
    \caption{Structure of the network framework\baike{:} 
    \textcolor{black}{The local authorities of the same cluster are responsible for collecting and reporting their local distributed reproduction numbers to the shuffler and then to the aggregator of the same cluster.
    The shuffler is responsible for anonymizing the local distributed reproduction numbers and randomly shuffling them to eliminate any usefulness in their order. The aggregator decodes and groups the local distributed reproduction numbers into cluster distributed reproduction numbers. The central authority collects the outputs of all central aggregators for analysis and is implemented independently of the central aggregators.}
    }
\label{fig:network_structure}
\end{figure}

Figure~\ref{fig:network_structure} shows  the \textcolor{black}{network communication framework} defined in Definition~\ref{Def_Authorities}. This communication framework isolates the central authority from the entity-level (\textcolor{black}{local}) data: the central authority only 
has access to \textcolor{black}{cluster} distributed reproduction numbers after 
they have been aggregated and privatized (which we describe in detail below). 


For the communication network in \textcolor{black}{Figure~\ref{fig:network_structure}}, we implement two forms of differential privacy \textcolor{black}{to process the local distributed reproduction numbers}. The communication network in Figure~\ref{fig:network_structure}, together with our implemented differential privacy mechanism, comprises our privacy framework for distributed reproduction numbers. 
One form of the differential privacy is the local model~\cite{Kasiviswanathan2008what}, where differential privacy is implemented by a local authority for \textcolor{black}{the local distributed reproduction numbers,} and such setups are often referred to as using a ``\textit{local randomizer.}'' 
The local randomizer is a minimal trust model because it does not require trusting any external entity with sensitive
information. 
The other form of differential privacy model is the central model, which is implemented at the central aggregator to further amplify the differential privacy guarantee of the local model against any curious or malicious individuals at the central aggregator.
We elaborate on both models below.

\subsubsection{Local Randomizer}
The local randomizer adds uncertainty (\textcolor{black}{calibrated} random noise in our case) to the data in order to privatize
it before sharing.
Intuitively, the local randomizer must produce outputs that are approximately indistinguishable from each other~\cite{dwork2014algorithmic}
when applied to adjacent
input vectors. 

\begin{definition}[Adjacency] \label{def:adjacency}
    Fix a vector $\vect={[\zeta_i]}_{i\in\setRegionNum}\in\domain^{\regionNum}$ where $\domain^\regionNum\subset\nonnegativeRealSpace^\regionNum$ is its domain. Then another vector 
    $\vect'={[\zeta_i']}_{i\in\setRegionNum}\in\domain^\regionNum$ is \emph{adjacent} to $\vect$, denoted $\vect\adjacent \vect'$, if 
        ${||\vect-\vect'||}_{2} = \sqrt{\sum_{i=1}^{\regionNum}(\zeta_{i}-\zeta'_{i})^2} \leq \adjacency$,
    where $\adjacency>0$ is a user-specified parameter. \hfill $\lozenge$
\end{definition}
Definition~\ref{def:adjacency} states that two vectors are adjacent 
if they have the same dimension 
and the $\ell_2$-distance between them is bounded by~$\adjacency$.
The goal of differential privacy is to render all such adjacent pairs
of vectors approximately indistinguishable, which is enforced
by a local randomizer according to the following definition. 

\begin{definition}[Local Randomizer \cite{Kasiviswanathan2008what}]\label{def:local_randomizer}
    Let $\localRandPrivacyLevel > 0$ be given and fix a probability space $(\Omega, \mathcal{F}, \mathbb{P})$. 
    Then, given a domain $\mathbb D^\regionNum\subseteq\nonnegativeRealSpace^\regionNum$, a local randomizer $\localRandomizer: \Omega \times \mathbb D^\regionNum \rightarrow \mathbb D^\regionNum$ is $\localRandPrivacyLevel$-differentially private if,
    for all adjacent vectors $\vect$ and $\vect'$ in $\mathbb D^\regionNum$, it satisfies $\mathbb{P}\big[\localRandomizer(\vect) \in S\big] \leq e^{\localRandPrivacyLevel} \cdot \mathbb{P}\big[\localRandomizer\left(\vect'\right) \in S\big]$
    for all sets $S$ in the Borel $\sigma$-algebra over $\mathbb D^\regionNum$. \hfill $\lozenge$
\end{definition}

Intuitively, a local randomizer guarantees that given an output~$\privateVect$, a malicious individual cannot reliably tell which candidate vector (i.e., $\vect$ or $\vect'$) generated $\privateVect$. Therefore, the information in the vector~$\vect$ that we want to protect
is concealed, \textcolor{black}{in the sense that it is approximately indistinguishable from 
any other adjacent vector~$\zeta'$.}
The privacy parameter~$\epsilon_0$ 
controls the strength of privacy, and a smaller $\epsilon_0$ implies stronger privacy. 
Typical values of $\epsilon_0$ range from~$0.01$ to~$10$~\cite{hsu2014differential}. 

\textcolor{black}{We consider that each local authority of an entity, 
defined in Definition~\ref{Def_Authorities}, implements a local randomizer on its local distributed reproduction numbers before sharing them with the central authority at the cluster level. Generally, each local randomizer is allowed to choose its own privacy level, i.e., different values of~$\localRandPrivacyLevel$.  However, for simplicity, 
we consider each local authority using the same privacy parameter, $\localRandPrivacyLevel$, for its local randomizer. Based on this setting, we have the following goal.}

\begin{problem}
\label{prob:local_randomizer_design}
\textcolor{black}{Develop a local randomizer that ensures differential privacy for local distributed reproduction numbers shared by local authorities with other parties.} \hfill $\lozenge$
\end{problem}

\subsubsection{Central Differential Privacy}
\textcolor{black}{According to the communication network in Figure~\ref{fig:network_structure}, the local authority uses its local randomizer to generate and share privatized local distributed reproduction numbers with the cluster-level central aggregators, which operate independently of the central authority. Consequently, we introduce the central model of differential privacy.}
The central model of differential privacy can be applied to a centrally-held dataset for which privacy is needed, e.g.,
to conceal \textcolor{black}{the identities of all private local distributed reproduction numbers of the cluster,} from curious or malicious individuals at the central authority. \textcolor{black}{In order to explain the central model, we first introduce neighboring databases.}

\begin{definition}[Neighboring Databases] \label{def:neighboring}
    Let a database $\database=[\vect_i]_{i\in\setLocAuthNum}$ denote a set of vectors received from local authorities. 
    Then two databases $\database$ and $\database'$ are \emph{neighboring} if they differ on a single record $\vect_i,\vect_i'$.\hfill $\lozenge$
\end{definition}

We point out that the notion of ``adjacency'' in Definition~\ref{def:adjacency} applies to two vectors
that differ by one entry, while the notion of ``neighboring'' in Definition~\ref{def:neighboring} applies
to two collections of vectors that differ in one of the vectors that they contain. \textcolor{black}{Next, we use neighboring databases to further introduce central differential privacy.}

\begin{definition}[Central Differential Privacy~\cite{dwork2014algorithmic}]\label{def:central_differential_privacy}
    Let $\epsilon > 0$ and $\delta\in(0,1)$ be given and fix a probability space $(\Omega, \mathcal{F}, \mathbb{P})$. Then given a domain $\mathbb D^{n\cdot m}\subseteq\nonnegativeRealSpace^{n\cdot m}$, an algorithm $\centralDP:\mathbb D^{n\cdot m}\times\Omega\rightarrow\mathbb D^n$ is $(\epsilon,\delta)$-differentially private if, for all neighboring databases $\database$ and $\database'$,
    it satisfies $\mathbb{P}\big[\centralDP(\database) \in S\big] \leq e^{\epsilon} \cdot \mathbb{P}\big[\centralDP\left(\database'\right) \in S\big]+\delta$, for all sets $S$ in the Borel $\sigma$-algebra over $\mathbb D^n$. \hfill $\lozenge$
\end{definition}

Central differential privacy ensures that outputs from neighboring databases, which differ by only a single vector from any local authority, remain statistically similar. This \textcolor{black}{property} makes it difficult to infer high-confidence information about individual vectors, \textcolor{black}{such as the privatized local distributed reproduction numbers that they contain}. Furthermore, we aim to establish an amplified differential privacy guarantee, i.e., that~$(\epsilon, \delta)$-differential privacy holds with some  $\epsilon < \epsilon_0$, by leveraging the central differential privacy model.

\begin{problem}
\label{prob:central_aggregator_design}
    \textcolor{black}{Implement a central differential privacy mechanism at each central aggregator in a 
    way that strengthens the privacy of the privatized cluster distributed reproduction numbers 
    before they are shared with the central authority. 
    }
    \hfill $\lozenge$
\end{problem}

\textcolor{black}{Differential privacy mechanisms add calibrated noise to data to obscure its true values, with higher privacy levels resulting in larger variance of noise. However, high-variance noise can 
produce semantically invalid data, such as negative reproduction numbers. 
Thus, our ultimate goal is to investigate the trade-off between privacy and accuracy.}

\begin{problem}
\label{prob:differential_privacy_accuracy_analysis}
    Quantify the accuracy of private  distributed reproduction numbers as a function of their privacy level. Demonstrate that, despite privacy protections, these reproduction numbers still provide valuable insights for \textcolor{black}{analyzing an} epidemic spreading process.\hfill $\lozenge$
\end{problem}

We address Problems~\ref{prob:1},~\ref{prob:2}, and~\ref{prob:3} in Section~\ref{Sec_The_Distributed_Reproduction_Number}, where we further discuss the benefits of using distributed reproduction numbers at different scales. In Section~\ref{Privacy_Mechanism_For_The_Distributed_Reproduction_Numbers}, we introduce our privacy framework by incorporating distributed reproduction numbers at different scales to address Problems~\ref{prob:local_randomizer_design}, \ref{prob:central_aggregator_design}, and~\ref{prob:differential_privacy_accuracy_analysis}.

\subsection{Probability Background} \label{Sec_Pro_BG}
\begin{definition}[Truncated Gaussian random variable\cite{burkardt2014truncated}]
    The \emph{truncated Gaussian} random variable, written as $\text{TrunG}(\mu,\sigma,l,u)$, that lies within the interval $(l,u]$, where $-\infty< l < u< +\infty$, and centers on $\mu\in(l,u]$ is defined by the probability density function $p_{TG}$ with
    \begin{equation}
        p_{TG}(x) = \begin{cases}
        \frac{1}{\sigma} \frac{\varphi\left(\frac{x-\mu}{\sigma}\right)}{\Phi\left(\frac{u-\mu}{\sigma}\right)-\Phi\left(\frac{l-\mu}{\sigma}\right)} & \text{if } x\in (l,u] \\
        0 & \text{otherwise}
    \end{cases}
    \end{equation}
    and $\sigma>0$, 
where~$\varphi(x) = \frac{1}{\sqrt{2\pi}}\exp\left(-\frac{1}{2}x^2\right)$ is the probability density
of the standard normal distribution and
$\Phi(x) = \frac{1}{2}\left(1+\frac{2}{\sqrt{\pi}}\int_0^{\frac{x}{\sqrt{2}}} \exp(-t^2)dt\right)$
is the cumulative distribution function of the standard normal distribution. 
\hfill $\lozenge$
\end{definition}

\section{Distributed Reproduction Numbers and Network Spreading Behavior}
\label{Sec_The_Distributed_Reproduction_Number}
In this section, we first define \textcolor{black}{local} distributed reproduction numbers for the networked $SIS$ and $SIR$ models to solve Problem~\ref{prob:1}. We then use \textcolor{black}{these reproduction numbers} to study the transient and steady-state behaviors of the spreading models, thus providing a solution to Problem~\ref{prob:2}. We also link the \textcolor{black}{local} distributed reproduction numbers to the network-level reproduction number in Definition~\ref{Def:Repro}. Furthermore, we show that the \textcolor{black}{local distributed reproduction numbers} can be aggregated \textcolor{black}{as cluster distributed reproduction numbers} at various scales to capture interactions between different clusters in the spreading network, which solves Problem~\ref{prob:3}.

\subsection{\textcolor{black}{Local} Distributed Reproduction Numbers}
\label{Sec_Definition of Distributed Reproduction Numbers}
One way to study epidemic spreading processes is to use the reproduction number to indicate the change of the infected population (e.g., increasing, decreasing \textcolor{black}{or} unchanging.). As indicated in~\cite{pare2020modeling,mei2017dynamics,she2021peak}, the network-level reproduction number can capture the overall spreading behavior. 
However, the spreading behavior within a single entity in the network will most likely not be captured by the network-level reproduction number. 
Instead, one could envision distributed reproduction numbers for the local spread that are (i) greater than one when the infected proportion in the local spread is increasing, (ii) less than one when the infected proportion is decreasing, and (iii) equal to one when the infected proportion remains unchanged. 
Building on this intuition, we introduce the following. 


\begin{definition}[\textcolor{black}{Local} Distributed Reproduction Numbers]
\label{Def:NRN}
Let Assumption~\ref{assum_graph_strongly_connected} hold.
Consider the network $SIS$ and $SIR$ models that capture the spread over $n$ entities as described in~\eqref{eq:SIS} and~\eqref{eq:SIR}, respectively.
\begin{itemize}
    \item For each $i \in \underline{n}$,
define $R^{0}_{ii} = \frac{\beta_{ii}}{\gamma_i}$ as the \textcolor{black}{local} endogenous
basic reproduction number (BRN) within entity $i$ itself, and
define $R^{0}_{ij} = \frac{\beta_{ij}}{\gamma_i}$  as the \textcolor{black}{local}
exogenous BRN
from entity $j$ to entity $i$ for each $j \in \underline{n}$. 
\item For each $i \in \underline{n}$, define $R^{t}_{ii} = \frac{s_i(t)\beta_{ii}}{\gamma_i}$ as the \textcolor{black}{local}
endogenous 
effective reproduction number (ERN) within entity $i$, and define $R^{t}_{ij} = \frac{s_i(t)\beta_{ij}}{\gamma_i}$ for each $j \in \underline{n}$ as the \textcolor{black}{local} exogenous pseudo-ERN
from entity $j$ 
to entity $i$. 
\item We define $I_{ij}(t) = \frac{x_j(t)}{x_i(t)}$ with $x_i(t), x_j(t)\in(0,1]$ as the infection ratio of the infected proportion of entity $j$ to 
the infected proportion of entity $i$. 
\item We define the \textcolor{black}{local} exogenous ERN
from entity $j$ 
to entity $i$ as $\bar{R}^t_{ij}=R^{t}_{ij}I_{ij}$. 
The \textcolor{black}{local} endogenous ERN
of entity $i$ is defined as $\bar{R}^t_{ii}=R^{t}_{ii}I_{ii} = R^{t}_{ii}$.
\end{itemize}
Together, these \textcolor{black}{local} endogenous and exogenous BRNs and ERNs 
are referred to as the  \textbf{\textcolor{black}{local} distributed basic reproduction numbers} and \textbf{\textcolor{black}{local} distributed effective reproduction numbers}, respectively.
We refer to all of them collectively as 
\textbf{the local distributed reproduction numbers}. 
\end{definition} 

\begin{figure}
    \centering
    \includegraphics[width=1\linewidth]{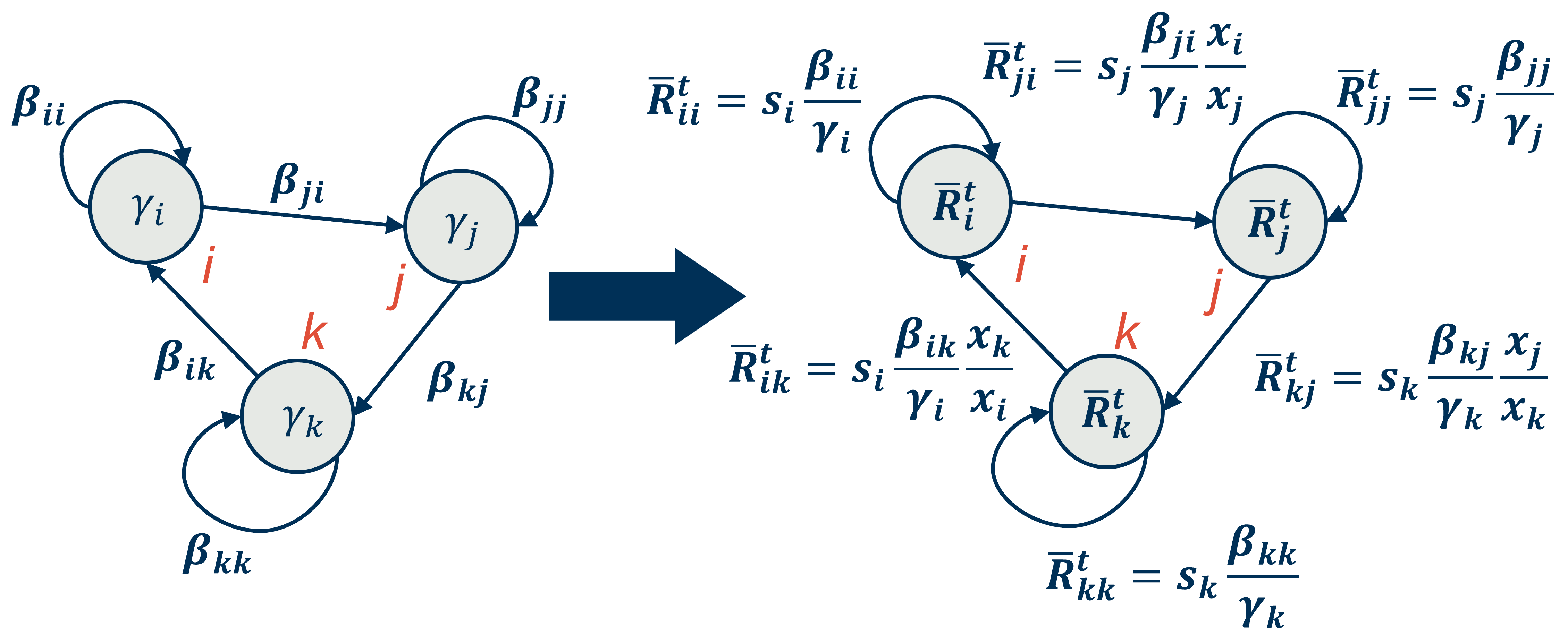}
    \caption{\baike{Local distributed ERNs. The network on the left depicts a spreading network with three nodes $i,j,k$, where the network is modeled by the transmission rates within and between the entities, and the recovery rate within entities, and the network on the right depicts how to model the system using local distributed ERNs. 
    }}
    \label{fig_Local_dis_network_R}
\end{figure}

In order to explain the intuition behind Definition~\ref{Def:NRN}, we consider the group compartmental $SIS$ and $SIR$ models without a network with $\beta$ and $\gamma$ being the transmission and recovery rates, respectively~\cite{mei2017epidemics_review}. The two 
models admit $\frac{\beta}{\gamma}$ and $\frac{s(t)\beta}{\gamma}$ as the basic and effective reproduction numbers, respectively.
Based on these terms, we then choose to use
$R^0_{ij}=\frac{\beta_{ij}}{\gamma_{i}}$ and $R^t_{ij}=\frac{s_i(t)\beta_{ij}}{\gamma_{i}}$ for the \textcolor{black}{local} exogenous basic and pseudo-effective reproduction numbers of the infected proportion $x_{ij}(t)$, where $x_{ij}(t)$ denotes the infected proportion of the $i^{th}$ entity that has been infected by the infected proportion of the $j^{th}$ entity for all $i,j\in\underline{n}$.
Therefore, we can define
\begin{align}
\label{Eq_x_ij}
\dot{x}_{ij}(t) = s_i(t)\beta_{ij}x_j(t)-\gamma_ix_{ij}(t).   
\end{align}
According to Definition~\ref{Def:NRN}, we have $s_i(t)\beta_{ij}x_j(t) - \gamma_i x_i(t) > 0$ if and only if $\bar{R}^t_{ij} > 1$, and  we have $s_i(t)\beta_{ij}x_j(t) - \gamma_i x_i(t) < 0$ if and only if $\bar{R}^t_{ij} < 1$.
Thus, $\bar{R}^t_{ij}$ characterizes the \textcolor{black}{local} interaction between the infection process of~$x_{ij}(t)$ within $x_i(t)$ and the recovery process within $x_i(t)$. Consequently, even if $\bar{R}^t_{ij} < 1$, it is still possible to have $\dot{x}_{ij}(t) \geq 0$ in~\eqref{Eq_x_ij}, as $\bar{R}^t_{ij}$ models the relationship between the infection process of $x_{ij}(t)$ and the recovery process of $x_i(t)$, rather than the infection and recovery processes of $x_{ij}(t)$ itself.

\begin{remark}
Definition~\ref{Def:NRN} proposes \textcolor{black}{local} distributed reproduction numbers by separating the infected cases generated in the $i^{th}$ entity in two ways: 
(i) the new cases that are generated through the infected cases within the entity itself, 
defined as endogenous infections, and (ii) the new cases that are generated through the infected cases from 
neighboring entities, defined as exogenous infections. 
Hence, we use two types of \textcolor{black}{local} reproduction numbers, namely the 
\textcolor{black}{local} endogenous
reproduction numbers ($R^0_{ii}$ and $R^t_{ii}$) and the 
\textcolor{black}{local} exogenous 
reproduction numbers ($R^0_{ij}$ and $\bar{R}^t_{ij}$) to capture the two types of \textcolor{black}{local} transmission processes. Furthermore, the \textcolor{black}{local} exogenous ERN
$\bar{R}^t_{ij}$ is defined with respect to the recovery process of the overall infection within the entity, represented by $\gamma_i x_i(t)$, rather than the recovery of infections generated by individual sources, represented by $\gamma_i x_{ij}(t)$. In addition, similar to the reproduction numbers of group compartmental models, we have $R^t_{ii}=s_i(t)R^0_{ii}$ within entity $i$ for all $i\in\underline{n}$. 
For the \textcolor{black}{local} exogenous pseudo-ERNs
from entity $j$ to entity $i$, 
we have that $R^t_{ij}=s_i(t)R^0_{ij}$ for all $i,j\in\underline{n}$ with $i\neq j$. 
\end{remark}

Definition~\ref{Def:NRN} proposes a way to explain spreading processes through the \textcolor{black}{local} endogenous and exogenous distributed reproduction numbers,~\baike{as illustrated in Figure~\ref{fig_Local_dis_network_R}}. For the purpose of characterizing the spreading process of individual entities \textcolor{black}{locally}, we further define the \textcolor{black}{local} basic reproduction number and \textcolor{black}{local} effective reproduction number of \textcolor{black}{an entity} in the network, through the \textcolor{black}{local} distributed reproduction numbers in Definition~\ref{Def:NRN}.
\begin{definition}[\textcolor{black}{Local} Reproduction Numbers]
\label{def:R_c}
For all ${i \in \underline{n}}$, let $R^0_{i}$ denote the \textbf{\textcolor{black}{local} basic reproduction number} \textcolor{black}{(LBRN)} of entity~$i$, and let $R^t_{i}$ denote the \textbf{\textcolor{black}{local} effective reproduction number} \textcolor{black}{(LERN)} of entity~$i$, where
\begin{align} 
   R^0_{i} &= \sum_{j=1}^{n} R^{0}_{ij}, \label{Net_R0} \\
   \bar{R}_i^t &= \sum_{j=1}^{n}\bar{R}_{ij}^t   =\sum_{j=1}^{n} R^t_{ij}I_{ij}. \label{Net_R}
\end{align}
\end{definition}

\begin{remark}
\label{Remark_reproduction_number}
The \textcolor{black}{local} reproduction numbers defined in~\eqref{Net_R0} and~\eqref{Net_R} 
unify the endogenous and  exogenous infections of entity $i$. Specifically, 
the \textcolor{black}{LBRN} and \textcolor{black}{LERN}
are built upon the \textcolor{black}{local} distributed reproduction numbers from Definition~\ref{Def:NRN}.  Equation~\eqref{Net_R0} indicates that 
\textcolor{black}{the LBRN}
of entity~$i$ within the network is the sum of the
\textcolor{black}{local} 
distributed
BRNs
of entity~$i$.  
Similarly,~\eqref{Net_R} indicates that 
\textcolor{black}{the LERN}
of entity $i$ within the network is the sum of the
\textcolor{black}{local} 
distributed
ERNs
of 
entity $i$.

Unlike the network-level ERN,
where $R^t$ is determined by the transmission rates $B$,  the recovery rates $\Gamma$, and the susceptible proportions $\diag (s(t))$, 
\textcolor{black}{the LERN} of
the  $i^{th}$ entity is determined by  
its \textcolor{black}{local} endogenous and exogenous ERNs, namely,  its 
distributed
ERNs.
Based on Definition~\ref{Def:NRN}, the \textcolor{black}{local} exogenous ERNs
are determined not only by the transmission rates, recovery rates, and susceptible proportions, but also by the scaling factor given by the infection ratio between entity $j$ and entity $i$, represented by $I_{ij}(t)$, for all $i, j \in \underline{n}$.
For instance, at time~$t$, if entity $i$ has a lower infected proportion than entity $j$ (i.e., $x_i(t)<x_j(t)$), 
then the \textcolor{black}{local} exogenous 
ERN
from entity $j$ to entity~$i$ will be scaled up by $I_{ij}(t)$. 
Hence, the \textcolor{black}{local} ERN
of entity $i$ ($\bar{R}^t_i$) can be high, even if its \textcolor{black}{local} endogenous 
ERN ($R^t_{ii}$) and the \textcolor{black}{local} exogenous pseudo-ERNs from entity $j$ to entity $i$ ($R^t_{ij}$) are low, since the weight $I_{ij}(t)$ that is also critical can be large. 
\end{remark}

\subsection{Properties of \textcolor{black}{Local} Reproduction Numbers}
\label{Sec_Properties of Distributed Reproduction Numbers}
Through the \textcolor{black}{local} distributed reproduction numbers introduced in Definition~\ref{Def:NRN}, we can compute the \textcolor{black}{local} reproduction numbers of an entity through the sum of its local distributed reproduction numbers, as shown in Definition~\ref{def:R_c}. 
Compared to the network-level ERN
($R^t$), the \textcolor{black}{local} 
ERN
of an entity ($\bar{R}^t_i$) can facilitate the study of the spreading behavior of entity~$i$. 
\begin{thm}
\label{thm:net_r}
Consider the network $SIS$ and  $SIR$ models in~\eqref{eq:SIS} and~\eqref{eq:SIR}, respectively.
When the infection in entity~$i$ is nonzero for all $i\in\underline{n}$, i.e., $x(t)\gg0$, 
\textcolor{black}{the LERN of entity~$i$ given by $\bar{R}^t_i$ satisfies the following properties:}
\begin{itemize}
    \item ${\bar{R}^t_i>1}$ if and only if the infected proportion $x_i$ increases;
    \item $\bar{R}^t_i<1$ if and only if $x_i(t)$ decreases;
    \item $\bar{R}^t_i = 1$ if and only if $x_i(t)$ remains unchanged.
\end{itemize}
\end{thm}

\begin{proof}
We show the first statement since the proofs of the other statements  follow the same procedure.

$\Leftarrow:$ Recall the definition of \textcolor{black}{the LERN} of entity $i$, namely $\bar{R}_i^t = \sum_{j=1}^{n} R^t_{ij}I_{ij}(t)$, for all $i \in \underline{n}$. 
Hence, $\bar{R}_i^t>1$ gives $\bar{R}_i^t = \sum_{j=1}^{n} R^t_{ij}I_{ij}(t)>1$. Then, through Definition~\ref{Def:NRN}, it is true that $\sum_{j=1}^{n}\frac{s_i(t)\beta_{ij}x_j(t)}{x_i(t)\gamma_i}>1$, which leads to $\frac{dx_i(t)}{dt}=\sum_{j=1}^{n} s_i(t)\beta_{ij}x_j(t)-\gamma_ix_i(t)>0$, for all $i\in \underline{n}$. Hence,~$x_i(t)$ is increasing. 

$\Rightarrow:$ If the infected proportion of entity $i$ increases,  then $\frac{dx_i(t)}{dt}=\sum_{j=1}^{n} s_i\beta_{ij}x_j(t)-\gamma_ix_i(t)>0$, 
and re-arranging terms immediately gives $\bar{R}_i^t=\sum_{j=1}^{n}\frac{s_i(t)\beta_{ij}x_j(t)}{x_i(t)\gamma_i}>1$, since we have $x_i(t)>0$.
\end{proof}

Theorem~\ref{thm:net_r} demonstrates that LERNs exhibit threshold behavior at a value of one, allowing us to use 
\textcolor{black}{them}
to capture \textcolor{black}{individual entities'}  spreading behaviors. 
In addition, the network-level
ERN of the network $SIR$ dynamics,
namely $R^t$, is monotonically non-increasing as a function of~$t$, 
since for all~$i \in \underline{n}$, the value of $s_i(t)$ is monotonically non-increasing \cite{mei2017epidemics_review}. However, 
for all~$i \in \underline{n}$ the value of 
$\bar{R}_i^t$ can be non-monotonic.
\begin{lemma}
\label{lem:non}
Consider the network $SIS$ and  $SIR$ models in~\eqref{eq:SIS} and~\eqref{eq:SIR}, respectively.
For all~$i \in \underline{n}$, 
\textcolor{black}{the LERN} of the~$i^{th}$ entity \textcolor{black}{given by $\bar{R}_i^t$}  can be non-monotonic.
If $x_{j}(t)$  decreases no slower than~$x_i(t)$, and $s_i(t)$ is monotonically decreasing for all $t\in[t_1,t_2]$, then 
\textcolor{black}{the LERN} of the $i^{th}$ entity $\bar{R}_i^t$ decreases monotonically with respect to $t$, for all $t\in[t_1,t_2]$.
\end{lemma}

\begin{proof}
The LERN $\bar{R}_i^t$ of entity $i$ is a weighted sum of 
its local distributed ERNs
$R^t_{ij}$ with the weights $I_{ij}$. Based on Definition~\ref{Def:NRN}, $R^t_{ij}=\frac{s_i(t)\beta_{ij}}{\gamma_i}$ is monotonically decreasing for network $SIR$ dynamics, since $s_i(t)$ is monotonically non-increasing. However, $R^t_{ij}=\frac{s_i(t)\beta_{ij}}{\gamma_i}$ is non-monotonic 
for the network $SIS$ dynamics due to the fact that $s_i(t)$ is non-monotonic. 
Further, the weights $I_{ij}(t)$ are determined by the ratio between the infected proportions of entities $j$ and $i$ 
for all $i,j\in \underline{n}$. The weights $I_{ij}(t)$ can be non-monotonic, 
and therefore,
for all~$i \in \underline{n}$, 
the LERN~$\bar{R}_i^t$ 
can be non-monotonic for both the network $SIS$ and $SIR$ models.

Further, consider the $i^{th}$ entity for $i \in \underline{n}$. Under the condition that $x_{j}(t)$ for all $j \in \underline{n}$ and $j \neq i$ decreases no slower than $x_i(t)$, for all $t\in[t_1,t_2]$, we have that $I_{ij}(t) = \frac{x_{j}(t)}{x_i(t)}$ is a monotonically non-increasing function of~$t$ for all $i \in \underline{n}$. In addition, we have that $s_i(t)$ is 
monotonically non-increasing, for all $t\in[t_1,t_2]$.
Thus, the LERN $\bar{R}_i^t$ is also a monotonically decreasing function of~$t$, for all $t\in[t_1,t_2]$.
\end{proof}

Theorem~\ref{thm:net_r} and Lemma~\ref{lem:non} demonstrate that we can leverage \textcolor{black}{LERNs}
to capture spreading behaviors at the entity level in the network. Hence, we have answered Problem~\ref{prob:1}. In order to answer Problem~\ref{prob:2}, we connect 
\textcolor{black}{LERNs}
to the network-level reproduction numbers, 
namely $R^0$ and $R^t$ in Definition~\ref{Def:Repro}. First we define the \textcolor{black}{local} distributed reproduction number matrices.
\begin{definition}[\textcolor{black}{Local} Distributed Reproduction Number Matrices]
\label{def:MNR}
The \textcolor{black}{local} distributed basic reproduction number matrix is
\begin{equation}
\!\!\!\!\mathcal{R}^0 = 
\begin{bmatrix}
         R^0_{11} & R^0_{12} & \cdots & R^0_{1n}\\
         R^0_{21} & R^0_{22} & \cdots & R^0_{2n}\\ 
         \vdots & \vdots & \ddots & \vdots\\ 
         R^0_{n1} & R^0_{n2} & \cdots & R^0_{nn} 
\end{bmatrix}, \label{eq:R_0},
\end{equation}
the \textcolor{black}{local} distributed pseudo-effective reproduction number matrix is
\begin{equation}
\!\!\!\!\mathcal{R}^t = \diag([s_1,\dots,s_n])\mathcal{R}^0 \!=\!
\begin{bmatrix}
         R^t_{11} & R^t_{12} & \cdots & R^t_{1n}\\
         R^t_{21} & R^t_{22} & \cdots & R^t_{2n}\\ 
         \vdots & \vdots & \ddots & \vdots\\ 
         R^t_{n1} & R^t_{n2} & \cdots & R^t_{nn} 
     \end{bmatrix}\!\!, \label{eq:R_t}
\end{equation}
and the \textcolor{black}{local} distributed effective reproduction number matrix is defined as 
\begin{align}
\!\!\!\!\bar{\mathcal{R}}^t &= \diag([\frac{1}{x_1},\dots,\frac{1}{x_n}])  \mathcal{R}^t \diag([x_1,\dots,x_n])\\
   &=\begin{bmatrix}
         \bar{R}^t_{11} & \bar{R}^t_{12} & \cdots & \bar{R}^t_{1n}\\
         \bar{R}^t_{21} & \bar{R}^t_{22} & \cdots & \bar{R}^t_{2n}\\ 
         \vdots & \vdots & \ddots & \vdots\\ 
         \bar{R}^t_{n1} & \bar{R}^t_{n2} & \cdots & \bar{R}^t_{nn} 
     \end{bmatrix}. \label{eq:bar_R_t}  
\end{align}
\end{definition}
\begin{remark}
\label{Remark_distributed_R_matrix}
The \textcolor{black}{local} distributed 
BRN
matrix
$\mathcal{R}^0=\Gamma^{-1}B$ is the \textit{next generation matrix \cite{diekmann2010construction}} of the network $SIS$/$SIR$ models. Thus, the \textcolor{black}{local} distributed pseudo-ERN
matrix $\mathcal{R}^t$ is equal to~$\diag(s(t))\mathcal{R}^0$. 
The advantage of viewing $\mathcal{R}^0$ and $\mathcal{R}^t$ as the composition of \textcolor{black}{local distributed reproduction numbers} defined in~Definition~\ref{Def:NRN}, is that we can construct these matrices through the 
\textcolor{black}{local} 
distributed reproduction numbers.
\textcolor{black}{We illustrate the benefits of this formulation in Section~\ref{sec_Simulation}. In addition, this construction lays the foundation for the proposed differential privacy framework introduced in Section~\ref{Privacy_Mechanism_For_The_Distributed_Reproduction_Numbers}.}
\end{remark}

\begin{remark}
\label{Remark_local_distr_Rt}
\textcolor{black}{We explain the structure of  $\bar{\mathcal{R}}^t$ in a detailed manner. 
The off-diagonal entries of the $i^{th}$ row of $\bar{\mathcal{R}}^t$ are the \textcolor{black}{local} exogenous ERNs of the $i^{th}$ entity,  as defined by $\bar{R}^t_{ij}$ in~\eqref{Net_R}, while 
the diagonal entry of the $i^{th}$ row of $\bar{\mathcal{R}}^t$ is the \textcolor{black}{local} endogenous ERN of the $i^{th}$ entity,  as defined by $\bar{R}^t_{ii}$ in~\eqref{Net_R}. 
\textcolor{black}{Therefore, the $i^{th}$ row of $\bar{\mathcal{R}}^t$, denoted as $\bar{\mathcal{R}}_{i,:}^t$, comprises the local distributed ERNs of entity~$i$, as defined in Definition~\ref{Def:NRN}. We further name $\bar{\mathcal{R}}_{i,:}^t$ as the \textit{\textbf{local distributed effective reproduction number vector}} of the $i^{th}$ entity, for all $i\in\underline{n}$.}
Consequently, the \textcolor{black}{local authority} of entity~$i$ can leverage its own infection data, such as primary infected cases that cause secondary infections (via contact tracing) and the current susceptible proportion within the entity, 
to compute its local distributed ERN vector, $\bar{\mathcal{R}}^t_{i,:}$. 
According to the communication framework introduced in Figure~\ref{fig:network_structure}, the local authorities within the same cluster can share their local distributed ERN vectors 
with the central aggregator of the cluster,
enabling a distributed approach to constructing $\bar{\mathcal{R}}^t$. We further introduce the usage of the local distributed ERN vectors in Section~\ref{Privacy_Mechanism_For_The_Distributed_Reproduction_Numbers}.} 
\end{remark}

Based on Definition~\ref{Def:Repro}, it can be observed that $\mathcal{R}^0=\Gamma^{-1}B$ and $\mathcal{R}^t=\diag(s)\Gamma^{-1}B$. Hence, the spectral radius of the \textcolor{black}{local} distributed BRN matrix,
denoted by $\rho(\mathcal{R}^0)$, is the network-level BRN, i.e., $\rho(\mathcal{R}^0)=R^0$. Meanwhile, the spectral radius of the \textcolor{black}{local} distributed pseudo-ERN matrix, denoted $\rho(\mathcal{R}^t)$, is the ERN of the network, i.e., $\rho(\mathcal{R}^t)=R^t$. Consequently, for all~$i \in \underline{n}$,
the $i^{th}$ row sum of the \textcolor{black}{local} distributed BRN matrix is the \textcolor{black}{LBRN} of the $i^{th}$ entity, 
i.e., we have $\sum_{j=1}^n [\mathcal{R}^0]_{ij}=R^0_i$. 
The  $i^{th}$ row sum of the \textcolor{black}{local} distributed ERN matrix, $\bar{\mathcal{R}}^t$, is equal to $R_i^t$, i.e., \textcolor{black}{the LERN} of the $i^{th}$ entity. 
Note that the $i^{th}$ row of the \textcolor{black}{local} distributed ERN matrix is also the local distributed ERN vector the of $i^{th}$ entity. 
Through studying the spreading behavior of the network, we connect the network-level effective reproduction number to the \textcolor{black}{local effective reproduction numbers} of the entities in the network, which are defined in Definition~\ref{def:R_c}.

\begin{thm}
\label{thm:connection}
For all $i\in \underline{n}$, if $x(t)\gg 0$, then the following statements hold:
\begin{itemize}
    \item $\bar{R}^t_i=1$ for all $i\in \underline{n}$ only if $\rho(\mathcal{R}^t)=1$;
    \item $\bar{R}^t_i<1$ for all $i\in \underline{n}$ only if $\rho(\mathcal{R}^t)<1$;
    \item $\bar{R}^t_i>1$ for all $i\in \underline{n}$ only if $\rho(\mathcal{R}^t)>1$.
\end{itemize}
\end{thm}
\begin{proof}
We start by proving the first statement. If ${\bar{R}_i^t=1}$ and $x_i(t)>0$ for all $i\in \underline{n}$, 
then we have that the matrix $\diag(x(t))^{-1}\mathcal{R}^t\diag(x(t))$ is a row stochastic matrix;
this can be seen by noting that the $i^{th}$ row sum is $\bar{R}_i^t = \sum_{j=1}^{n} R^t_{ij}I_{ij}(t) = 1$. 
Hence, based on the fact that the spectral radius of a row stochastic matrix is one, we have
\begin{equation*}
    \rho\Big(\diag(x(t))^{-1}\mathcal{R}^t\diag(x(t))\Big)=1.
\end{equation*}
Further, it is true that $\rho\left(\diag(x(t))^{-1}\mathcal{R}^t\diag(x(t))\right)=\rho(\mathcal{R}^t)=1$,
since similarity transformations preserve eigenvalues, i.e.,
$\diag(x(t))^{-1}\mathcal{R}^t\diag(x(t))$ and $\mathcal{R}^t$ have the same spectrum.
Therefore, if $\bar{R}_i^t=1$ for all $i\in \underline{n}$, 
then we must have $\rho(\mathcal{R}^t) = 1$.

Next we show the second statement. If $\bar{R}^t_i  = \sum_{j=1}^{n} R^t_{ij}I_{ij} <1$ for all $i\in \underline{n}$, then,
using the fact that~$\bar{R}_i^t$ is equal
to the~$i^{th}$ row sum of~$\diag(x(t))^{-1}\mathcal{R}^t\diag(x(t))$, we see that 
we must have $[(\diag(x(t))^{-1}\mathcal{R}^t\diag(x(t))]_{ij}\in[0,1)$ for all $i,j\in \underline{n}$. 
Now suppose for the sake of contradiction 
that ${\rho\left(\diag(x(t))^{-1}\mathcal{R}^t\diag(x(t))\right)=\rho(\mathcal{R}^t)\geq 1}$. 
Then, by increasing some non-zero entries of the matrix  $\diag(x(t))^{-1}\mathcal{R}^t\diag(x(t))$ through changing $\mathcal{R}^t$, we can construct a new matrix $\diag(x(t))^{-1}\mathcal{\tilde{R}}^t\diag(x(t))$ such that $\diag(x(t))^{-1}\mathcal{\tilde{R}}^t\diag(x(t))$ is a stochastic matrix,
i.e., its row sums equal~$1$. 
Thus, $\rho \left(\diag(x(t))^{-1}\mathcal{\tilde{R}}^t\diag(x(t))\right)=1$. 

According to Assumption~\ref{assum_graph_strongly_connected}, the transmission matrix $B$ is an irreducible matrix since  the transmission network is strongly connected.
Further, the model parameters $\beta_{ij}$ and $\gamma_{i}$ for all $i,j\in\underline{n}$
along with the infected state $x_i(t)$ for all $i\in\underline{n}$ are positive. Thus, the matrices $\diag(x(t))^{-1}\mathcal{R}^t\diag(x(t))$ and $\diag(x(t))^{-1}\mathcal{\tilde{R}}^t\diag(x(t))$ are nonnegative and irreducible. 

Based on \cite[ Thm. 2.7 and Lemma 2.4]{varga2009matrix_book},
the spectral radius of a non-negative irreducible matrix will increase when any entry of the matrix increases, which gives 
\begin{multline} 
  \rho \left(\diag(x(t))^{-1}\mathcal{R}^t\diag(x(t))\right) < \\ \rho \left(\diag(x(t))^{-1}\mathcal{\tilde{R}}^t\diag(x(t))\right) = 1.
\end{multline}
This result 
contradicts the aforementioned hypothesis that $\rho\left(\diag(x(t))^{-1}\mathcal{R}^t\diag(x(t))\right)=\rho(\mathcal{R}^t)\geq1$. Therefore, we must have $\rho\left(\diag(x(t))^{-1}\mathcal{R}^t\diag(x(t))\right)=\rho(\mathcal{R}^t)<1$.

We can use the same techniques to show the third statement. Hence, we complete the proof.
\end{proof}
\begin{remark}
\textcolor{black}{Theorem~\ref{thm:connection} bridges the gap between the network-level ERN $R^t$ and the LERNs $\bar{R}^t_i$. Definition~\ref{def:R_c} and Remark~\ref{Remark_local_distr_Rt} indicate  that a local authority can compute its own LERN by aggregating its local distributed ERN vector. 
Theorem~\ref{thm:connection} ensures that the central authority of the overall network can use the LERNs provided by the local authorities to assess the network's overall spreading behavior. This approach eliminates the need for the central authority to gather detailed network information, such as the transmission matrix $B$, the recovery matrix $\Gamma$, and the susceptible proportions $\diag(s(t))$ across the network.} 
\end{remark}

\begin{coro}
\label{Coro_ditri_effe_repro_matrix}
\textcolor{black}{If $x(t)\gg 0$, then the following statements hold:
\begin{itemize}
    \item $\bar{R}^t_i=1$ for all $i\in \underline{n}$ only if $\rho(\bar{\mathcal{R}^t})=1$;
    \item $\bar{R}^t_i<1$ for all $i\in \underline{n}$ only if $\rho(\bar{\mathcal{R}^t})<1$;
    \item $\bar{R}^t_i>1$ for all $i\in \underline{n}$ only if $\rho(\bar{\mathcal{R}^t})>1$.
\end{itemize}
}
\end{coro}

Corollary~\ref{Coro_ditri_effe_repro_matrix} leverages the condition in the proof of Theorem~\ref{thm:connection} where the spectrum of the \textcolor{black}{local} distributed pseudo-ERN matrix $\rho(\mathcal{R}^t)$ is the same as that of the \textcolor{black}{local} ERN matrix $\rho(\bar{\mathcal{R}^t})$, under the condition that $x(t)\gg 0$. Therefore, we omit the proof. In addition, we can use $\rho(\bar{\mathcal{R}^t})$ to characterize the spreading behavior of the network $SIS$ and $SIR$ models, as defined in~\eqref{eq:SIS} and~\eqref{eq:SIR}, respectively. Hence, Theorem~\ref{thm:connection} establishes the conditions under which we can switch from $\rho(\mathcal{R}^t)$ to $\rho(\bar{\mathcal{R}}^t)$ to analyze the spreading network.

\begin{remark}
\label{Remark_ditribute_effe_Rt_matrix}
 As discussed in Remark~\ref{Remark_distributed_R_matrix}, the advantage of leveraging $\bar{\mathcal{R}^t}$ instead of  $\mathcal{R}^t$ is that the $i^{th}$ row of  $\bar{\mathcal{R}^t}$ is given by the \textcolor{black}{local} 
 distributed ERN vector
 of the $i^{th}$ entity in the network, where $i\in\underline{n}$. In contrast, the $i^{th}$ row of  $\mathcal{R}^t$ is given by the \textcolor{black}{local} 
 distributed
 pseudo-ERNs. 
Definition~\ref{Def:NRN} states that the threshold of $\bar{R}^t_{ij}$, rather than $R^t_{ij}$, at the value of one captures the infection dynamics from entity $j$ to entity $i$ for $i,j\in\underline{n}$.
Additionally, according to Corollary~\ref{Coro_ditri_effe_repro_matrix},
 $\bar{\mathcal{R}^t}$ not only captures the network-level reproduction numbers in the same way as $\mathcal{R}^t$, but it also provides the \textcolor{black}{local distributed ERN vectors of the entities. 
 Therefore, constructing the \textcolor{black}{local} distributed ERN matrix
$\bar{\mathcal{R}}^t$ in~\eqref{eq:bar_R_t} provides more valuable information for analysis and is more practical for real-world implementation (illustrated in~Sections~\ref{Privacy_Mechanism_For_The_Distributed_Reproduction_Numbers} and~\ref{sec_Simulation}). 
\baike{We \baike{leverage} this connection to effective reproduction numbers to design the privacy framework for local authorities
sharing their local distributed ERNs with higher authorities in Section~\ref{Privacy_Mechanism_For_The_Distributed_Reproduction_Numbers}.}}
\end{remark}


\textcolor{black}{Theorem~\ref{thm:connection} and Corollary~\ref{Coro_ditri_effe_repro_matrix} address Problem~\ref{prob:2},} demonstrating that \textcolor{black}{local} distributed effective reproduction numbers not only aid in analyzing the behavior of individual entities but also offer insights into the overall spreading dynamics of the network, through their connection to both the \textcolor{black}{local} distributed pseudo-effective reproduction number matrix $\mathcal{R}^t$, the local distributed effective reproduction number matrix $\bar{\mathcal{R}}^t$, and the local  effective reproduction number vector $\bar{\mathcal{R}}^t_{i,:}$, for all $i\in\underline{n}$. 

\subsection{\textcolor{black}{Cluster} Effective Distributed Reproduction Numbers}
\textcolor{black}{In real-world epidemic spreading processes, data can be collected and shared at various scales by different authorities. The local distributed ERN matrix in~\eqref{eq:bar_R_t} that is composed of the local distributed ERN vectors
provides a mechanism for local authorities to share their distributed ERN vectors with the higher authorities of the network. However, it is reasonable for local authorities to first share their information with an intermediate higher authority for better data management and aggregation.
For example, it can be overwhelming for the central authority of the network, e.g., at the country-level, to process all the information at the highest resolution, i.e., directly from local authorities at the county- or community-level. This intermediate-level procedure not only aids in managing data but also isolates the local distributed ERNs from the central authority, providing an additional layer of privacy. Note that, according to the discussion on the distributed ERNs in  Remark~\ref{Remark_ditribute_effe_Rt_matrix}, starting from this section, we focus on the effective reproduction number. We do not consider the basic or pseudo-effective reproduction numbers.}

\textcolor{black}{
Recall from Figure~\ref{fig:network_structure}, 
we consider the central aggregators of the clusters as defined in Definition~\ref{Def_Authorities}, where each central aggregator aggregates only the local distributed ERN vectors of the entities within its respective cluster.
We demonstrate that, by using the local ERNs and local distributed ERNs defined in Section~\ref{Sec_Definition of Distributed Reproduction Numbers}, we can define the ERN and distributed ERN at the cluster level, namely, the cluster effective reproduction number and the cluster distributed effective reproduction number.
Thus, we provide a method for the central aggregators of the clusters to generate their cluster distributed ERNs using their own local distributed  ERNs. Furthermore, we show that the cluster ERN and the cluster distributed ERN can model and describe the spreading behavior among clusters, thereby addressing Problem~\ref{prob:3}.}

\baike{Consider partitioning the node set $V$ of a spreading graph $G$ with $n$ nodes into $m$ clusters. Let  $\pi: V\rightarrow \ X_{\pi}=\{\chi_1,\dots,\chi_m\}$ denote a map that partitions the node set $V$ into a set of distinct nonempty clusters~$\chi_q$, where $p,q\in\underline{m}$, such that $\bigcup_{q=1}^m\chi_q=V$ and $\chi_p\cap\chi_q = \varnothing$ for all $p\neq q$.}


\begin{definition}[\textcolor{black}{Cluster} Effective Reproduction Number~\textcolor{black}{(CERN)}]
\label{Def:Generalized_Rt}
  Consider the cluster $\chi_q \in X_\pi=\{\chi_1,\dots,\chi_m\}$ and $\chi_q \neq \varnothing$ for each $q\in\underline{m}$. If~$x(t)\gg0$,
    then the \textit{\textcolor{black}{cluster} effective reproduction number \textcolor{black}{(CERN)} of  $\chi_q$} 
    is given by 
$\bar{R}_{\chi_q}^t=\frac{\sum_{i\in\chi_q} \gamma_ix_i(t) \bar{R}^t_{i}}{\sum_{i\in\chi_q}\gamma_ix_i(t)}$.
\end{definition}
\begin{remark}
Similar to 
\textcolor{black}{the LERN}
of individual entities in~Definition~\ref{def:R_c}, \textcolor{black}{the CERN}
of the cluster~$\chi_q$ \textcolor{black}{given by $\bar{R}^t_{\chi_q}$} captures the aggregated spreading behavior of the group of entities that comprise~$\chi_q$.
According to Definition~\ref{Def:Generalized_Rt}, 
\textcolor{black}{the CERN}
of  the cluster $\chi_q$ can be computed by aggregating the \textcolor{black}{LERNs}
of all entities within the cluster, i.e., by aggregating $\bar{R}_i^t$ for all $i \in \chi_q$ via its central aggregator (Figure~\ref{fig:network_structure}). This aggregation process also requires the infected proportions $x_i$ and recovery rates $\gamma_i$ of entity $i$, for all $i \in \chi_q$. Note that the infected proportion $x_i$ and the recovery rate $\gamma_i$ are typically not considered sensitive information in the context of disease spreading. Infection proportions, rather than the identities of individual cases, are often publicly reported during pandemics, and recovery rates can typically be estimated based on the average duration of infection observed in such outbreaks. We further discuss this point in Section~\ref{Privacy_Mechanism_For_The_Distributed_Reproduction_Numbers}. 
\end{remark}
\begin{thm}
\label{Thm:Gen_Reprod}
For all $q\in\underline{m}$ and $x(t)\gg 0$, 
\textcolor{black}{the CERN}
of the cluster~$\chi_q$, \textcolor{black}{given by $\bar{R}_{\chi_q}^t$}, satisfies the following properties:
\begin{itemize}
    \item 
    $\bar{R}_{\chi_q}^t=1$ if and only if $\sum_{i\in\chi_q}\dot x_i(t)=0$;
    \item 
    $\bar{R}_{\chi_q}^t>1$ if and only if $\sum_{i\in\chi_q}\dot x_i(t)>0$;
    \item 
    $\bar{R}_{\chi_q}^t<1$ if and only if $\sum_{i\in\chi_q}\dot x_i(t)<0$.
\end{itemize}
\end{thm}
\begin{proof}
We consider the first statement. We show the necessary condition first. Consider~$|\chi_q|=m\leq n$ clusters within the spreading network. 
According to~\eqref{eq:SIS},
the clustered spreading behavior of the total~$m$ entities is captured by 
\begin{align*}
    \sum_{i\in\chi_q} \dot x_i(t) = -\sum_{i\in\chi_q}\gamma_ix_i(t)+\sum_{i\in\chi_q}\sum_{j=1}^{n}s_i(t)\beta_{ij}x_j(t).
\end{align*}
If~$\sum_{i\in\chi_q}\dot x_i(t)=0$, then 
\begin{align*}
-\sum_{i\in\chi_q}\gamma_ix_i(t)+\sum_{i\in\chi_q}\sum_{j=1}^{n}s_i(t)\beta_{ij}x_j(t)=0,
\end{align*}
which indicates that
\begin{align*}
\frac{\sum_{i\in\chi_q}\gamma_ix_i(t)}{\prod_{k=1}^n\gamma_kx_k(t)}&=\frac{\sum_{i\in\chi_q}\sum_{j=1}^{n}s_i(t)\beta_{ij}x_j(t)}{\prod_{k=1}^n\gamma_kx_k(t)}\\
&=\sum_{i\in\chi_q}\frac{1}{\prod_{k=1,k\neq i}^n\gamma_kx_k(t)}\sum_{j=1}^{n}s_i(t)\frac{\beta_{ij}x_j(t)}{\gamma_ix_i(t)}\\
&= \sum_{i\in\chi_q}\frac{1}{\prod_{k=1,k\neq i}^n\gamma_kx_k(t)}\sum_{j=1}^{n}\bar{R}_{ij}^t\\
&= \sum_{i\in\chi_q}\frac{1}{\prod_{k=1,k\neq i}^n\gamma_kx_k(t)}\bar{R}_i^t.
\end{align*}
We use the fact that  
\textcolor{black}{the LERN}
of entity~$i$ is given by~${\sum_{j=1}^{n}s_i(t)\frac{\beta_{ij}x_j(t)}{\gamma_ix_i(t)}=\sum_{j=1}^{n}\bar{R}^t_{ij}=\bar{R}^t_{i}}$.
Further, by multiplying by~$\frac{\prod_{k=1}^n\gamma_kx_k(t)}{\sum_{i\in\chi_q}\gamma_ix_i(t)}$ on both sides of the equation, we obtain that 
\begin{align*}
1=\frac{\sum_{i\in\chi_q} \gamma_ix_i(t) \bar{R}^t_{i}}{\sum_{i\in\chi_q}\gamma_ix_i(t)}=\bar R^t_{\chi_q}.
\end{align*}
Therefore, we have shown that $\sum_{i\in\chi_q}\dot x_i(t)=0$ will result in ${\bar{R}^t_{\chi_q}=1}$. 
We can reverse the process to show that, when~$ {\bar R^t_{\chi}=1}$, the sum of the infected proportions within the clustered entities is zero, i.e.,~$\sum_{i\in\chi}\dot x_i(t)=0$. 
Additionally, we can show the second and third statements using the same approach, and we omit their proofs. 
\end{proof}

\begin{remark}
Theorem~\ref{thm:net_r}  indicates that
\textcolor{black}{the LERN}
of entity~$i$, denoted as~$\bar{R}^t_{i}$, effectively captures the spreading behavior of that entity. Similarly, 
\textcolor{black}{Definition~\ref{Def:Generalized_Rt} bridges the gap between \textcolor{black}{the CERN}
of the cluster~$\chi_q$, given by $\bar{R}^t_{\chi_q}$, 
where $q\in\underline{m}$, 
and 
\textcolor{black}{the LERNs}
of the individual entities that comprise the cluster, given by $\bar{R}^t_i$, for all $i\in\chi_q$.}
When local authorities compute the \textcolor{black}{local} distributed ERNs of individual entities, they can report these reproduction numbers, along with the infected proportions and recovery rates, to their corresponding central aggregators, \textcolor{black}{as illustrated in Figure~\ref{fig:network_structure}}. This \textcolor{black}{framework} allows those central aggregators of the clusters to further compute their \textcolor{black}{CERNs}.

When~$\chi=\underline{n}$, i.e., when all $n$ entities  are considered as one cluster, the 
\textcolor{black}{CERN} of
the~$n$ clustered entities, namely $\bar R^t_{\chi}$, is not equivalent to the overall network-level effective reproduction number~$\rho(\mathcal{R}^t)$. The network-level ERN captures the change in the weighted average of the infected proportions across all entities, whereas the CERN captures the change in the total sum of infected proportions. \textcolor{black}{When we have the LERNs from all local authorities of the entities across the spreading network, we can formulate a distributed framework to aggregate the LERNs to update
$\bar R^t_{\chi}$, in order to assess the overall spread of the network. 
However, if we are only interested in analyzing the overall spreading behavior, the network-level ERN~$\rho(\mathcal{R}^t)$ is sufficient to determine the network-level 
equilibrium properties~\cite{mei2017epidemics_review,pare2020modeling_review}.}
\end{remark}
Based on Definition~\ref{Def:Generalized_Rt} and the proof of Theorem~\ref{Thm:Gen_Reprod}, we provide a method for obtaining the 
\textcolor{black}{CERNs by aggregating the LERNs. The CERN is calculated as a weighted sum of the LERNs of its corresponding entities. We further develop the following result to demonstrate that we can compute the CERN at different scales.} 

\baike{Consider another mapping that partitions the node set $V$ of the spreading graph $G$ with $n$ into a set of distinct non-empty clusters, such that 
$\hat{\pi}:V\rightarrow \hat{X}_{\pi} =\{\hat{\chi}_1, \dots, \hat{\chi}_{\hat{m}}\}$, where $\bigcup_{q=1}^{\hat{m}}\hat{\chi}_q=V$ and  $\hat{\chi}_p\cap \hat{\chi}_q=\varnothing$, for all $p,q\in\underline{\hat{m}}$ and $p\neq q$. 
We further consider a mapping from
$X_{\pi}\rightarrow\hat{X}_\pi$ that maps each cluster in the partition $X_{\pi}=\{\chi_1,\dots,\chi_m\}$ 
to a cluster in $\hat{X}_{\pi}=\{\hat{\chi}_1,\dots,\hat{\chi}_{\hat{m}}\}$, 
where $2\leq\hat{m}<m$. 
Therefore, we can say that $X_{\pi}$ is a finer partition of $\hat{X}_{\pi}$ in $V$. 
We demonstrate that the CERNs of the clusters in the coarser partition $\hat{X}_{\pi}$ can be derived from the CERNs of the clusters in its finer partition $X_{\pi}$ given the mapping $X_{\pi}\rightarrow\hat{X}_\pi$. Without loss of generality, we specify one map in presenting the following result.}

\begin{coro}
\label{coro:reproduction_subpop}
Consider the cluster $\hat{\chi}_{o}\in\hat{X}_{\pi}$, $o\in\underline{\hat{m}}$, where $\hat{\chi}_{o}=\bigcup^{m'}_{q=1} \chi_q$, where $2\leq m'<m$. 
Under the condition that~$x(t)\gg0$,
    \textcolor{black}{the CERN} of $\hat{\chi}_{o}$  
    satisfies
\begin{align}
\label{Eq:Total_Clustered_Reproduction}
\bar{R}^t_{\hat{\chi}_o}=
\sum_{q=1}^{m'}\frac{\sum_{i\in\chi_q}\gamma_ix_i(t)\bar{R}^t_{\chi_q}}{\sum_{i\in \hat{\chi}_o}\gamma_ix_i(t)}.
    \end{align} 
\end{coro}

\begin{proof}
Based on~\eqref{Eq:Total_Clustered_Reproduction} and the proof of~Theorem~\ref{Thm:Gen_Reprod}, we have that
\begin{align*}
\bar{R}^t_{\hat{\chi}_o}&=    \sum_{q=1}^{\hat{m}}\frac{\sum_{i\in\chi_q}\gamma_ix_i(t)}{\sum_{i\in \hat{\chi}_o}\gamma_ix_i(t)} \frac{\sum_{i\in\chi_q} \gamma_ix_i(t) \bar{R}^t_{i}}{\sum_{i\in\chi_q}\gamma_ix_i(t)}\\
&=\sum_{q=1}^{\hat{m}}\frac{\sum_{i\in\chi_q} \gamma_ix_i(t) \bar{R}^t_{i}}{\sum_{i\in \hat{\chi}_o}\gamma_ix_i(t)} = \frac{\sum_{i\in \hat{\chi}_o} \gamma_ix_i(t) \bar{R}^t_{i}}{\sum_{i\in \hat{\chi}_o}\gamma_ix_i(t)}, 
\end{align*}
which is the definition of~$\bar{R}^t_{\hat{\chi}_o}$ given in~Definition~\ref{Def:Generalized_Rt}. 
\end{proof}
\textcolor{black}{
Corollary~\ref{coro:reproduction_subpop} provides a method for using the \textcolor{black}{CERNs} of \baike{the clusters of a finer partition~$X_{\pi}$} to calculate the \textcolor{black}{CERNs} of \baike{the clusters of the corresponding coarser partition~$\hat{X}_{\pi}$}. This calculation enables aggregation at lower cluster levels, allowing them to report their CERNs to higher cluster levels, forming a hierarchical structure.
Without loss of generality, we only consider one level of cluster aggregation in this work, performed by the central aggregator of the cluster (Figure~\ref{fig:network_structure}). However, Corollary~\ref{coro:reproduction_subpop} provides a foundation to generalize the framework in Figure~\ref{fig:network_structure} to include additional intermediate aggregators for clusters at different scales.}

Similarly to how we construct \textcolor{black}{local} distributed ERN matrices to describe connections between individual entities using local distributed ERNs, we define the cluster distributed effective reproduction numbers via the cluster distributed effective reproduction number matrix to capture interactions between clusters.
\baike{Again, we consider
$X_\pi=\{\chi_1, \dots, \chi_m\}$, where each cluster~$\chi_q$ with $q \in \underline{m}$ represents a group of $|\chi_q|$ entities, 
where $p,q\in\underline{m}$, such that $\bigcup_{q=1}^m\chi_q=V$ and $\chi_p\cap\chi_q = \varnothing$ for all $p\neq q$.}

\begin{definition}[\textcolor{black}{Cluster} Distributed Effective Reproduction Number Matrix]
\label{Def:Rep_Mat_Cluster}
 The \textcolor{black}{cluster}  distributed effective reproduction number matrix 
 is 
\begin{gather}
\label{eq_cluster_dis_matrix}
\mathcal{R}_X^t = 
 \begin{bmatrix}
    \bar{R}^t_{\chi_1,\chi_1} & \bar{R}^t_{\chi_1,\chi_2} & \cdots & \bar{R}^t_{\chi_1,\chi_m}\\
 \bar{R}^t_{\chi_2,\chi_1} & \bar{R}^t_{\chi_2,\chi_2} & \cdots & \bar{R}^t_{\chi_2,\chi_m}\\ 
         \vdots & \vdots & \ddots & \vdots\\ 
         \bar{R}^t_{\chi_m,\chi_1} & \bar{R}^t_{\chi_m,\chi_1} & \cdots & \bar{R}^t_{\chi_m,\chi_m} 
     \end{bmatrix},  
\end{gather}
where \textcolor{black}{we define}
\begin{equation}
\label{Eq_Distributed_R_Clusters}
\bar{R}^t_{\chi_q,\chi_r}=\frac{\sum_{i\in\chi_q}(\gamma_ix_i(t)\sum_{j\in\chi_r}\bar{R}_{ij}^t)}{\sum_{i\in\chi_q}\gamma_ix_i(t)}
\end{equation} 
\textcolor{black}{as the cluster distributed effective reproduction number from cluster $\chi_r$ to cluster $\chi_q$, for all $r,q\in \underline{m}$. We further define 
$\bar{R}^t_{\chi_q,\chi_q}$ for all $q\in\underline{m}$ as the cluster distributed effective reproduction number of from the cluster $\chi_q$ to itself.}
\end{definition}

Definition~\ref{Def:Rep_Mat_Cluster} addresses Problem~\ref{prob:3} by modeling the aggregated interactions between any pair of the clusters~$\chi_q$ and~$\chi_r$ through the cluster distributed ERN $\bar{R}^t_{\chi_q,\chi_r}$,
for all $q,r\in \underline{m}$. \baike{We use Figure~\ref{fig:Cluster_dist_eff_rep_num} to illustrate the connection between the local and cluster distributed ERNs at different scales.}
Similar to the construction of the \textcolor{black}{local}  distributed ERN matrix $\bar{\mathcal{R}}^t$, 
the off-diagonal entries of the $q^{th}$ row of $\mathcal{R}_X^t$ are defined as the 
\textcolor{black}{cluster} 
exogenous effective reproduction numbers of the $q^{th}$ cluster,  defined as $\bar{R}^t_{\chi_q,\chi_r}$, while 
the diagonal entry of the $q^{th}$ row of $\mathcal{R}_X^t$ is defined as the \textcolor{black}{cluster}  endogenous effective reproduction number of the $q^{th}$ cluster, defined as $\bar{R}^t_{\chi_q,\chi_q}$. Thus, $\bar{R}^t_{{\chi_q},\chi_q}$ represents the endogenous transmission within the $q^{th}$ cluster, while  $\bar{R}^t_{\chi_q,\chi_r}$ captures the exogenous
transmissions from the $r^{th}$ cluster to the $q^{th}$ cluster, where $q, r \in \underline{m}$. 
Together, the  $q^{th}$ row of $\mathcal{R}_X^t$ includes the cluster distributed effective reproduction numbers of the cluster $\chi_q$, which models the 
endogenous and exogenous spreading behavior of the $q^{th}$ cluster in the network. We name the 
the  $q^{th}$ row of $\mathcal{R}_X^t$, denoted by $\mathcal{R}_{X_{\chi_q,:}}^t$, as the \textit{\textbf{cluster distributed effective reproduction number vector}} of $q^{th}$ cluster, for all $q\in\underline{m}$. 
We further explain the entries of the matrix~$\mathcal{R}^t_X$ through the following corollary.

\begin{coro} 
\label{Coro_def_ditri_effe_repro_cluster_2}
For~$q\in\underline{m}$, 
the  \textcolor{black}{CERN}  of cluster $\chi_q$ given in
Definition~\ref{Def:Generalized_Rt} can also be obtained by 
$\bar{R}^t_{\chi_q} = \sum_{r=1}^m \bar{R}^t_{{\chi_q},\chi_r}$,
for all $r\in\underline{m}$.
\end{coro}
\begin{proof}
Based on~\eqref{eq:SIS}, the change in the infected proportions of the cluster $\chi_q$  can be viewed as the 
sum of (i) the 
infections generated by the infected cases within the cluster (endogenous transmissions) and 
(ii) the infections generated by the infected cases from other clusters (exogenous transmissions). 
Therefore, we have that
\begingroup
\allowdisplaybreaks
    \begin{align*}
        \sum_{i\in \chi_q} \dot x_i(t) &= -\sum_{i\in \chi_q} \gamma_ix_i(t)+\underbrace{\sum_{i\in \chi_q}\sum_{j\in \chi_q}s_i\beta_{ij}x_j(t)}_{\textnormal{endogenous transmissions}} \\
        &\qquad+
        \underbrace{\sum_{r\in\underline{m}, r\neq q}\sum_{i\in \chi_q}\sum_{j\in\chi_r}s_i\beta_{ij}x_j(t)}_{\textnormal{exogenous transmissions}}.
    \end{align*}
\endgroup
Then, if $\sum_{i\in \chi_i} \dot x_i(t) =0$, we have that 
\begin{align*}
  &\frac{\sum_{i\in \chi_q}\sum_{j\in \chi_q}s_i\beta_{ij}x_j(t)}{\sum_{i\in \chi_q} \gamma_ix_i(t)} +\sum_{r\in\underline{m}, r\neq q}\frac{\sum_{i\in \chi_q}\sum_{j\in\chi_r}s_i\beta_{ij}x_j(t)}{\sum_{i\in \chi_q} \gamma_ix_i(t)}\\
 &=\frac{\sum_{i\in \chi_q}\gamma_ix_i(t)\frac{\sum_{j\in \chi_q}s_i\beta_{ij}x_j(t)}{\gamma_ix_i(t)}}{\sum_{i\in \chi_q} \gamma_ix_i(t)} \\
  & \ \ \ \ \ \ \ \ \ \ \ \ \ \ \ \ \ +
 \sum_{r\in\underline{m}, r\neq q}\frac{\sum_{i\in \chi_q}
 \gamma_ix_i(t)
 \frac{\sum_{j\in\chi_r}s_i\beta_{ij}x_j(t)}{\gamma_ix_i(t)}}{\sum_{i\in \chi_q} \gamma_ix_i(t)}\\
  &= \frac{\sum_{i\in \chi_q}\gamma_ix_i(t)\sum_{j\in \chi_q}\bar{R}^t_{ij}}{\sum_{i\in \chi_q} \gamma_ix_i(t)} \\
 & \ \ \ \ \ \ \ \ \ \ \ \ \ \ \ \ \ +
\sum_{r\in\underline{m}, r\neq q} \frac{\sum_{i\in \chi_q}
 \gamma_ix_i(t) \sum_{j\in \chi_r}
 \bar{R}^t_{ij}}{\sum_{i\in \chi_q} \gamma_ix_i(t)}\\
 \\
 &=\bar{R}^t_{\chi_q,\chi_q} + \sum_{r\in\underline{m}, r\neq q} \bar{R}^t_{\chi_q,\chi_r} = 1.
\end{align*}
By comparing the equation above with Definition~\ref{Def:Generalized_Rt}  and the proof of Theorem~\ref{Thm:Gen_Reprod}, we see that the \textcolor{black}{CERN}  of cluster $\chi_q$
can be computed by summing the
entries in the $q^{th}$ row of the \textcolor{black}{cluster} distributed ERN matrix~$\mathcal{R}^t_X$ in~\eqref{eq_cluster_dis_matrix}, i.e., 
the cluster distributed
ERN vector~$\mathcal{R}^t_{X_{\chi_q,:}}$.
\end{proof}


\begin{figure}[h]
    \centering
    \includegraphics[width=1\linewidth]{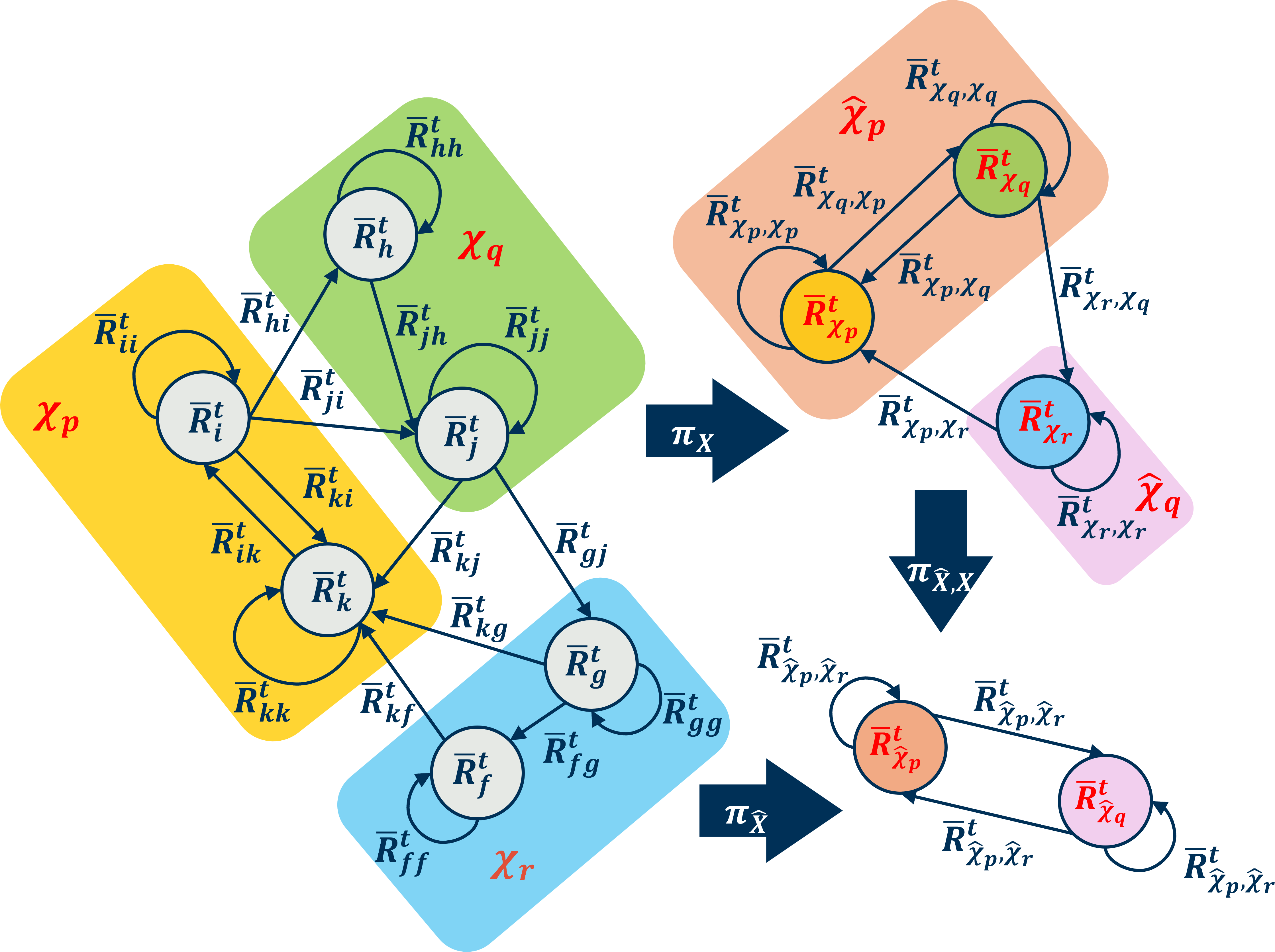}
    \caption{\baike{Cluster distributed effective reproduction numbers. The network on the left depicts a spreading network with six nodes $f$, $g$, $h$, $i$, $j$, and $k$. The network can be partitioned via a finer partition $\pi_{X}$ such that $X_{\pi}=\{ \chi_{p}, \chi_{q}, \chi_{r} \}$ and its coarser partition $\pi_{\hat{X}}$ such that $\hat{X}_{\pi} = \{ \hat{\chi}_{p}, \hat{\chi}_{q} \}$, as depicted in the figures on the top right and bottom right, respectively. 
      For example, we can use the local distributed ERNs and LERNs of entities $i$ and $k$ to compute the cluster distributed ERNs and CERNs of the cluster
    $\chi_{p}=\{i,k\}$ in $X_{\pi}$. Additionally, we can use the local distributed ERNs and LERNs of entities $h$, $i$, $j$, and $k$ to compute the cluster distributed ERNs and CERNs of
    the cluster $\hat{\chi}_{p}=\{h,i,j,k\}$ in $\hat{X}_{\pi}$. 
    Further, according to Corollary~\ref{coro:reproduction_subpop}, we can also aggregate the CERNs in the figure on the top right, such as $\bar{R}^t_{\chi_{p}}$ and $\bar{R}^t_{\chi_{q}}$
   to obtain the CERNs of the corresponding cluster in the coarser partition, given by  $\bar{R}^t_{\hat{\chi}_{p}}$, in the figure on the bottom right. Note that we have that $\pi_{\hat{\chi},\chi}: X\rightarrow \hat{X}$.
   }}
    \label{fig:Cluster_dist_eff_rep_num}
\end{figure}

In the local distributed ERN matrix $\bar{R}^t$, the sum of the $i^{th}$ row (i.e., the local distributed ERN vector) corresponds to the LERN of entity $i$, as shown in~\eqref{Net_R}. Similarly, Corollary~\ref{Coro_def_ditri_effe_repro_cluster_2} establishes that the sum of the $q^{th}$ row (i.e., the cluster distributed ERN vector)
in the cluster distributed ERN matrix, defined in Definition~\ref{Def:Rep_Mat_Cluster}, represents the CERN of cluster $\chi_q$. 
Definition~\ref{Def:Rep_Mat_Cluster}, along with Corollary~\ref{Coro_def_ditri_effe_repro_cluster_2}, provides a method to aggregate local distributed ERNs into cluster distributed ERNs, which further characterize spreading interactions at the cluster level. Hence, they address Problem~\ref{prob:3}.

Further, according to Definition~\ref{Def:Generalized_Rt} and Corollary~\ref{coro:reproduction_subpop}, the central aggregator of cluster $\chi_q$ 
can request all relevant infection information from its local authorities.
This information includes local distributed ERNs, recovery rates, and susceptible proportions of the entities within cluster $\chi_q$, 
which can be used to construct its cluster distributed ERN vector $\mathcal{R}^t_{X_{\chi_q,:}}$. 
Additionally, similar to the local distributed ERN vector, the cluster distributed ERN vector
reflects the coupling strength between clusters to some extent. This information raises privacy concerns about revealing sensitive information, such as the frequency of interactions between clusters. Therefore, we proceed in the next section to propose a differential privacy framework to protect the local and cluster ERNs before sharing them with the policymakers in charge of the central authority.

\section{Differential Privacy for  Distributed Effective Reproduction Numbers}
\label{Privacy_Mechanism_For_The_Distributed_Reproduction_Numbers}

In this section, we solve Problems~\ref{prob:local_randomizer_design},~\ref{prob:central_aggregator_design}, and~\ref{prob:differential_privacy_accuracy_analysis}. In Section~\ref{sec:communication_framework}, we provide an overview of how \textcolor{black}{local} distributed ERNs are aggregated into private cluster distributed ERNs. We briefly mention the differential privacy mechanisms in Section~\ref{sec:communication_framework}, with a more detailed introduction of these mechanisms, including the local randomizer and shuffle model, provided in Section~\ref{sec:local_randomizer_design} and Section~\ref{sec:central_aggregator_design}, respectively.
Finally, Section~\ref{sec:differential_privacy_accuracy_analysis} analyzes the privacy-accuracy trade-off between the distributed ERNs and the privatized distributed ERNs using the proposed privacy-preserving communication framework.

\begin{assumption}
\label{assum_private_info}
For the spreading dynamics over the graph $G$, the susceptible and infected state vectors $s(t)$ and $x(t)$, as well as the recovery matrix $\Gamma$ are publicly available. The local authority of entity $i$ only has access to the transmission rates related to itself, given by the $i^{th}$ row of the transmission matrix, i.e., $B_{i,:}$, for all $i\in\underline{n}$.
\end{assumption}

\begin{remark}
    \textcolor{black}{Assumption~\ref{assum_private_info} further clarifies the sensitive information we aim to protect, specifically the transmission rates between entities. It is reasonable to suppose that $s(t)$ and $x(t)$ are publicly available and non-sentitive, since during pandemics the number of infected cases and the total population of regions are often shared openly. 
    Furthermore, the recovery rate of each entity, represented by the recovery matrix $\Gamma$, can be determined from the average duration of the infection window, which is typically provided by public health authorities.}
    
    \textcolor{black}{
   However, transmission rates across the network, which are often proportional to the duration or frequency of interactions between entities, are sensitive information. It is reasonable to assume that the local authority of entity~$i$ only has access to the transmission rates related to itself. For example, as an individual, one could use sensors or a smartphone to record the duration of interactions with others during a fixed time window, which would remain unknown to others. As a region, a local authority could record the population flow within and into the region via transportation hubs.} 

    \textcolor{black}{
    According to Definition~\ref{Def:NRN}, the local authority of entity~$i$ can use $s(t)$, $x(t)$, $\Gamma$, and $B_{i,:}$ to compute its local distributed ERNs.
    Based on~\eqref{Eq_Distributed_R_Clusters}, the central authority of the cluster can use $s(t)$, $x(t)$, $\Gamma$, and the local distributed ERNs of its entities to further compute its cluster distributed ERN. 
    Therefore, the local and cluster distributed ERNs and their aggregations also reveal $B_{i,:}$. Thus, we design a privacy framework to protect the local distributed ERNs and then to protect the cluster distributed ERNs.}  
\end{remark}


\subsection{A Privacy-Preserving Communication Framework}\label{sec:communication_framework}
As illustrated in Figure~\ref{fig:network_structure}, consider a spreading network with $n$ entities, each corresponding to a local authority. 
These $n$ entities are divided into $m$ clusters, where each cluster is a non-empty set with a central aggregator.
The central authority of the network is interested in investigating the spreading dynamics using the cluster distributed ERN matrix $\mathcal{R}_X^t$ between $m$ clusters ($m\leq n$).  Meanwhile, we will implement differential privacy mechanisms   to protect the original $\mathcal{R}_X^t$.
We first propose Algorithm~\ref{AL_1} to provide an overview of our privacy framework. 

\begin{algorithm}[h] 
\caption{The Private Cluster Distributed ERN Matrix}
\textbf{Inputs:} 
Local distributed ERNs
\\
\textbf{Output:} Private cluster distributed ERN matrix
\begin{algorithmic}[1] 
\State \textbf{Step 1:} The central authority sends a request to all local authorities.
\State \textbf{Step 2:} Each local authority computes its local distributed ERN vector ($\bar{\mathcal{R}}^t_{i,:}$ in ~\eqref{eq:bar_R_t}).
\State \textbf{Step 3:} Each local authority computes its local aggregated ERN vector ($\zeta_i$ in~\eqref{eq_aggregated_reproduction_number_vector}).
\State \textbf{Step 4:} Each local authority applies a local randomizer (Section~\ref{sec:local_randomizer_design}) to add differential privacy to the local aggregated ERN vector ($\tilde{\zeta}_i$ in~\eqref{eq_priv_ERN_Vector}).
\State \textbf{Step 5:} Each local authority sends its privatized local aggregated ERN vector to the corresponding shuffler of its cluster, which then shuffles the vectors  (Section~\ref{sec:central_aggregator_design}) 
\State \textbf{Step 6:} The central aggregator of the cluster receives shuffled privatized vectors~$\tilde{\zeta}_i$ (Section~\ref{sec:central_aggregator_design}) and aggregates them
to generate the cluster distributed ERN vector in~\eqref{eq:assemble_central_aggregation}.
\State \textbf{Step 7:} All central aggregators send their cluster distributed ERN vectors to the data center, and the data center generates the cluster distributed ERN matrix and shares it with the central authority.
\end{algorithmic}
\label{AL_1}
\end{algorithm}

\textcolor{black}{We explain Algorithm~\ref{AL_1} step-by-step.}
As given by Step~$1$ in Algorithm~\ref{AL_1},
the central authority sends a request to all local authorities, including the number of clusters in the network, denoted by $\regionNum$, and the identities of the entities in each cluster.
After receiving the request, in Step~$2$, 
the $i^{th}$ local authority, with $i\in\underline{n}$ in cluster $\chi_q$ and $q\in\underline{m}$,  uses its locally stored information given in Assumption~\ref{assum_private_info} to compute the \textcolor{black}{local} distributed ERNs $\Bar{R}_{ij}^t$ (Definition~\ref{Def:NRN}), for all $j\in\underline{n}$. Then, the~$i^{th}$ local authority forms its 
local distributed ERN vector~$\bar{\mathcal{R}}^t_{i,:}$ in~\eqref{eq:bar_R_t}.

In Step~$3$, the~$i^{th}$ local authority computes the following \textbf{\textit{local aggregated effective reproduction number}} using the cluster identity information for the overall network provided by the central authority:
\begin{equation}
    \Bar{R}_{i,{\chi_r}}^t = \gamma_ix_i(t)\sum_{k\in\chi_r}\Bar{R}_{ik}^t, \text{ for all } i\in\chi_q.\label{eq:pre_aggregation}
\end{equation}
By summing $\Bar{R}_{i,{\chi_r}}^t$ over all $i\in\chi_q$, the local authority can recover the numerator of the $q^{th}r^{th}$ entry of the \textcolor{black}{cluster} distributed ERN matrix in~\eqref{Eq_Distributed_R_Clusters}. 
Using publicly available information in Assumption~\ref{assum_private_info}, such as the total infected proportion within each entity and the recovery rate of each entity, the local authority can compute the 
$q^{th}r^{th}$ entry of the \textcolor{black}{cluster} distributed ERN matrix, given by~$\mathcal{R}^t_{X}$ in~\eqref{Eq_Distributed_R_Clusters}.
Furthermore, the $i^{th}$ local authority stacks the \textcolor{black}{local} aggregated ERNs with respect to all clusters into a \textbf{\textit{\textcolor{black}{local} aggregated effective reproduction number vector}}, given by
\begin{align}
\label{eq_aggregated_reproduction_number_vector}
    \zeta_i = [\Bar{R}_{i,{\chi_1}}^t,\dots, \Bar{R}_{i,{\chi_m}}^t]\in\mathbb{R}^m,
\end{align}
with the $r^{th}$ entry of $\zeta_i$ being $\Bar{R}_{i,\chi_r}^t$ in~\eqref{eq:pre_aggregation}, for all $r \in \underline{m}$ and $i \in \underline{n}$.
The reason for performing this pre-aggregation step is to reduce the length of the vectors needed for privatization for local authorities and to improve accuracy. 

In Step~$4$, the local authority of entity $i$ implements a local randomizer 
(to be detailed in Section~\ref{sec:local_randomizer_design})
to privatize the local aggregated ERN vector $\zeta_i \in \mathbb{R}^{\regionNum}$, for all $i \in \underline{n}$. 
We denote the private local aggregated ERN vector of entity $i$ as
\begin{align}
\label{eq_priv_ERN_Vector}
    \tilde{\zeta}_i = [\tilde{R}_{i,\chi_1}^t,\dots, \tilde{R}_{i,\chi_m}^t]\in\mathbb{R}^m,
\end{align}
for all $i\in\underline n$. Further, $\tilde{R}_{i,\chi_r}^t$ is the private \textcolor{black}{local aggregated} effective reproduction number. 

Then, in Step~$5$, for each~$i \in \underline{n}$,  local authority~$i$ sends $\tilde{\zeta}_i$ to its corresponding trusted shuffler at its own central aggregator of the cluster it belongs to. \textcolor{black}{Recall from the introduction that the central authority of the network has permission only to read the outputs from the central aggregators but does not have permission to inspect the privacy mechanisms implemented by the central aggregators. In addition, in this work, we consider that the central aggregators to be implemented by the authorities of the clusters.}

For a total of $\regionNum$ clusters, there are $\regionNum$ shufflers implemented at the central aggregators,
each of which 
is responsible for anonymizing the private
local aggregated ERN vectors $\tilde{\zeta}_i$ reported to it,
as well as applying a random permutation to remove any information in their order.
The detailed shuffling mechanism will be introduced in Section~\ref{sec:central_aggregator_design}.


Based on~Definition~\ref{Def:Rep_Mat_Cluster}, 
in Step~$6$,
the central aggregator of the cluster $\chi_q$, 
for all $q\in\underline{m}$, 
leverages its own private aggregated  \textcolor{black}{local} ERN \textcolor{black}{vectors}, i.e., $\tilde{\zeta}_i$ for all $i\in\chi_q$, to construct the \textit{\textbf{private \textcolor{black}{cluster} distributed effective reproduction number vector}}, denoted as $\tilde{\mathcal{R}}_{X_{\chi_q,:}}^t$, where
\begin{align}
\tilde{\mathcal{R}}_{X_{\chi_q,:}}^t = [\tilde{R}_{\chi_q,\chi_1}^t,\dots, \tilde{R}_{\chi_q,\chi_m}^t],    
\end{align}
and where
\begin{equation}
    \tilde{R}_{\chi_q,\chi_r}^t = \frac{\sum_{k\in\chi_q}\tilde{R}_{k,\chi_r}^t}{\sum_{k\in\chi_q}\gamma_kx_k(t)}\label{eq:assemble_central_aggregation}
\end{equation}
is defined as the \textbf{\textit{private cluster distributed effective reproduction number}}
for all ${q,r\in\underline{m}}$.
Note that $\tilde{\mathcal{R}}_{X_{\chi_q,:}}^t$ is the privatized version of the $q^{th}$ row of the cluster distributed ERN matrix $\mathcal{R}^t_X$ in~\eqref{eq_cluster_dis_matrix}, i.e., 
the cluster distributed ERN vector~$\mathcal{R}^t_{X_{\chi_q,:}}$. In addition, we did not 
take any additional steps (such as adding noise) to privatize
the cluster distributed ERN $\tilde{R}_{\chi_q,\chi_r}^t$. According to~\eqref{eq:assemble_central_aggregation}, the private cluster distributed ERN of cluster $\chi_q$, for all $q\in\underline{m}$, given by
$\tilde{R}_{\chi_q,\chi_r}^t$, for all $r\in\underline{m}$,  
can be computed from the  entries of the private local aggregated ERN
vectors $\tilde{\zeta}_i$ in~\eqref{eq_priv_ERN_Vector}, for all $i\in\chi_q$. Then, all central aggregators at their clusters send their own private cluster distributed ERN vectors to the central authority. Consequently, the central authority will form the \textit{\textbf{private cluster distributed effective reproduction number matrix}} $\tilde{\mathcal{R}}_X^t$, where $[\tilde{\mathcal{R}}_X^t]_{q,r} = \tilde{R}^t_{\chi_q,\chi_r}$, for all $q,r\in\underline{m}$. 
The private \textcolor{black}{cluster} distributed ERN matrix  is a privatized version of the \textcolor{black}{cluster}
distributed ERN matrix~$\mathcal{R}_X^t$
in Definition~\ref{Def:Rep_Mat_Cluster},  which captures the transmission coupling through the private entries $\tilde{R}_{\chi_q,\chi_r}^t$
between the clusters
 $\chi_{q}$ and~$\chi_r$.

\begin{remark}

Based on Assumption~\ref{assum_private_info}, the proposed communication framework in Algorithm~\ref{AL_1} is a distributed, privacy-enhanced method for obtaining the \textcolor{black}{private cluster} distributed ERN matrix whose entries are defined in~\eqref{eq:assemble_central_aggregation}. We discuss in the following sections how the framework can conceal the local and cluster distributed ERNs by using differential privacy. 
If we ignore the privatization step in Algorithm~\ref{AL_1}, then we will recover the \textcolor{black}{cluster} distributed ERN matrix in Definition~\ref{Def:Rep_Mat_Cluster}. We use Figure~\ref{fig:communication_network} as an example to further illustrate Algorithm~\ref{AL_1}, which defines a communication framework for the central authority of a spreading network to retrieve its private cluster distributed ERN matrix $\tilde{\mathcal{R}}_X^t$.

\end{remark}

\begin{figure*}[ht!]
    \centering
    \includegraphics[width=2\columnwidth]{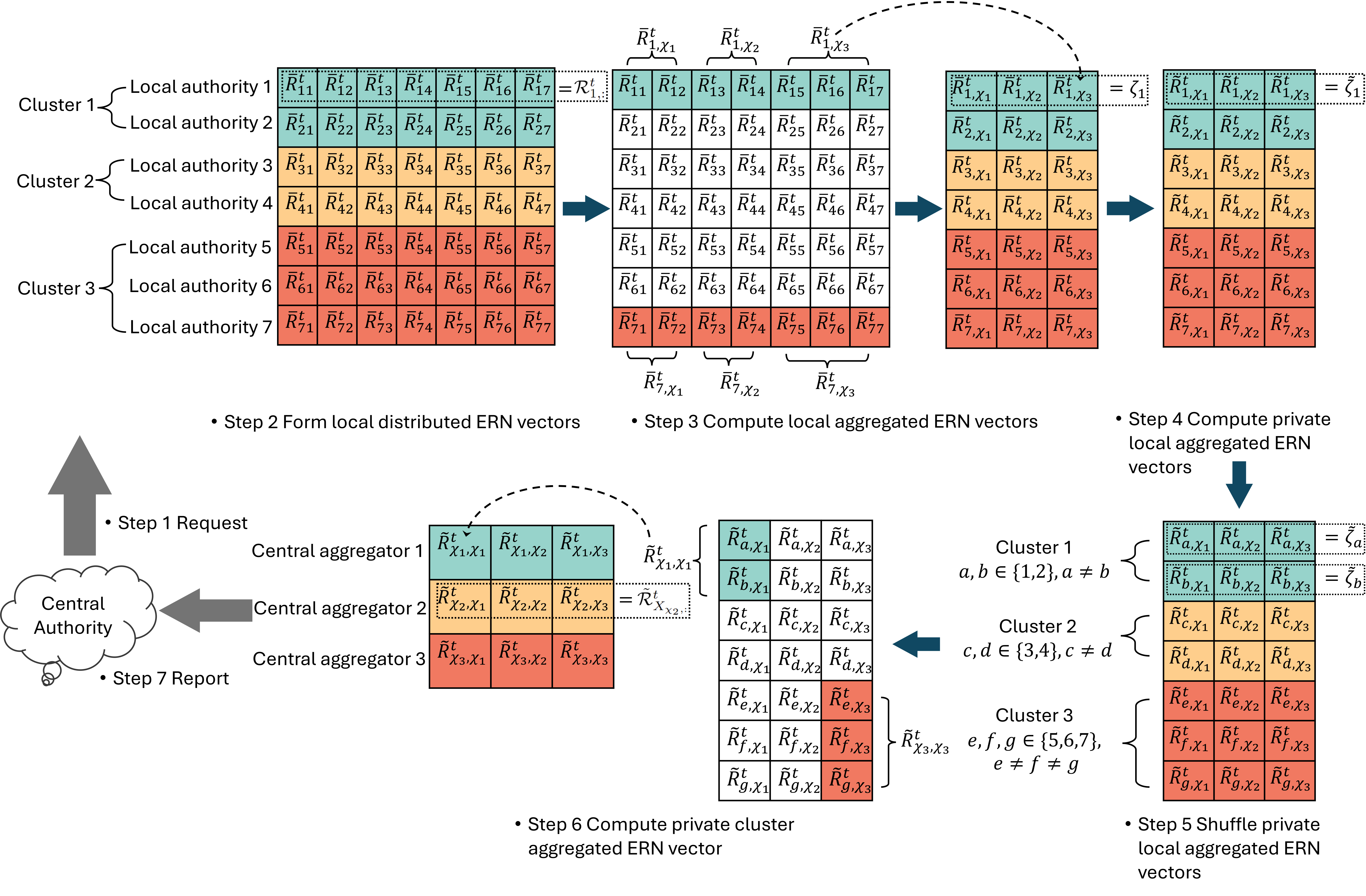}
    \caption{\textcolor{black}{Construction of the private cluster distributed ERN matrix of a disease spreading network with seven entities. The seven entities are managed by their own local authorities, labeled from $1$ to $7$. In addition, these seven authorities are organized into three clusters: $\chi_1 = \{1,2\}$, $\chi_2 = \{3,4\}$, and $\chi_3 = \{5,6,7\}$. We use three different colors to mark the entries that belong to these three clusters.
    This example illustrates the seven steps in Algorithm~\ref{AL_1}. 
    1) The central authority sends a request to the local authorities of the seven entities in the network, along with necessary information such as the identities of the other entities, the clusters they belong to, and the information in Assumption~\ref{assum_private_info}. 
    2) The local authority of entity $i$ computes its local distributed ERNs $\bar{R}^t_{ij}$ given in Definition~\ref{Def:NRN} for all $j$. Then, the local authority of entity $i$ forms its local distributed ERN vector $\bar{\mathcal{R}}^t_{i,:}$ in \eqref{eq:bar_R_t}. For example, $\bar{\mathcal{R}}^t_{1,:}=[\bar{R}^t_{11}, \dots, \bar{R}^t_{17}]$ formed by local authority~$1$ is given by the first row of the matrix shown in Step~$2$. 
    3) Local authority~$i$ computes its local aggregated ERN vector defined in~\eqref{eq_aggregated_reproduction_number_vector}. For example, we have $\zeta_1 = [\bar{R}^t_{1,\chi_1}, \bar{R}^t_{1,\chi_2}, \bar{R}^t_{1,\chi_3}]$. As illustrated by the aggregation process in the first row of the matrix in Step~$3$, local authority~$i$ only requires its own local distributed ERN vector, thanks to~\eqref{eq:pre_aggregation}. 
    4) The $i^{th}$ local authority implements the local randomizer (as in Section~\ref{sec:local_randomizer_design}) to its own $\zeta_i$ to obtain the private local aggregated ERN vector $\tilde{\zeta}_i$. For example, the local authority~$1$ privatizes $\zeta_1 = [\bar{R}^t_{1,\chi_1}, \bar{R}^t_{1,\chi_2}, \bar{R}^t_{1,\chi_3}]$ as $\zeta_1 = [\tilde{R}^t_{1,\chi_1}, \tilde{R}^t_{1,\chi_2}, \tilde{R}^t_{1,\chi_3}]$. 
    5) Starting from this step, all processes will be performed by the central aggregators of the clusters. The local authorities send their privatized local aggregated ERN vectors to the shuffler of their corresponding clusters. The shuffler hides the identities of its local aggregated ERN vectors (as in Section~\ref{sec:central_aggregator_design}). For instance, the shuffler of Cluster~$3$ collects the local aggregated ERN vectors from the local authorities of its entities, denoted as $\tilde{\zeta}_i$, $i\in\{5,6,7\}$. Then, the shuffler shuffles these three vectors and generates three anonymous vectors, $\tilde{\zeta}_e$, $\tilde{\zeta}_f$, and $\tilde{\zeta}_g$. 
    6) The central aggregator of cluster $\chi_q$ then uses its anonymous private local aggregated ERN vectors to compute its private cluster distributed ERN vector $\tilde{\mathcal{R}}_{X_{\chi_q,:}}^t$ according to~\eqref{eq:assemble_central_aggregation}. For instance, the central aggregator of the cluster $\chi_1$ can use the first entries of the vectors $\tilde{\zeta}_a$ and $\tilde{\zeta}_b$ to generate the first entry of its private cluster distributed ERN vector $\tilde{R}^t_{\chi_1,\chi_1}$, thanks to~\eqref{eq:assemble_central_aggregation}. Following the same procedure, the central aggregator of the cluster $\chi_1$ can obtain its private cluster distributed ERN vector, given by $\tilde{\mathcal{R}}^t_{X_{\chi_1,:}}$. 
    7) All central aggregators share their private cluster distributed ERN vectors with the central authority.}} 
\label{fig:communication_network}
\end{figure*}

\subsection{Local Randomizer Design}\label{sec:local_randomizer_design}
\textcolor{black}{In this section, we introduce the local randomizer implemented by local authorities in Step~$4$ of Algorithm~\ref{AL_1}.}
Let $\mathbb{D}$ be the domain of the aggregated \textcolor{black}{local} distributed ERNs, as introduced in~\eqref{eq:pre_aggregation}.
\textcolor{black}{Based on its local aggregated ERN vector $\zeta_i=[\vectElement]_{q\in\setRegionNum}\in\domain^\regionNum$, 
the local authority 
$i\in\setLocAuthNum$ uses the local randomizer to generate a private local aggregated ERN vector
$\privateVect_i=[\privateVectElement]_{q\in\setRegionNum}$ by implementing the bounded Gaussian mechanism \textcolor{black}{for $\zeta_i$}. 
Since local authorities seek to share a privatized form of~$\zeta_i$ itself, from a privacy perspective
we must privatize the identity mapping acting on~$\zeta_i$, which is sometimes called
``input perturbation''. 
We will use the bounded Gaussian mechanism to implement 
input perturbation differential privacy, and, to do so,
we calibrate its parameters to the identity mapping. 
}

Towards such an implementation, 
we first define the notion of 
sensitivity, which quantifies the maximum possible difference 
between two adjacent \textcolor{black}{local} aggregated ERN vectors, 
as defined in Definition~\ref{def:adjacency}.   

\begin{definition}[$L_2$-sensitivity of local aggregated ERN vectors]\label{def:l2_sensitivity}
Fix an adjacency parameter~$k > 0$. 
Consider two \textcolor{black}{local aggregated ERN} vectors $\vect_i, \vect_i'\in\nonnegativeRealSpace^\regionNum$ 
at local authority $i$ that are adjacent in the sense of Definition~\ref{def:adjacency}. 
Then the $L_2$-sensitivity of the identity mapping with respect to this adjacency relation is 
$\Delta_2\vect_i$, defined as $\Delta_2\vect_i=\max_{\vect_i\adjacent\vect_i'}\sqrt{\sum_{r=1}^\regionNum (\vectElement-(\vectElement)')^2}$, where $\regionNum=|\vect_i|$ is the number of clusters in the network, and
where each~$\vectElement$ is computed
from~$\vect_i$ and each~$(\vectElement)'$
is computed from~$\vect_i'$~\textcolor{black}{as in~\eqref{eq_aggregated_reproduction_number_vector}}. 
 \hfill $\lozenge$
\end{definition}

\begin{remark} \label{rem:sens}
By using Definitions~\ref{def:adjacency} and~\ref{def:l2_sensitivity}, we can see that the sensitivity
of the identity mapping acting on a local aggregated ERN vector is~$\Delta_2\zeta_i = k$, where~$k$
is the user-specified adjacency parameter. 
\end{remark} 

We use the $L_2$-sensitivity to calibrate the variance of noise that is added for privacy, and we next define the bounded Gaussian mechanism.

\begin{mechanism}[Bounded Gaussian mechanism~\cite{chen2022bounded}]\label{mech:bounded_Gaussian_mechanism}
    Fix a probability space $(\Omega,\mathcal{F},\mathbb{P})$. For a \textcolor{black}{local aggregated ERN} vector $\vect_i=[\vectElement]_{r\in\underline{\regionNum}}\in\domain^\regionNum_i$ with $\domain^\regionNum_i=\prod_{r=1}^\regionNum\domain_{ir}$ and $\domain_{ir}=[\vectElementLowerBound_{ir},\vectElementUpperBound_{ir}]$ at  local authority $i\in\setLocAuthNum$, the bounded Gaussian mechanism $M_{BG}:\domain^n_i\times\Omega\rightarrow\domain^n_i$ generates private \textcolor{black}{local aggregated ERN} vectors $\privateVect_i=[\privateVectElement]_{r\in\underline{\regionNum}}\in\domain^n_i$ with 
    \begin{equation}
        \privateVectElement=\begin{cases}
            TrunG(\vectElement,\sigma,\vectElementLowerBound_{ir},\vectElementUpperBound_{ir}), & \text{if } \vectElement>0, \\
            0, & \text{if } \vectElement=0.
        \end{cases}
    \end{equation}
    
    This mechanism is a local randomizer that satisfies {$\localRandPrivacyLevel$-differential privacy} for $\localRandPrivacyLevel>0$ if
    \begin{equation}
        \sigma^2 \geq \frac{k\left(\frac{k}{2}+\sqrt{\sum_{r=1}^n(\vectElementUpperBound_{ir}-\vectElementLowerBound_{ir})^2\cdot\mathbf{1}_{\mathbb{R}_{> 0}}(\vectElement)}\right)}{\localRandPrivacyLevel-\log(\Delta C(\sigma,\offsetVect))},\label{ineq:sigma}
    \end{equation}
    where 
    \begin{equation*}
        \Delta C(\sigma,\offsetVect) = \prod_{r=1}^n\frac{\Phi\left(\frac{\vectElementUpperBound_{ir}-\vectElementLowerBound_{ir}-\offsetVectElement_{ir}}{\sigma}\right)-\Phi\left(\frac{-\offsetVectElement_{ir}}{\sigma}\right)}{\Phi\left(\frac{\vectElementUpperBound_{ir}-\vectElementLowerBound_{ir}}{\sigma}\right)-\Phi\left(0\right)}\cdot\mathbf{1}_{\mathbb{R}_{> 0}}(\vectElement)
    \end{equation*}
    and $\offsetVect\in\nonnegativeRealSpace^n$ is an offset vector. The vector $\offsetVect$ can be found by solving the optimization problem in~\cite[(3.3)]{chen2022bounded}.\hfill $\lozenge$
\end{mechanism}

\begin{remark}\label{rmk:implication_epsilon}
    The minimal value of $\sigma$ that satisfies~\eqref{ineq:sigma} can be found using~\cite[Algorithm 2]{chen2022bounded}. Note that~\eqref{ineq:sigma} implies that a larger $\epsilon$ gives weaker privacy and leads to a smaller $\sigma$.\hfill $\lozenge$
\end{remark}

Mechanism~\ref{mech:bounded_Gaussian_mechanism} defines a privacy mechanism that generates a private
\textcolor{black}{local aggregated ERN} vector $\privateVect_i=[\privateVectElement]_{r\in\setRegionNum}$ within a bounded domain $\domain^n_i$ around the original \textcolor{black}{local aggregated ERN vector} $\vect_i$. Conventional unbounded mechanisms, e.g., the standard Gaussian and Laplace mechanisms, 
can generate reproduction numbers that are negative or implausibly high~\cite{dwork2014algorithmic}. Conversely, the 
bounded Gaussian mechanism prevents such infeasible values, maintaining both their credibility and usefulness in analysis (See Section~\ref{sec_Simulation}). 

\begin{remark}
 \textcolor{black}{An aggregated local ERN vector $\vectElement = 0$ indicates that there are no direct transmissions from cluster $\chi_r$ to entity $i$. A benefit of using the bounded Gaussian mechanism is that it always preserves these zero vectors, leaving non-existing transmission interactions unchanged.}
    Thus, the bounded Gaussian mechanism does not alter the presence or absence of transmission interactions in the network, though it does alter values of transmissions when they exist in order to implement privacy. 
\end{remark} 

\textcolor{black}{
In summary, following Step~$4$ in Algorithm~\ref{AL_1}, after generating the local aggregated ERN vector $\vect_i = [\vectElement]_{r \in \setRegionNum}$ at local authority~$i$, where~$i \in \underline{n}$, the local authority uses the local randomizer to generate a private version of the local aggregated ERN vector, denoted $\privateVect_i = [\privateVectElement]_{q \in \setRegionNum}$.  This private vector is then shared with the corresponding central aggregator at its cluster.
Once the central aggregator at the $q^{th}$ cluster receives all 
necessary private local aggregated ERN vectors $\privateVect_i$ for all $i \in \chi_q$, 
it groups these vectors to generate its private cluster distributed ERN vector, as introduced in Step~$6$. To further enhance privacy, 
the central aggregator \textcolor{black}{at each cluster} applies a shuffling mechanism before grouping the local aggregated ERN vectors from its own local authorities.}

\subsection{Shuffler Design}\label{sec:central_aggregator_design}
We introduce the shuffle mechanism for privacy in this section.  
There is one shuffler per cluster, and,
within cluster~$q$, 
local authority~$i$ 
sends its \textcolor{black}{private local aggregated ERN} vector $\privateVect_i = [\privateVectElement]_{r \in \setRegionNum}$ to the shuffler at the central aggregator. 
The shuffler at cluster~$q$
anonymizes all of 
these vectors, and randomly shuffles 
their order
before they are aggregated by the central aggregator. 
\textcolor{black}{Therefore, when using~\eqref{eq:assemble_central_aggregation} to aggregate the private local aggregated ERN vectors into private cluster distributed ERN vectors in Step~$6$ of 
Algorithm~\ref{AL_1}, the shuffling mechanism ensures that the aggregation process within the cluster cannot distinguish the identity of the
local authority that produced each of the
private local aggregated ERN vectors. 
}

The shuffle model introduces additional randomness when data is grouped at the central aggregator, offering significantly stronger privacy guarantees compared to directly sharing the private \textcolor{black}{local} aggregated ERN vectors with the central aggregator~\cite{Cheu2019Distributed}. We use the following result to quantify the privacy of the shuffler implemented in Algorithm~\ref{AL_1}.

\begin{lemma}[Shuffle Model] \label{thm:shuffle}
    Fix a cluster~$\chi_q$. 
    For each~$i \in \chi_q$, 
    let $\localRandomizer_i:\Omega_i\times\domain_i^\regionNum\rightarrow\domain_i^\regionNum$ 
    denote the~$\epsilon_0$-differentially private bounded
    Gaussian mechanism at local authority~$i$.     
    Given a collection of sensitive vectors $\{\vect_i\in\domain^n \mid i\in\chi_q\}$,
    let $\centralDP_s:\Omega\times\domain^{n\cdot \vectNumInShuffler_q}\rightarrow\mathbb{D}^n$ 
    (i)~generate private outputs $\{\privateVect_i\in\domain^n \mid i\in\chi_q\}$, where $\privateVect_i=\localRandomizer_i(\vect_i)$ for each $i\in\chi_q$, 
    (ii)~anonymize the vectors, and 
    (iii)~sample a uniform random permutation $\permutation$ over $\chi_q$ and output $\{\privateVect_{\permutation(i)} \mid  i\in\chi_q\}$. Then $\centralDP$ is $(\epsilon,\delta)$-differentially private 
    in the sense of Definition~\ref{def:central_differential_privacy},
    with
    \begin{equation}
    \epsilon\leq\ln\left(1+(e^{\epsilon_0}-1)\left(\frac{4\sqrt{2\ln(4/\delta)}}{\sqrt{(e^{\epsilon_0}+1)|\chi_q|}}+\frac{4}{|\chi_q|}\right)\right)
    \end{equation}
    and any $\delta\in(0,1)$ such that $\epsilon_0\leq \ln\left(\frac{\vectNumInShuffler_q}{8\ln(2/\delta)}-1\right)$. \hfill $\lozenge$
\end{lemma}
\textcolor{black}{We use the shuffler design from~\cite{Cheu2019Distributed} to directly obtain Lemma~\ref{thm:shuffle} by applying our setting to their results. Therefore, we omit its proof.}

Intuitively, any recipient of privatized data, 
including the central aggregator and any downstream recipients, 
is restricted to seeing 
a uniformly random permutation of the private local aggregated ERN vectors received from the local authorities, making it challenging to link any vector back to its sender. 
In addition, the computational mechanism of the private cluster distributed ERN vector  in~\eqref{eq:assemble_central_aggregation} ensures that the random permutation of the private local aggregated ERN vectors does not affect the value produced by the aggregation step. 
\textcolor{black}{As a result, the 
differential
privacy guarantee provided to each $\tilde{\zeta}_i$ 
is amplified by the shuffler in Lemma~\ref{thm:shuffle}. For an adversary to infer information about $\vect_i$, which contains interaction frequencies for the $i^{th}$ entity, they must not only extract information from the noisy version $\privateVect_i$, but also identify $\privateVect_i$ within the shuffled set~$\{\privateVect_{\permutation(i)} \mid i\in\chi_q\}$ of $|\chi_q|$ vectors.} 

\subsection{Accuracy Analysis at the Central Aggregator}\label{sec:differential_privacy_accuracy_analysis}
As introduced in Step~$6$ of 
Algorithm~\ref{AL_1}, once the local private vectors are shuffled at their designated shufflers, the central aggregator will aggregate them into its \textcolor{black}{private cluster distributed ERN vector}~$\tilde{\mathcal{R}}_{X_{\chi_q,:}}^t$, where each entry can be computed using~\eqref{eq:assemble_central_aggregation}.
We next bound the error that privacy introduces into
private cluster distributed ERNs $\tilde{\mathcal{R}}_{\chi_q,\chi_r}^t$ in~\eqref{eq:assemble_central_aggregation}, which will answer Problem~\ref{prob:differential_privacy_accuracy_analysis}. 
Before that, we first introduce the following result on truncated Gaussian random variables.

\begin{lemma}
\cite[Chapter 3]{burkardt2014truncated}\label{thm:truncated_gaussian_moments}
    For each $\tilde{z}\sim\text{TrunG}(\mu,\sigma,l,u) $,  we have
      $\mathbb{E}[\tilde{z}] = \mu + \sigma\cdot\frac{\varphi\left(\alpha\right)-\varphi\left(\beta\right)}{\Phi\left(\beta\right)- \Phi\left(\alpha\right)}$ and 
    \begin{align}
        \text{Var}[\tilde{z}] &= \sigma^2\left[1-\frac{\beta\varphi(\beta)-\alpha\varphi(\alpha)}{\Phi\left(\beta\right)-\Phi\left(\alpha\right)} - \left(\frac{\varphi\left(\alpha\right)-\varphi\left(\beta\right)}{\Phi\left(\beta\right)-\Phi\left(\alpha\right)}\right)^2\right], \label{eq:trunc_gauss_var}
    \end{align}
    where $\alpha=\frac{l-\mu}{\sigma}$, $\beta=\frac{u-\mu}{\sigma}$, and $\Phi(\cdot)$ and $\varphi(\cdot)$ are defined in~Section~\ref{Sec_Pro_BG}.
    \hfill $\blacksquare$
\end{lemma}


\begin{thm}[Accuracy of Private Cluster Distributed Effective Reproduction Numbers]\label{thm:differential_privacy_accuracy}
    The first and second moment of the  
    private cluster distributed effective reproduction numbers~$\tilde{R}_{\chi_q,\chi_r}^t$ 
    in~\eqref{eq:assemble_central_aggregation} are 
    \begin{align*}
\expectation{\privateReproNum_{\chi_q,\chi_r}} &= \sum_{q=1}^m\frac{\sum_{i\in\chi_q}\gamma_ix_i(t)\expectation{\tilde{R}_{i,\chi_r}^t}}{\sum_{i\in{(\cup\chi_q})}\gamma_ix_i(t)},\\
        \variance{\privateReproNum_{\chi_q,\chi_r}} &= \sum_{q=1}^m\frac{\sum_{i\in\chi_q}\gamma_ix_i(t)\variance{\privateReproNum_{i, \chi_r}}}{\sum_{i\in{(\cup\chi_q})}\gamma_ix_i(t)},
    \end{align*}
    where $\expectation{\tilde{R}_{i,\chi_r}^t}$ and $\variance{\privateReproNum_{i, \chi_r}}$ are defined in  Lemma~\ref{thm:truncated_gaussian_moments}.\hfill $\blacksquare$
\end{thm}

\textcolor{black}{We can obtain  Theorem~\ref{thm:differential_privacy_accuracy} by applying the linearity of expectation to~\eqref{eq:assemble_central_aggregation}. The result follows from the fact that~$\tilde{R}_{\chi_q, \chi_r}$   is a bounded Gaussian random variable and from the known expressions for its expectation and variance given in Lemma~\ref{thm:truncated_gaussian_moments}. Therefore, we omit the proof.}
\textcolor{black}{In differential privacy analysis, a larger value of $\epsilon$ provides weaker privacy protections and, as stated in Remark~\ref{rmk:implication_epsilon}, corresponds to a smaller $\sigma$. Note that~\eqref{ineq:sigma} quantifies that the lower bound of $\sigma$ increases if $\epsilon$ increases, and vice versa. Both the expectation and variance in Theorem~\ref{thm:differential_privacy_accuracy} become smaller as $\epsilon$ grows, which agrees with the intuition that weaker privacy
requires lower-variance noise.} 

\begin{remark}
\label{Remark_T_Gau}
\textcolor{black}{
Theorem~\ref{thm:differential_privacy_accuracy} offers a framework for calibrating privacy levels via acceptable error tolerance. Unlike typical differential privacy implementations that use
unbounded, zero-mean noise, 
Theorem~\ref{thm:differential_privacy_accuracy} indicates that the mean of the private cluster distributed ERNs, given by $\mathbb{E}[\tilde{R}_{\chi_q,\chi_r}]$, differs from the value of $\bar{R}_{\chi_q,\chi_r}$.  This deviation is an inherent trade-off when using the bounded mechanism described in Mechanism~\ref{mech:bounded_Gaussian_mechanism}. Nevertheless, the analytical expressions derived in Theorem~\ref{thm:differential_privacy_accuracy} are valuable for assessing the accuracy of the private cluster distributed ERNs.
We use real-world examples in Section~\ref{sec_Simulation} to demonstrate the effectiveness of the private cluster ERNs generated through Algorithm~\ref{AL_1}, as well as to show that the private ERNs can still provide useful information for analyzing the spreading network.}
\end{remark}

\section{Simulation and Application}
In this section, we present a real-world example to demonstrate how distributed reproduction numbers can be used to analyze disease spread across different regions in a disease spreading network in the United States. We then showcase the implementation of the proposed privacy framework to protect transmission interactions between regions, emphasizing how privacy can be ensured while maintaining the utility of distributed reproduction numbers.
\label{sec_Simulation}
\subsection{Data Processing and Local Distributed ERNs}
\label{sec_data_pre}
Consider a disease spreading network that models the spread of COVID-19 across regions in the United States. Mobility data between pairs of regions are used to reflect and compute the transmission coupling strength. 
The mobility data are provided by SafeGraph~\cite{safegraph2021Distancing}, which tracked the location of approximately $20$ million cell phones in the United States from August 9, 2020 to April 20, 2021. The data include information for over $200,000$ Census Block Groups (CBGs). Each cell phone is assigned a primary residence based on where it spends the majority of its time, and daily visits to other locations are recorded. The CBGs are then mapped to their corresponding local authorities, defining the regions under each authority’s jurisdiction. The transmission rates $[\beta_{ij}]$ between these regions are computed using the approach outlined in~\cite{le2022high}, based on the mobility data stored by the local authorities.
Using spatial-temporal data, we consider $1,023$ local authorities across the country, as illustrated by the colored markers in Figure~\ref{fig:clusters_after_first_clustering}. Each authority is responsible for storing and managing the mobility data of individuals within the region under its jurisdiction. 

\begin{figure}[ht]
\begin{subfigure}{1\linewidth}
    \centering \includegraphics[width=1\columnwidth]{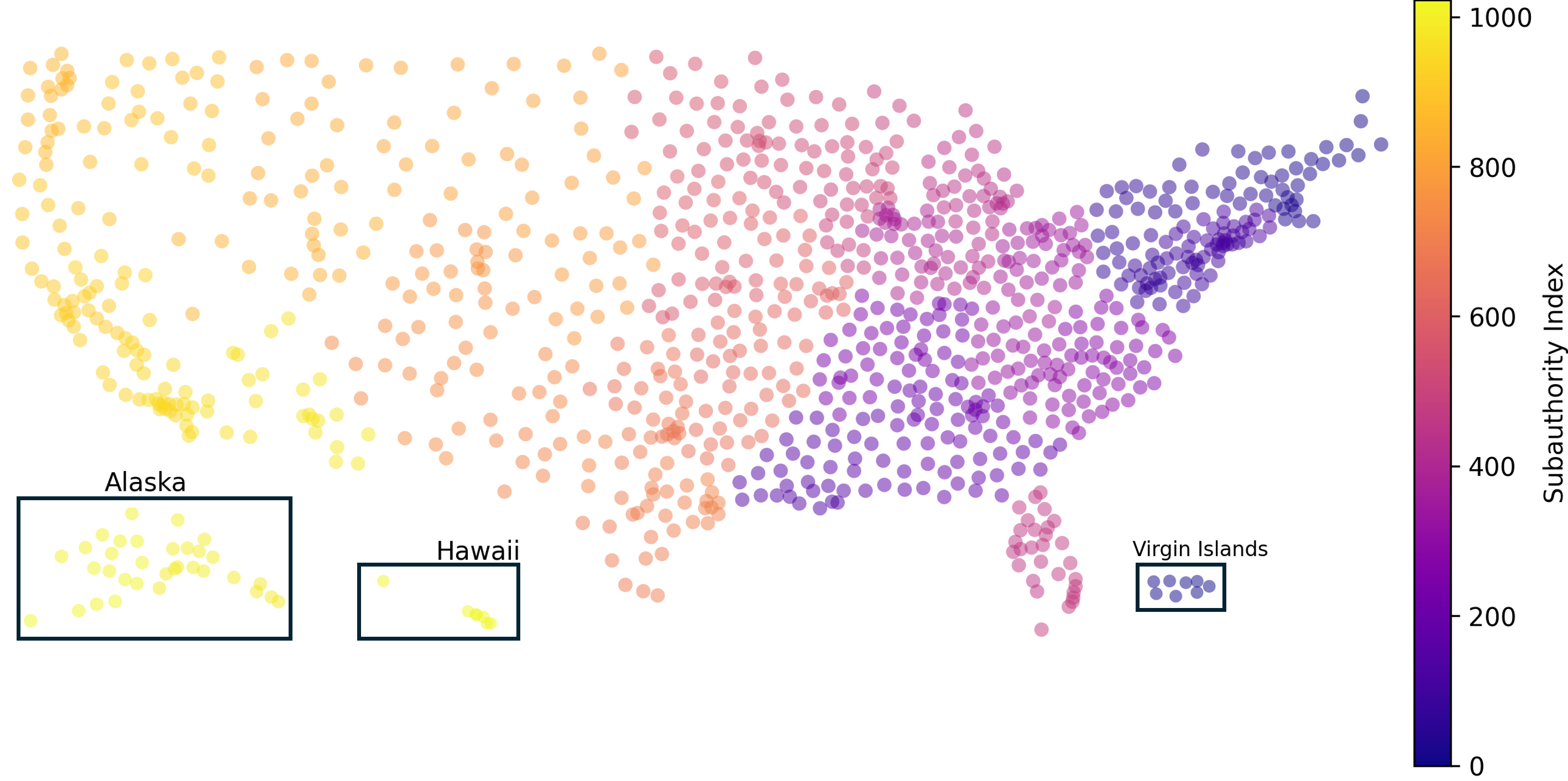}
    \caption{Total of 1,023 local authorities in the USA. Each local authority aggregates a number of Census Block
Groups (CBGs).}
\label{fig:clusters_after_first_clustering}
\end{subfigure}

\begin{subfigure}{1\linewidth}
    \centering \includegraphics[width=1\columnwidth]{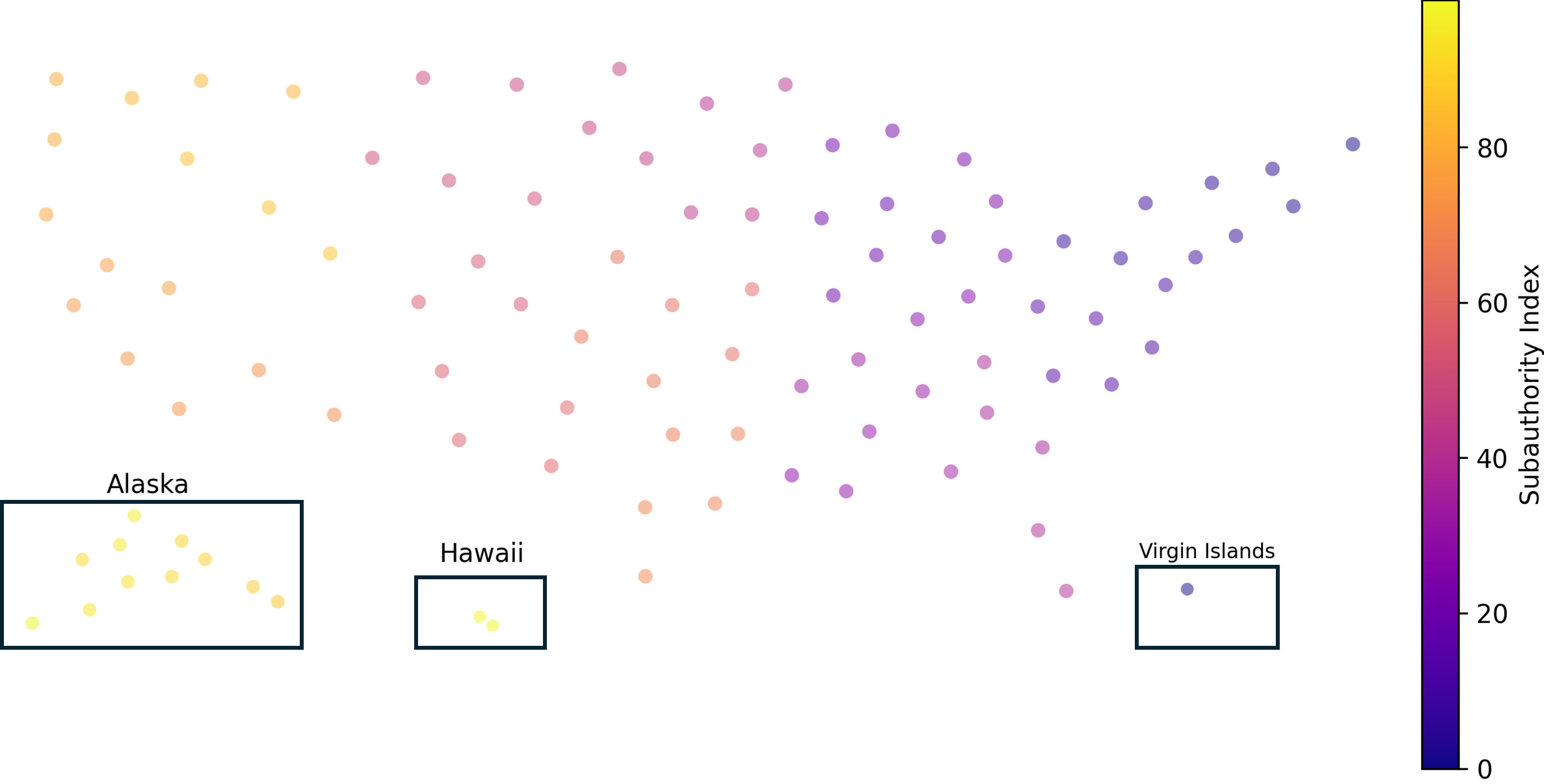}
    \caption{Total of 100 areas in the USA. \baike{The central aggregator of} each area aggregates a number of \baike{local}
    authorities.}
    \label{fig:clusters_after_second_clustering}
\end{subfigure}
\caption{The centroids of the jurisdictions of local authorities and the areas that are of interest.}
\end{figure}

From August 9, 2020, to April 20, 2021, the daily confirmed COVID-19 case numbers were obtained from the Center for Systems Science and Engineering (CSSE) at Johns Hopkins University~\cite{covid192020Dong}. Using these daily confirmed cases, along with total population data, we computed the susceptible and infected proportions of the population under the jurisdiction of each local authority~$i$, denoted as $s_i(t)$ and $x_i(t)$, respectively, for $i \in \{1, \dots, 1023\}$. The total population data were sourced from the SafeGraph Open Census Data~\cite{safegraph2020Census}.

Using these data, the \textcolor{black}{local} distributed ERNs between the jurisdictions of $1,023$ local authorities were computed using Definition~\ref{Def:NRN}, covering the period from August 9, 2020, to April 20, 2021.  
According to Definition~\ref{Def:NRN}, these local distributed ERNs can be organized into a real matrix with $1,023$ rows and $1,023$ columns, where the $i^{th}$  row is computed using only the local data stored by the corresponding $i^{th}$ local authority. This data includes the total population of the region, the susceptible and infected populations of the region, the transmission rates within and into the region, the recovery rate of the region, as well as the infected population and total population from other regions with direct transmission to the region.

These \textcolor{black}{local} distributed ERNs are then post-processed by making some key assumptions: 1) for each local authority $i$, we assume that at least one person is infected when we compute the \textcolor{black}{local} distributed ERN $\bar{R}_{ij}^t$ for each~$j \neq i$, since $\bar{R}_{ij}^t$ in Definition~\ref{Def:NRN} is introduced under the condition that the infected proportion within each entity must be greater than zero. Assuming a non-zero infection level in a global pandemic is a mild assumption.
 2) Each $\bar{R}_{ij}^t$ for all $i$ and $j$ will be projected to a range of $[0, 14]$. \textcolor{black}{Although the effective reproduction number of the COVID-19 pandemic in its early stages was around  $2-4$~\cite{caicedo2020effective}, it is reasonable for $\bar{R}_{ij}$ to be in a wider range. According to the Definition~\ref{Def:NRN}, the ratio of infections also plays a role in determining the \textcolor{black}{local} distributed ERNs between two entities. Therefore, the upper bound on the \textcolor{black}{local} distributed ERNs between different regions can be much higher or lower than the network-level effective reproduction number of the overall spread. In addition,  when $\bar{R}_{ij}$  becomes too large due to the ratio between the infected proportions, we project it to~$14$, establishing an upper limit on the transmission.}

\subsection{Cluster Distributed ERNs and Analyses}
\label{Sec_clus_dis_ERNs}
We consider a central authority in the United States, such as public health officials, elected leaders, and other decision-makers, interested in understanding disease spread across the country. Additionally, we assume that the central authority aims to capture regional trends and patterns at a coarser resolution than the interactions across the $1,023$ regions.
To achieve this, the central authority further groups these $1,023$ local authorities into $100$ regions, as illustrated in Figure~\ref{fig:clusters_after_second_clustering}. Then, each of the $100$ regions has a central aggregator that computes its cluster distributed ERN vector using the local distributed ERN vectors, based on~\eqref{eq_aggregated_reproduction_number_vector} and~\eqref{eq:assemble_central_aggregation}. We first follow the steps in Algorithm~\ref{AL_1} to compute the cluster distributed ERNs without implementing any privacy mechanisms, including the randomizer and the shuffler.

We then select three representative regions from the $100$ regions for analysis: two areas with large populations and one with a small population.
Specifically, we consider the Detroit area in the state of Michigan ($\chi_1$), with a total population of $7,330,520$, the Miami area in the state of Florida ($\chi_2$) with a total population of $8,057,252$, and the Delta Junction area in the state of Alaska ($\chi_3$) with a much smaller total population of $18,898$. 
\textcolor{black}{This selection allows us to study the spreading processes between urban and rural areas, across states with differing public health policies, and within interconnected economic zones.} The infected proportions of the populations in these three areas are shown in Figure~\ref{fig:infected_portion}, and the \textcolor{black}{cluster} distributed ERNs between these regions are illustrated in Figures~\ref{fig:dist_rep_num_DD} and~\ref{fig:dist_rep_num_DM}. 

\begin{figure}[H]
    \centering \includegraphics[width=1\columnwidth]{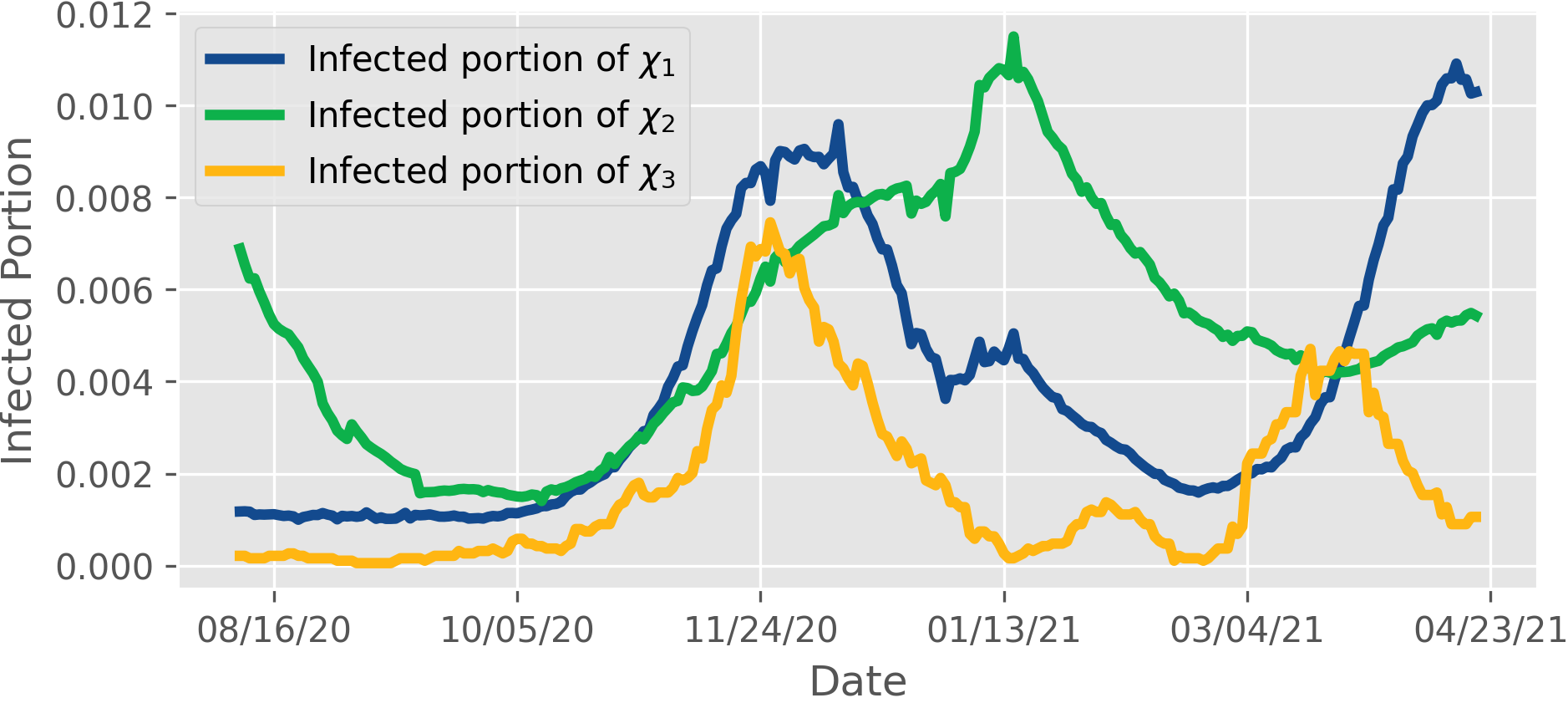}
    \caption{The infected portions of region $\chi_1$ (Detroit area), $\chi_2$ (Miami area), $\chi_3$ (Delta Junction) over time.}
\label{fig:infected_portion}
\end{figure}

\begin{figure}[H]
    \centering \includegraphics[width=1\columnwidth]{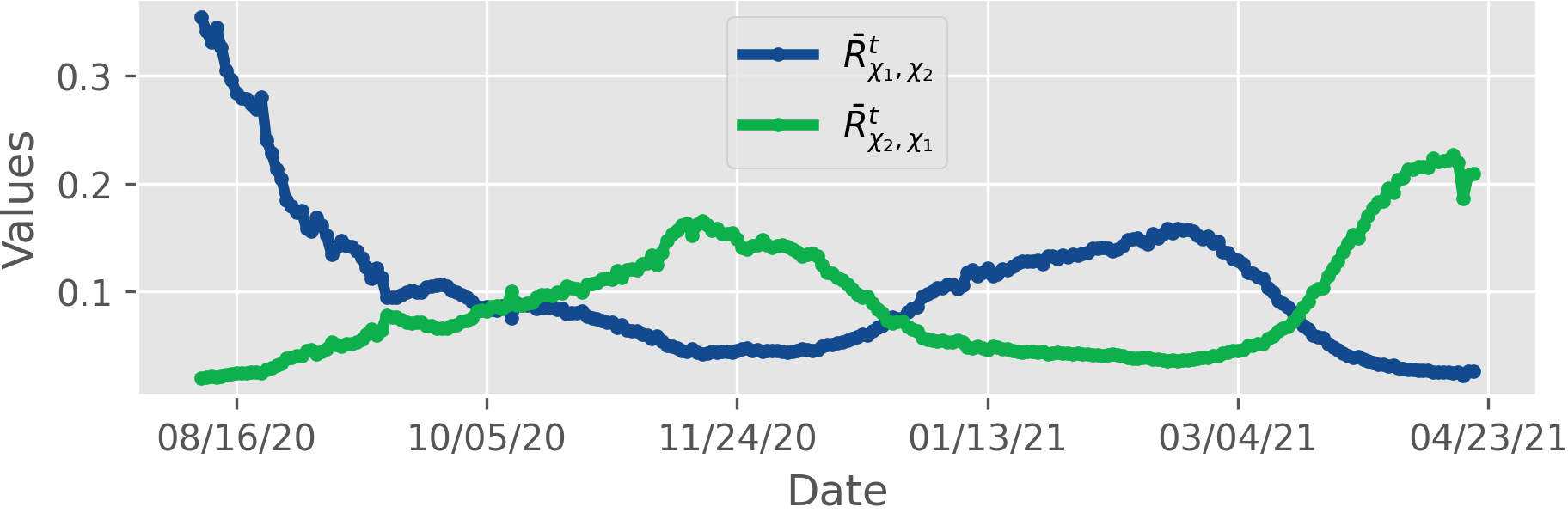}
    \caption{Cluster distributed ERNs between region $\chi_1$ (Detroit area) and $\chi_2$ (Miami area) over time.}
\label{fig:dist_rep_num_DD}
\end{figure}

\begin{figure}[H]
    \centering \includegraphics[width=1\columnwidth]{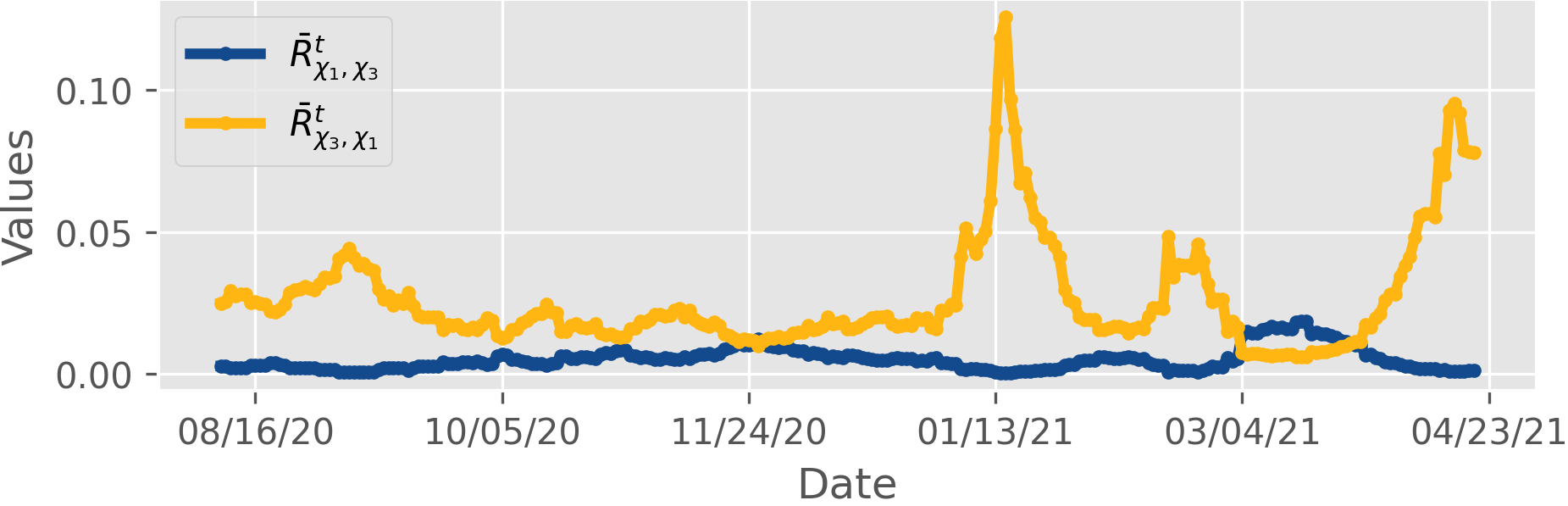}
    \caption{Cluster distributed ERNs between regions $\chi_1$ (Detroit area) and $\chi_3$ (Delta Junction) over time.}
\label{fig:dist_rep_num_DM}
\end{figure}

The green line in Figure~\ref{fig:dist_rep_num_DD} illustrates the \textcolor{black}{cluster} distributed ERNs from $\chi_1$ to  $\chi_2$, given by $\bar{R}^t_{\chi_2,\chi_1}$. It shows two main spikes: the first spike occurs in mid-November 2020, and the second at the end of April 2021. These spikes coincide with increases in the infected proportion of $\chi_2$, as seen by the green line in Figure~\ref{fig:infected_portion}.
However, when the infected proportion of $\chi_2$ reaches its peak in mid-January 2021,
$\bar{R}^t_{\chi_2,\chi_1}$ drops to a relatively low value. This phenomenon can be explained by the definition of $\bar{R}^t_{\chi_q,\chi_r}$  in~\eqref{Eq_Distributed_R_Clusters}.  Recall that $\bar{R}^t_{\chi_q,\chi_r}$ is not only influenced by the infected proportion within a cluster but also by the local distributed ERNs from entities outside the cluster to those within it, i.e., $\bar{R}^t_{ij}$ for $i\in {\chi_q}$ and $j\in {\chi_r}$. Furthermore, $\bar{R}^t_{ij}$ is determined by the ratio of the infected proportions from entity  $j$ to entity $i$ for  $i\in {\chi_q}$ and $j\in {\chi_r}$. 

\textcolor{black}{Hence, we conclude that these spikes in $\bar{R}^t_{\chi_2,\chi_1}$ are due to the relative changes in the infected proportions within these two regions. Specifically, the increase in the infected population in the Miami region and the decrease in the infected population in Delta Junction, Alaska,  generate these spikes.}
\textcolor{black}{When the infected population in $\chi_1$ spikes
{\textemdash}
first in mid-November 2020 and again at the end of April 2021
{\textemdash}
the ratio is large, leading to spikes in $\bar{R}^t_{\chi_2,\chi_1}$ in 
Figure~\ref{fig:dist_rep_num_DD}. Conversely, when the infected population of $\chi_1$ is lower, such as in March 2021, $\bar{R}^t_{\chi_2,\chi_1}$ is also low. 
The same reasoning can be applied to analyze the shape of the blue line representing $\bar{R}^t_{\chi_1,\chi_2}$ in Figure~\ref{fig:dist_rep_num_DD}.}

Figure~\ref{fig:dist_rep_num_DM} illustrates the cluster distributed ERNs between the Detroit area $\chi_1$ and the Delta Junction area $\chi_3$. For $\bar{R}^t_{\chi_3, \chi_1}$, we observe three main spikes starting from mid-January 2021.  Referring to Figure~\ref{fig:infected_portion}, we conclude that these spikes occur when the infected proportion in $\chi_3$ is decreasing and/or the infected proportion in $\chi_1$ is increasing or reaching a local peak. In contrast,
$\bar{R}^t_{\chi_1, \chi_3}$ shows a more stable trend, with only a mild spike occurring in mid-March.  At that time, Figure~\ref{fig:infected_portion} indicates that the infected proportion in $\chi_3$ surpasses the infected proportion in $\chi_1$.  This observation highlights the influence of infected proportions from $\chi_1$ to $\chi_3$ on computing $\bar{R}^t_{\chi_3, \chi_1}$, particularly when analyzing the interactions between a large population area and a region with a significantly smaller population.

From this analysis, we observe that, unlike the network-level ERN, the cluster distributed ERNs can infer not only infection dynamics within the cluster itself but also potential outbreaks or spreading behavior in other connected communities, thereby capturing causal relationships.

\subsection{Private Cluster Distributed ERNs and Accuracy}

In this section, we consider the frequency of interactions between regions, as introduced in Section~\ref{sec_data_pre}, to be sensitive information. When sharing local or cluster distributed ERNs with the public or higher authorities, these distributed reproduction numbers could reveal the frequency of interactions between regions over time. As a result, the population flow between these regions in the United States from August 9, 2020 to April 20, 2021 could be inferred.
Thus, we follow the same procedure as in Section~\ref{Sec_clus_dis_ERNs}, i.e., the steps in Algorithm~\ref{AL_1} with all privacy mechanisms implemented, to generate the private cluster distributed ERNs for the Detroit area, the Miami area, and the Delta Junction area. 

We present the root mean squared error (RMSE) over
time in Figure~\ref{fig:rmse} and  
the private cluster distributed ERNs between these clusters in Figure~\ref{fig:dp_distributed_rep_num}. 
Figures~\ref{fig:Detroit_to_Miami}-\ref{fig:Delta_to_Detroit}  illustrate the comparisons between the cluster distributed ERNs and their corresponding private cluster distributed ERNs, represented by the mean and 
points that are within 
one standard deviation of the mean. Figures~\ref{fig:Detroit_to_Miami}–\ref{fig:Delta_to_Detroit} show $\tilde {R}^t_{\chi_2,\chi_1}$, $\tilde {R}^t_{\chi_1,\chi_2}$, $\tilde {R}^t_{\chi_3,\chi_1}$, and $\tilde {R}^t_{\chi_1,\chi_3}$, respectively.  

\begin{figure}[H]
    \centering \includegraphics[width=1\columnwidth]{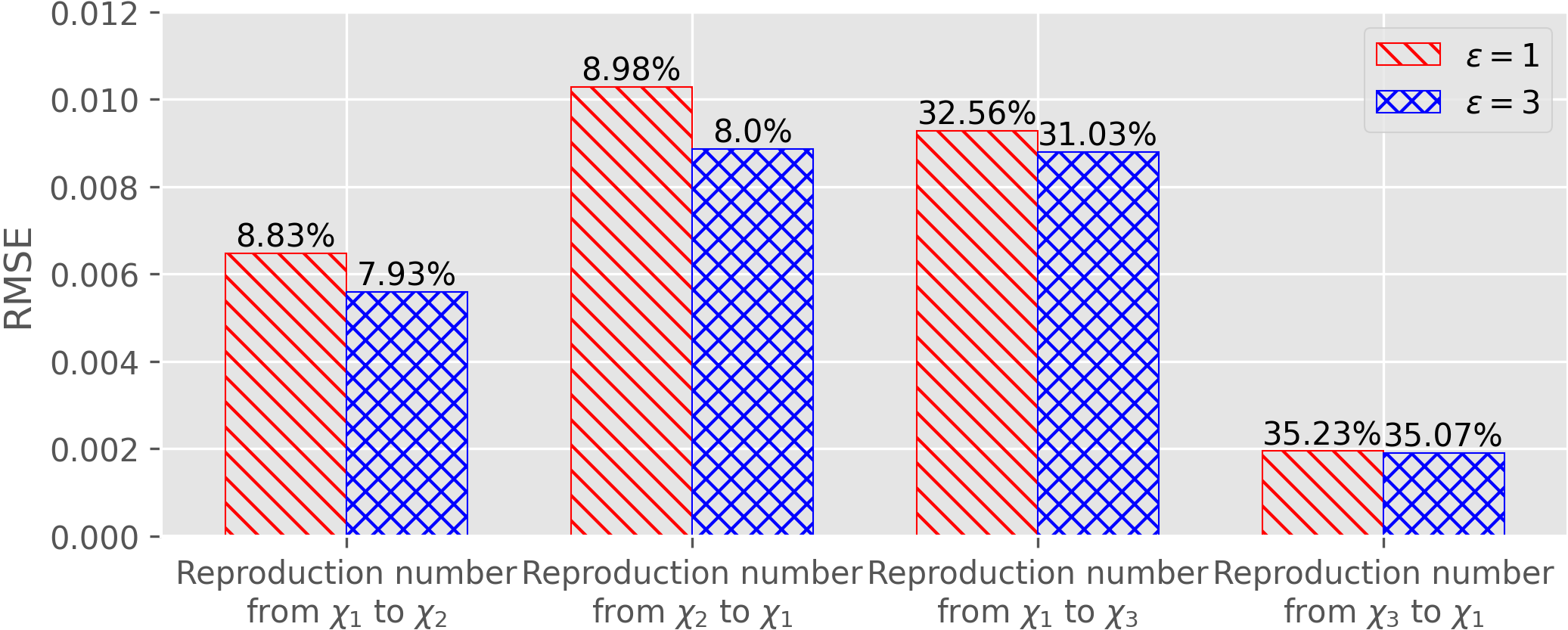}
    \caption{The root mean squared error (RMSE) of the private cluster distributed ERNs  under different privacy levels over time. The average percentage errors are also displayed on top of the bars.
    As~$\epsilon$ decreases from~$3$ to~$1$, privacy strengthens, with only slight increases in the percentage error induced by the privacy mechanism. This result indicates that it is possible to achieve both strong differential privacy and accurate computations of the cluster distributed ERNs simultaneously. 
    }
\label{fig:rmse}
\end{figure}

Figure~\ref{fig:dp_distributed_rep_num} illustrates the accuracy of the private \textcolor{black}{cluster distributed} ERNs.
Specifically, we plot the empirical mean and standard deviation of 100 differentially private samples with privacy level $\epsilon=1$ and adjacency parameter~$k=10^{-5}$ for all $i$. 
The value of this adjacency parameter is chosen by the maximum variation in 
the distributed ERNs that a single mobile data point can cause when it changes by its maximum possible amount.

\begin{figure*}
    \begin{subfigure}[b]{0.5\linewidth}
        \includegraphics[width=\linewidth]{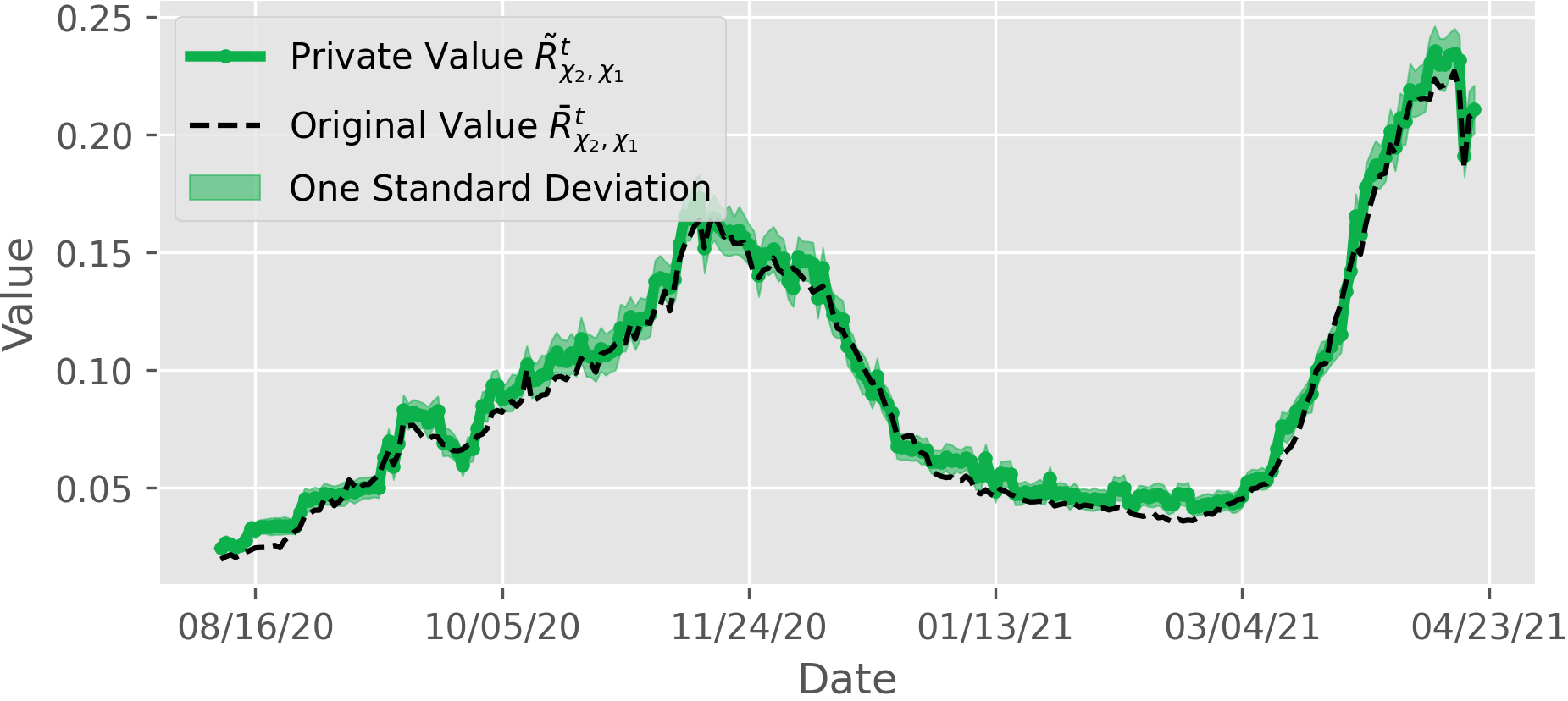}
        \caption{The original value of the cluster distributed ERNs from the region $\chi_1$ (Detroit region) to the region $\chi_2$ (Miami region)  is shown in dashed line. The mean and area within a standard deviation of private values are shown in green.}
        \label{fig:Detroit_to_Miami}
    \end{subfigure}
    \hfill
    \begin{subfigure}[b]{0.5\linewidth}
        \includegraphics[width=\linewidth]{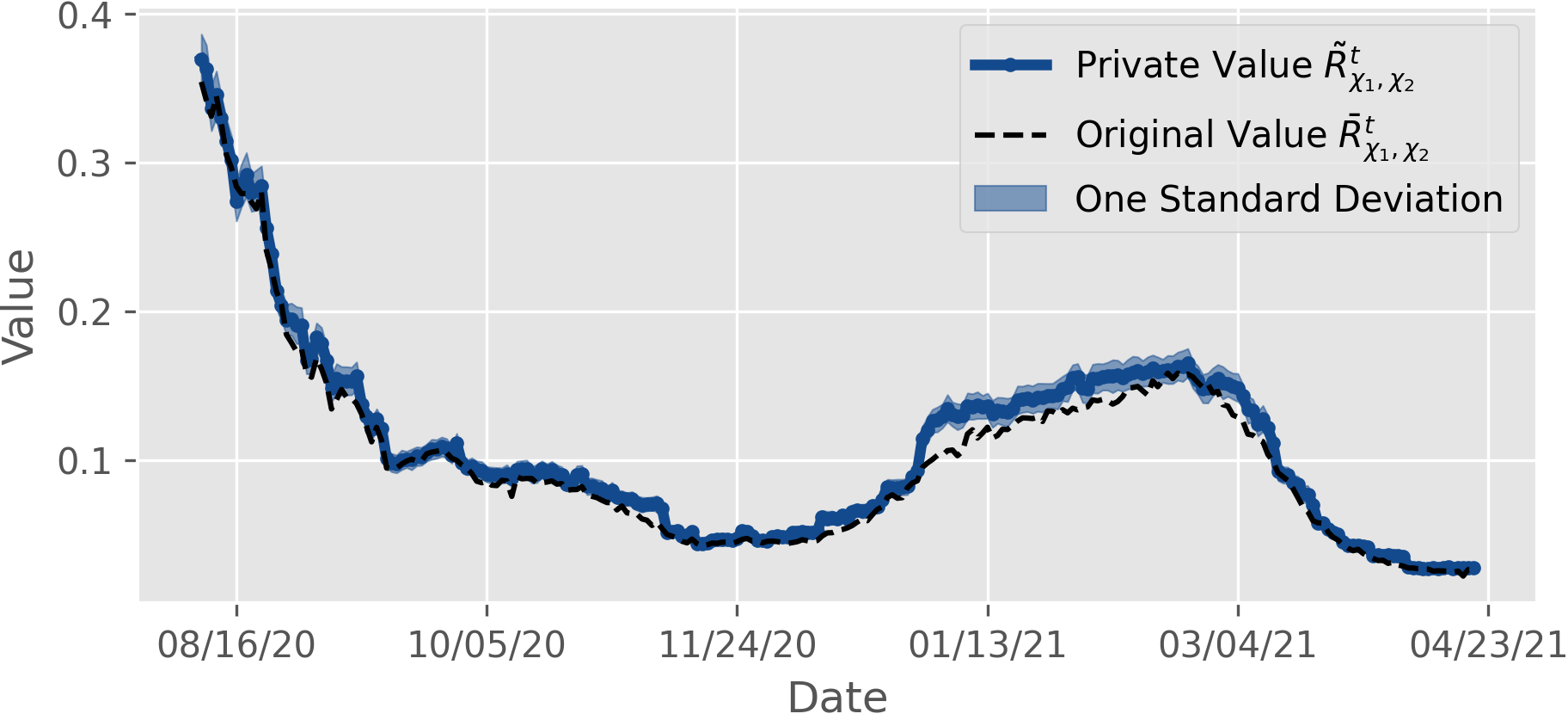}
        \caption{The original value of the cluster distributed ERNs from the region $\chi_2$ (Miami region) to the region $\chi_1$ (Detroit region)  is shown in dashed line. The mean and area within a standard deviation of private values are shown in blue.}
        \label{fig:Miami_to_Detroit}
    \end{subfigure}
    
    \begin{subfigure}[b]{0.5\linewidth}
        \includegraphics[width=\linewidth]{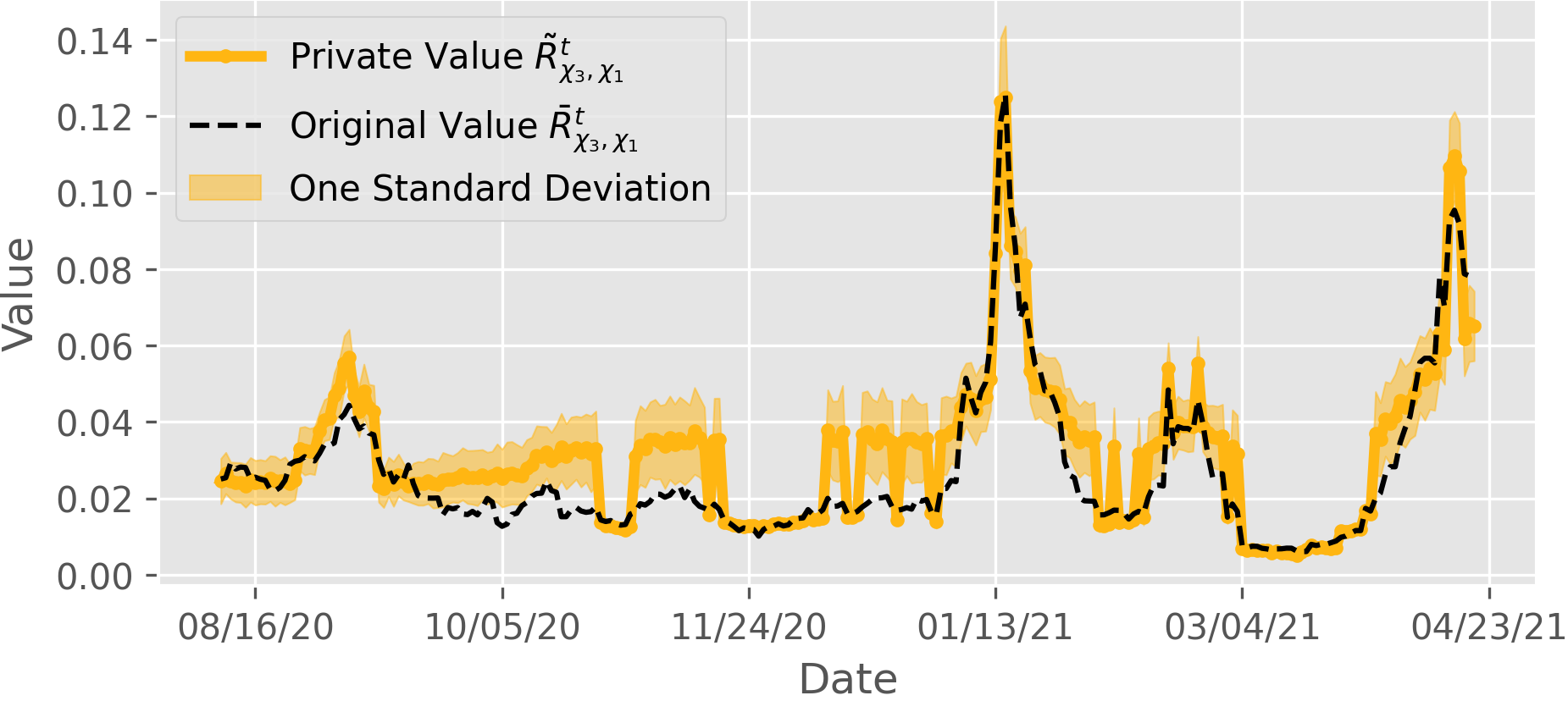}
        \caption{The original value of the cluster distributed ERNs from the region $\chi_1$ (Detroit region) to the region $\chi_3$ (Delta Junction region) is shown in dashed line. The mean and area within a standard deviation of private values are shown in gold.}
        \label{fig:Detroit_to_Delta}
    \end{subfigure}
      \hfill
    \begin{subfigure}[b]{0.5\linewidth}
        \includegraphics[width=\linewidth]{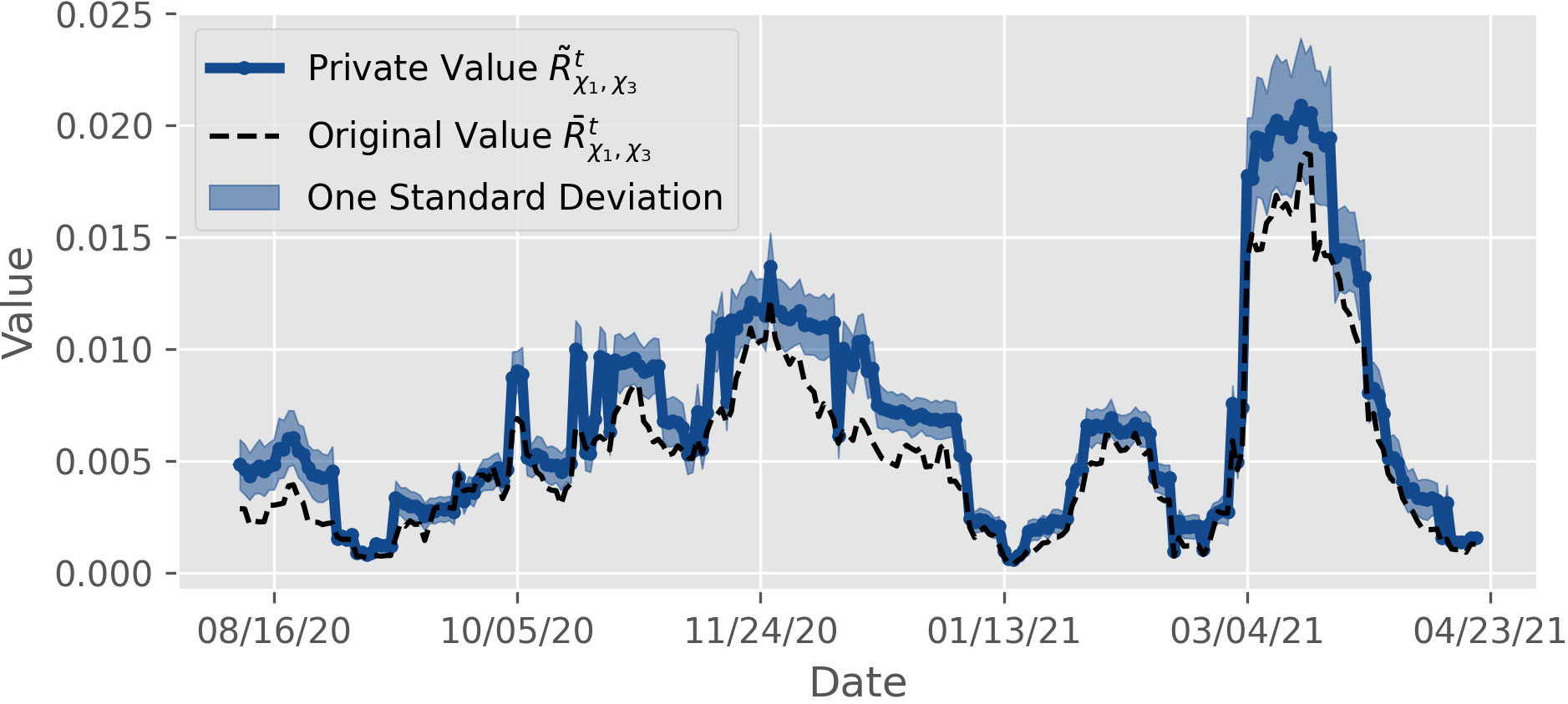}
        \caption{The original value of the cluster distributed ERNs from the region $\chi_3$ (Delta Junction region) to the region $\chi_1$ (Detroit region) is shown in dashed line. The mean and area within a standard deviation of private values are shown in blue.}
        \label{fig:Delta_to_Detroit}
    \end{subfigure}
\caption{Original values of the cluster distributed ERNs and their private values with~$\epsilon=1$.
The private values tend to be mildly conservative
in the sense that they are often larger than their non-private counterparts,
though they are generally close and provide accurate estimates of them. 
}
\label{fig:dp_distributed_rep_num}
\end{figure*}

Figure~\ref{fig:rmse} shows the root mean squared error (RMSE) over time. In Figure~\ref{fig:rmse}, the private cluster-distributed ERNs with a strong differential privacy guarantee ($\epsilon=1$) incur an error percentage ranging from $8.83\%$ to $5.23\%$. 
Specifically, we observe that the magnitudes of errors across all four private cluster distributed ERNs are on the same scale.
However, since the magnitudes of cluster distributed ERNs between $\chi_1$ and $\chi_3$ are significantly smaller than those between $\chi_1$ and $\chi_2$, the percentage error in the private cluster distributed ERNs between $\chi_1$ and $\chi_2$ is higher. 
On the other hand, in Figures~\ref{fig:Detroit_to_Miami}–\ref{fig:Delta_to_Detroit}, we observe that the private cluster distributed ERNs tend to be higher than their non-private counterparts. This observation suggests that estimates using private cluster distributed ERNs may slightly overestimate the severity of the spread. This result is caused by the truncated Gaussian mechanism, as discussed in Remark~\ref{Remark_T_Gau}.
Additionally, we observe that the private values are concentrated around their empirical averages. These observations support the effectiveness of using private cluster distributed ERNs.

Recall that the parameter~$\epsilon$ controls the strength of 
differential privacy's protections. Lower values of  $\epsilon$ correspond to stronger privacy levels, typically achieved by adding more noise to the effective reproduction numbers. In Figure~\ref{fig:rmse}, we observe that the accuracy of the private cluster ERNs remains consistently high regardless of the value of~$\epsilon$. This observation demonstrates that the cluster distributed ERNs with differential privacy are generally robust to the level of privacy enforced.

\section{Conclusion}
This paper developed methods to develop reproduction numbers for
epidemics at varying resolutions. It was shown that these new reproduction numbers can effectively 
give insight into the spread of an epidemic across different regions, and a differential privacy framework was developed
to protect sensitive data, such as individuals' travel patters, when computing them.
These developments were validated through real-world spreading processes, highlighting the utility of distributed reproduction numbers and the balance between model error and privacy strength.
In future work, we aim to leverage local and cluster distributed reproduction numbers to design pandemic control algorithms. Specifically, we will study how changes in local distributed reproduction numbers affect cluster distributed reproduction numbers at different scales, forming a hierarchical control framework that uses reproduction numbers to model and control the spreading network. Furthermore, we will develop a robust control framework that accounts for model errors introduced by the privacy mechanism designed in this work.
\label{sec_Conclusion}
\bibliographystyle{IEEEtran}
\bibliography{main}

\begin{thebibliography}{10}
\providecommand{\url}[1]{#1}
\csname url@samestyle\endcsname
\providecommand{\newblock}{\relax}
\providecommand{\bibinfo}[2]{#2}
\providecommand{\BIBentrySTDinterwordspacing}{\spaceskip=0pt\relax}
\providecommand{\BIBentryALTinterwordstretchfactor}{4}
\providecommand{\BIBentryALTinterwordspacing}{\spaceskip=\fontdimen2\font plus
\BIBentryALTinterwordstretchfactor\fontdimen3\font minus \fontdimen4\font\relax}
\providecommand{\BIBforeignlanguage}[2]{{%
\expandafter\ifx\csname l@#1\endcsname\relax
\typeout{** WARNING: IEEEtran.bst: No hyphenation pattern has been}%
\typeout{** loaded for the language `#1'. Using the pattern for}%
\typeout{** the default language instead.}%
\else
\language=\csname l@#1\endcsname
\fi
#2}}
\providecommand{\BIBdecl}{\relax}
\BIBdecl

\bibitem{van2017reproduction}
P.~van~den Driessche, ``Reproduction numbers of infectious disease models,'' \emph{Infec. Dise. Model.}, vol.~2, no.~3, pp. 288--303, 2017.

\bibitem{soltesz2020effect}
K.~Soltesz, F.~Gustafsson, T.~Timpka, J.~Jald{\'e}n, C.~Jidling, A.~Heimerson, T.~B. Sch{\"o}n, A.~Spreco, J.~Ekberg, {\"O}.~Dahlstr{\"o}m \emph{et~al.}, ``The effect of interventions on {COVID-19},'' \emph{Nature}, vol. 588, no. 7839, pp. E26--E28, 2020.

\bibitem{she2024framework}
B.~She, R.~L. Smith, I.~Pytlarz, S.~Sundaram, and P.~E. Par{\'e}, ``A framework for counterfactual analysis, strategy evaluation, and control of epidemics using reproduction number estimates,'' \emph{PLoS Comp. Bio.}, vol.~20, no.~11, p. e1012569, 2024.

\bibitem{she2021network}
B.~She, J.~Liu, S.~Sundaram, and P.~E. Par{\'e}, ``On a networked {$SIS$} epidemic model with cooperative and antagonistic opinion dynamics,'' \emph{IEEE Trans, on Contr. of Netw. Syst.}, vol.~9, pp. 1154 -- 1165, 2022.

\bibitem{pascal2022nonsmooth}
B.~Pascal, P.~Abry, N.~Pustelnik, S.~Roux, R.~Gribonval, and P.~Flandrin, ``Nonsmooth convex optimization to estimate the {Covid-19} reproduction number space-time evolution with robustness against low quality data,'' \emph{IEEE Trans. on Sig. Proc.}, vol.~70, pp. 2859--2868, 2022.

\bibitem{smith2021convex}
K.~D. Smith and F.~Bullo, ``Convex optimization of the basic reproduction number,'' \emph{IEEE Trans. on Autom. Contr.}, vol.~68, no.~7, pp. 4398--4404, 2022.

\bibitem{mei2017epidemics_review}
W.~Mei, S.~Mohagheghi, S.~Zampieri, and F.~Bullo, ``On the dynamics of deterministic epidemic propagation over networks,'' \emph{Annu. Rev. in Control}, vol.~44, pp. 116--128, 2017.

\bibitem{pare2020modeling_review}
P.~E. Par{\'e}, C.~L. Beck, and T.~Ba{\c{s}}ar, ``Modeling, estimation, and analysis of epidemics over networks: An overview,'' \emph{Annu. Rev. in Control}, vol.~50, pp. 345--360, 2020.

\bibitem{zino2021analysis}
L.~Zino and M.~Cao, ``Analysis, prediction, and control of epidemics: A survey from scalar to dynamic network models,'' \emph{IEEE Cir. and Syst. Mag.}, vol.~21, no.~4, pp. 4--23, 2021.

\bibitem{nowzari2016epidemics}
C.~Nowzari, V.~M. Preciado, and G.~J. Pappas, ``{Analysis and control of epidemics: A survey of spreading processes on complex networks},'' \emph{IEEE Control Syst. Magazine}, vol.~36, no.~1, pp. 26--46, 2016.

\bibitem{van2011n}
P.~Van~Mieghem, ``The {N-intertwined SIS} epidemic network model,'' \emph{Computing}, vol.~93, no.~2, pp. 147--169, 2011.

\bibitem{bivirus}
J.~Liu, P.~E. Par\'e, A.~Nedi\'c, C.~Tang, C.~Beck, and T.~Ba\c{s}ar, ``Analysis and control of a continuous-time bi-virus model,'' \emph{IEEE Trans. Autom. Control}, vol.~64, no.~12, pp. 4891--4906, 2019.

\bibitem{ihme2021modeling}
I.~C.-.~F. Team, ``Modeling {COVID-19} scenarios for the {United States},'' \emph{Nature Medicine}, vol.~27, no.~1, pp. 94--105, 2021.

\bibitem{arino2003multi_city_SIS}
J.~Arino and P.~Van~den Driessche, ``A multi-city epidemic model,'' \emph{Mathematical Population Studies}, vol.~10, no.~3, pp. 175--193, 2003.

\bibitem{eames2010assessing}
K.~T. Eames, C.~Webb, K.~Thomas, J.~Smith, R.~Salmon, and J.~M.~F. Temple, ``Assessing the role of contact tracing in a suspected {H7N2} influenza a outbreak in humans in wales,'' \emph{BMC infectious diseases}, vol.~10, pp. 1--6, 2010.

\bibitem{le2022high}
G.~Le~Treut, G.~Huber, M.~Kamb, K.~Kawagoe, A.~McGeever, J.~Miller, R.~Pnini, B.~Veytsman, and D.~Yllanes, ``A high-resolution flux-matrix model describes the spread of diseases in a spatial network and the effect of mitigation strategies,'' \emph{Sci. Rep.}, vol.~12, no.~1, p. 15946, 2022.

\bibitem{balcan2009multiscale}
D.~Balcan, V.~Colizza, B.~Gon{\c{c}}alves, H.~Hu, J.~J. Ramasco, and A.~Vespignani, ``Multiscale mobility networks and the spatial spreading of infectious diseases,'' \emph{Proceedings of the National Academy of Sciences}, vol. 106, no.~51, pp. 21\,484--21\,489, 2009.

\bibitem{imola2021locally}
J.~Imola, T.~Murakami, and K.~Chaudhuri, ``Locally differentially private analysis of graph statistics,'' in \emph{Proc. 30th USENIX security symposium (USENIX Security 21)}, 2021, pp. 983--1000.

\bibitem{Karwa2014Private}
V.~Karwa, S.~Raskhodnikova, A.~Smith, and G.~Yaroslavtsev, ``Private analysis of graph structure,'' \emph{ACM Trans. Database Syst.}, vol.~39, no.~3, oct 2014.

\bibitem{Day2016Publishing}
W.-Y. Day, N.~Li, and M.~Lyu, ``Publishing graph degree distribution with node differential privacy,'' in \emph{Proc. 2016 Int. Conf. on Mana of Data}, ser. SIGMOD '16.\hskip 1em plus 0.5em minus 0.4em\relax New York, NY, USA: Association for Computing Machinery, 2016, p. 123–138.

\bibitem{ZHANG2021Differentially}
S.~Zhang, W.~Ni, and N.~Fu, ``Differentially private graph publishing with degree distribution preservation,'' \emph{Computers \& Security}, vol. 106, p. 102285, 2021.

\bibitem{chen2021edge}
B.~Chen, C.~Hawkins, K.~Yazdani, and M.~Hale, ``Edge differential privacy for algebraic connectivity of graphs,'' in \emph{Proc. of the 60th IEEE Conf. on Dec. and Cont. (CDC)}.\hskip 1em plus 0.5em minus 0.4em\relax IEEE, 2021, pp. 2764--2769.

\bibitem{dwork2014algorithmic}
C.~Dwork, A.~Roth \emph{et~al.}, ``The algorithmic foundations of differential privacy,'' \emph{Foundations and Trends{\textregistered} in Theoretical Computer Science}, vol.~9, no. 3--4, pp. 211--407, 2014.

\bibitem{cortes2016differential}
J.~Cort{\'e}s, G.~E. Dullerud, S.~Han, J.~Le~Ny, S.~Mitra, and G.~J. Pappas, ``Differential privacy in control and network systems,'' in \emph{Proc. of the 2016 IEEE Conf. on Deci. and Contr. (CDC)}.\hskip 1em plus 0.5em minus 0.4em\relax IEEE, 2016, pp. 4252--4272.

\bibitem{hale19}
M.~Hale, P.~Barooah, K.~Parker, and K.~Yazdani, ``Differentially private smart metering: Implementation, analytics, and billing,'' in \emph{Proceedings of the 1st ACM International Workshop on Urban Building Energy Sensing, Controls, Big Data Analysis, and Visualization}, ser. UrbSys'19, 2019, p. 33–42.

\bibitem{hawkins20}
C.~Hawkins and M.~Hale, ``Differentially private formation control,'' in \emph{in Proc. of the 59th IEEE Conf. on Deci. and Cont. (CDC)}, 2020, pp. 6260--6265.

\bibitem{hawkins23}
C.~Hawkins, B.~Chen, K.~Yazdani, and M.~Hale, ``Node and edge differential privacy for graph laplacian spectra: Mechanisms and scaling laws,'' \emph{IEEE Transactions on Network Science and Engineering}, 2023.

\bibitem{she2023distributed}
B.~She, P.~E. Par{\'e}, and M.~Hale, ``Distributed reproduction numbers of networked epidemics,'' in \emph{Proc. of the 2023 Amer. Contr. Conf. (ACC)}.\hskip 1em plus 0.5em minus 0.4em\relax IEEE, 2023, pp. 4302--4307.

\bibitem{chen2023differentially}
B.~Chen, B.~She, C.~Hawkins, A.~Benvenuti, B.~Fallin, P.~E. Par{\'e}, and M.~Hale, ``Differentially private computation of basic reproduction numbers in networked epidemic models,'' in \emph{Proc. of the 2024 Amer. Contr. Conf. (ACC)}.\hskip 1em plus 0.5em minus 0.4em\relax IEEE, 2024, pp. 4422--4427.

\bibitem{pare2020modeling}
P.~E. Par{\'e}, C.~L. Beck, and T.~Ba{\c{s}}ar, ``Modeling, estimation, and analysis of epidemics over networks: An overview,'' \emph{Annual Reviews in Control}, vol.~50, pp. 345--360, 2020.

\bibitem{diekmann2010construction}
O.~Diekmann, J.~Heesterbeek, and M.~G. Roberts, ``The construction of next-generation matrices for compartmental epidemic models,'' \emph{J. of the Roy. Soci. Inter.}, vol.~7, no.~47, pp. 873--885, 2010.

\bibitem{she2021peak}
B.~She, H.~C. Leung, S.~Sundaram, and P.~E. Par{\'e}, ``Peak infection time for a networked {SIR} epidemic with opinion dynamics,'' in \emph{In Proc. 60th IEEE Conf. on Dec. and Contr.}\hskip 1em plus 0.5em minus 0.4em\relax IEEE, 2021, pp. 2104--2109.

\bibitem{Sealfon2016shortest}
A.~Sealfon, ``Shortest paths and distances with differential privacy,'' in \emph{Proc. of the 35th ACM SIGMOD-SIGACT-SIGAI Symp. on Prin. of Data. Syst.}, 2016, pp. 29--41.

\bibitem{Hay2009accurate}
M.~Hay, C.~Li, G.~Miklau, and D.~Jensen, ``Accurate estimation of the degree distribution of private networks,'' in \emph{Proc. 2009 Ninth IEEE Int. Conf. on Data Mining}, 2009, pp. 169--178.

\bibitem{blocki2013differentially}
J.~Blocki, A.~Blum, A.~Datta, and O.~Sheffet, ``Differentially private data analysis of social networks via restricted sensitivity,'' \emph{arXiv preprint arXiv:1208.4586}, 2013.

\bibitem{stehle2011high}
J.~Stehl{\'e}, N.~Voirin, A.~Barrat, C.~Cattuto, L.~Isella, J.-F. Pinton, M.~Quaggiotto, W.~Van~den Broeck, C.~R{\'e}gis, B.~Lina \emph{et~al.}, ``High-resolution measurements of face-to-face contact patterns in a primary school,'' \emph{PloS One}, vol.~6, no.~8, p. e23176, 2011.

\bibitem{Kasiviswanathan2008what}
S.~P. Kasiviswanathan, H.~K. Lee, K.~Nissim, S.~Raskhodnikova, and A.~Smith, ``What can we learn privately?'' in \emph{Proc. of the 49th Annual IEEE Symp. on Found. of Comp. Sci.}, 2008, pp. 531--540.

\bibitem{hsu2014differential}
J.~Hsu, M.~Gaboardi, A.~Haeberlen, S.~Khanna, A.~Narayan, B.~C. Pierce, and A.~Roth, ``Differential privacy: {An} economic method for choosing epsilon,'' in \emph{Proc. of the 2014 IEEE 27th Comp. Sec. Found. Symp.}\hskip 1em plus 0.5em minus 0.4em\relax IEEE, 2014, pp. 398--410.

\bibitem{burkardt2014truncated}
J.~Burkardt, ``The truncated normal distribution,'' \emph{Department of Scientific Computing Website, Florida State University}, vol.~1, p.~35, 2014.

\bibitem{mei2017dynamics}
W.~Mei, S.~Mohagheghi, S.~Zampieri, and F.~Bullo, ``On the dynamics of deterministic epidemic propagation over networks,'' \emph{Ann. Rev. in Cont.}, vol.~44, pp. 116--128, 2017.

\bibitem{varga2009matrix_book}
R.~S. Varga, \emph{{Matrix Iterative Analysis}}.\hskip 1em plus 0.5em minus 0.4em\relax Springer Science \& Business Media, 2009, vol.~27.

\bibitem{chen2022bounded}
B.~Chen and M.~Hale, ``The bounded {Gaussian} mechanism for differential privacy,'' \emph{Journal of Privacy and Confidentiality}, vol. 14(1), 2024.

\bibitem{Cheu2019Distributed}
A.~Cheu, A.~Smith, J.~Ullman, D.~Zeber, and M.~Zhilyaev, ``Distributed differential privacy via shuffling,'' in \emph{Advances in Cryptology -- EUROCRYPT 2019}, Y.~Ishai and V.~Rijmen, Eds.\hskip 1em plus 0.5em minus 0.4em\relax Cham: Springer International Publishing, 2019, pp. 375--403.

\bibitem{safegraph2021Distancing}
{Safe Graph}, ``{SafeGraph Social Distancing Metrics},'' 2021, available at \url{https://docs.safegraph.com/ docs/social-distancing-metrics}.

\bibitem{covid192020Dong}
E.~Dong, H.~Du, and L.~Gardner, ``{An interactive web-based dashboard to track COVID-19 in real time},'' \emph{The Lancet Infectious Diseases}, vol.~20, no.~5, pp. 533--534, 2020.

\bibitem{safegraph2020Census}
{Safe Graph}, ``{SafeGraph Open Census Data },'' 2020, available at \url{https://docs.safegraph.com/docs/open-census-data}.

\bibitem{caicedo2020effective}
Y.~Caicedo-Ochoa, D.~E. Rebell{\'o}n-S{\'a}nchez, M.~Pe{\~n}aloza-Rall{\'o}n, H.~F. Cort{\'e}s-Motta, and Y.~R. M{\'e}ndez-Fandi{\~n}o, ``{Effective Reproductive Number estimation for initial stage of COVID-19 pandemic in Latin American Countries},'' \emph{International Journal of Infectious Diseases}, vol.~95, pp. 316--318, 2020.

\end{thebibliography}
\end{document}